%% file: main.tex
\documentclass[twoside]{article}

\input{math_commands.tex}

\usepackage{microtype}
\usepackage{graphicx}
\usepackage{subcaption}
\usepackage{booktabs} 
\usepackage{wrapfig}
\usepackage[utf8]{inputenc} 
\usepackage[T1]{fontenc}    
\usepackage{hyperref}
\usepackage{url}            
\usepackage{booktabs}       
\usepackage{amsfonts}       
\usepackage{nicefrac}       
\usepackage{microtype}      
\usepackage{xcolor}         
\usepackage{pgf,tikz}
\usetikzlibrary{babel}
\usepackage{circuitikz}
\usepackage{siunitx}
\usepackage{bbding}
\usepackage{enumitem}
\usepackage{algorithm}
\usepackage{algorithmic}
\usepackage{multicol}
\usepackage[normalem]{ulem}

\usepackage{amsmath}
\usepackage{amssymb}
\usepackage{mathtools}
\usepackage{amsthm}
\usepackage[round]{natbib}
\usepackage[capitalize,noabbrev]{cleveref}

\theoremstyle{plain}
\newtheorem{theorem}{Theorem}[section]

\newtheorem{lemma}[theorem]{Lemma}
\newtheorem{corollary}[theorem]{Corollary}
\newtheorem{principle}[theorem]{Principle}
\theoremstyle{definition}

\newtheorem{assumption}[theorem]{Assumption}
\theoremstyle{remark}
\newtheorem{remark}[theorem]{Remark}

\def\Rex{\mathcal{R}_{\mathrm{excess}}}

\graphicspath{ {figures/} }

\usepackage[textsize=tiny]{todonotes}

\usepackage[accepted]{aistats2026}
\usepackage{titlesec}
\usepackage{textcase}

\titleformat{\section}
  {\normalfont\bfseries}{\thesection}{1em}{\MakeTextUppercase}

\begin{document}

%

%
\runningauthor{Kuan-Ta Li, Chia-Chun Lin, Ping-Chun Hsieh, Yu-Chih Huang}

\twocolumn[

\aistatstitle{A Modularized Framework for Piecewise-Stationary Restless Bandits}

\aistatsauthor{Kuan-Ta Li \And Chia-Chun Lin}
\aistatsaddress{ Institute of Communications Engineering \\ National Yang Ming Chaio Tung University \\ Hsinchu, Taiwan \And Institute of Communications Engineering \\ National Yang Ming Chaio Tung University \\ Hsinchu, Taiwan}

\aistatsauthor{Ping-Chun Hsieh \And Yu-Chih Huang}

\aistatsaddress{Department of Computer Science \\ National Yang Ming Chaio Tung University \\ Hsinchu, Taiwan \And Institute of Communications Engineering \\ National Yang Ming Chaio Tung University \\ Hsinchu, Taiwan } ]

\input{0-abstract}
\input{1-intro_new}

\input{2-problem}
\input{3-algorithm}

\input{4-analysis}
\input{6-simulate}
\input{7-conclude}

\bibliography{reference.bib}

\clearpage
\appendix
\thispagestyle{empty}

\appendix

\onecolumn
\input{1-changealg}
\input{2-proof}

\input{3-parameters_RMAB}
\input{4-one-state_simulation.tex}
\input{5-parameter.tex}

\end{document}

%% file: math_commands.tex
\usepackage{amsmath,amsfonts,bm}

\def\eqref#1{equation~\ref{#1}}
\def\Eqref#1{Equation~\ref{#1}}

\def\Algref#1{Algorithm~\ref{#1}}

\def\1{\bm{1}}

\DeclareMathAlphabet{\mathsfit}{\encodingdefault}{\sfdefault}{m}{sl}
\SetMathAlphabet{\mathsfit}{bold}{\encodingdefault}{\sfdefault}{bx}{n}

\def\gK{{\mathcal{K}}}

\def\gO{{\mathcal{O}}}

\def\gR{{\mathcal{R}}}

\newcommand{\E}{\mathbb{E}}

\DeclareMathOperator*{\argmax}{arg\,max}

%% file: 0-abstract.tex
\begin{abstract}
We study the piecewise-stationary restless multi-armed bandit (PS-RMAB) problem, 
where each arm evolves as a Markov chain but \emph{mean rewards may change across unknown segments}. 
To address the resulting exploration--detection delay trade-off, we propose a modular framework that integrates 
arbitrary RMAB base algorithms with change detection and a novel diminishing exploration mechanism. 
This design enables flexible plug-and-play use of existing solvers and detectors, while efficiently adapting 
to mean changes without prior knowledge of their number. 

To evaluate performance, we introduce a refined regret notion that measures the 
\emph{excess regret due to exploration and detection}, benchmarked against an oracle that restarts 
the base algorithm at the true change points. 
Under this metric, we prove a regret bound of 
$\tilde{O}(\sqrt{LMKT})$, where $L$ denotes the maximum mixing time of the Markov chains across all arms and segments, $M$ the number of segments, 
$K$ the number of arms, and $T$ the horizon. 
Simulations confirm that our framework achieves regret close to that of the segment oracle and consistently outperforms base solvers that do not incorporate any mechanism to handle environmental changes. 
\end{abstract}

%% file: 1-intro_new.tex
\section{Introduction}
\label{sec:intro}
\vspace{-5pt}
Restless multi-armed bandits (RMABs) provide a powerful abstraction for sequential decision-making under resource constraints. In contrast to the classical stochastic multi-armed bandit (MAB), where unselected arms remain frozen, RMABs model each arm as an evolving Markov chain that continues to transition regardless of whether it is activated. This restless property enables RMABs to capture a wide range of dynamic allocation problems where opportunities or risks evolve continuously, even in the absence of intervention.

Since the seminal work of \citet{Whittle1988RestlessBA}, RMABs have been studied extensively in both theory and practice. Whittle introduced the Whittle index, derived from a Lagrangian relaxation of the original optimization problem, which provides a computationally tractable priority rule under the indexability condition. Although finding the exact optimal policy is PSPACE-hard \citep{PSPACE_hard}, this line of work has inspired a rich literature on both theoretical guarantees and practical algorithms.

Beyond its theoretical significance, the RMAB framework has proven highly impactful in practice, demonstrating utility across multiple domains. In communication systems, RMAB formulations underpin opportunistic scheduling tasks such as dynamic spectrum access in cognitive radio networks and downlink scheduling in wireless networks \citep{2008qingzhao}. The same framework has been applied in healthcare, where RMAB-based policies help optimize limited intervention resources (e.g., scheduling follow-up calls to patients) in public health outreach programs to maximize engagement and outcomes \citep{2022healthcare}. Likewise, modern recommender systems leverage RMAB models to adapt content delivery as user preferences evolve over time \citep{2017recommender}. Another important application arises in maintaining information freshness, where scheduling to minimize the \emph{Age of Information} (AoI) can be naturally cast as an RMAB problem \citep{hsu2018age}.

However, most existing studies assume stationary dynamics, which rarely hold in practice. Real-world environments often undergo gradual or abrupt changes, requiring policies that adapt over time. This motivates the study of RMABs in piecewise-stationary settings, where environments are segmented and the transition dynamics of arms shift at change points. These shifts are typically characterized by changes in the underlying Markov chains. To remain effective, the player’s policy must continually adjust to the new segments.

Classical piecewise-stationary bandit problems employ two broad strategies for adaptation: active and passive methods. Active approaches couple standard bandit algorithms with change detection modules, which trigger resets upon detecting environmental shifts, making them well-suited for abrupt changes. In contrast, passive methods exploit mechanisms such as sliding windows or discount factors that gradually down-weight outdated samples, allowing the algorithm to adapt smoothly without explicit resets.

While piecewise-stationary bandits have been extensively studied, extending these ideas to RMABs introduces unique challenges. In RMABs, different states of an arm can yield varying reward means, complicating the distinction between routine state transitions and genuine environmental changes. This variability raises the risk of false alarms, where a detection module may mistakenly interpret state-based fluctuations as true shifts in the environment.

Building on this motivation, our work makes the following contributions:

\textbf{Modular framework for PS-RMABs.} 
We introduce a modular algorithmic framework that integrates three key components: 
a stationary base solver, a change detection module, and a diminishing exploration schedule. 
This design ensures flexibility and plug-and-play adaptability, while isolating the two 
core sources of cost in piecewise-stationary environments: forced exploration and detection delay.

\textbf{Refined excess-regret notion.} 
We formalize a new notion of \emph{excess regret}, which explicitly captures the 
overhead caused by exploration and change detection, distinct from the intrinsic 
regret of the stationary base solver. Existing results for piecewise-stationary bandits 
show that the dominant terms in regret already arise from this exploration–delay balance, 
with the state-of-the-art scaling at $\tilde{O}(\sqrt{MKT})$. 
Our formulation makes this balance transparent and provides a principled way to analyze 
its extension to restless settings.

\textbf{Diminishing exploration with minimal assumptions and low complexity.} 
Our diminishing exploration (DE) requires no prior knowledge of the number of change points \(M\),
yet achieves a nearly optimal extra-regret bound of \(\tilde{\mathcal{O}}(\sqrt{MKT})\).
It adds virtually no computational overhead, relying only on simple scheduling to trigger exploration phases.

\textbf{Theoretical guarantees and general-case analysis.} 
We establish general excess-regret bounds that hold for arbitrary choices of base solvers 
and change detectors. These bounds confirm that the exploration–delay overhead is 
tightly controlled and does not exceed the known minimax rate in the piecewise-stationary 
setting. Consequently, whenever the stationary base solver is near-optimal, our 
framework preserves this near-optimality under piecewise-stationary extensions. 

Unlike passive adaptation methods such as sliding-window or discounted bandits, which continuously down-weight past observations, our approach actively coordinates exploration and change detection through a diminishing exploration schedule. This mechanism ensures sufficient statistical evidence for reliable detection while avoiding excessive exploration, and crucially eliminates the need for prior knowledge of the number of change points. Moreover, our framework naturally extends these ideas to restless settings, where passive forgetting mechanisms alone may fail due to state-dependent reward fluctuations.

A natural question is whether one can simply combine existing change-detection schemes with learning algorithms designed for stationary restless bandits. However, naive integration may lead to unreliable detection, excessive resets, or degraded regret performance, particularly in restless environments where state evolution can obscure mean shifts. Our framework identifies conditions under which such integration remains stable and preserves regret guarantees.
\vspace{-5pt}
\subsection{Related Work}
\label{sec:intro:related}
\vspace{-5pt}
\textbf{Restless Bandit Problem.} 
Restless bandits demand a range of algorithmic approaches. 
In planning settings with known models, index-based policies (Whittle’s and its variants) and LP-relaxation methods offer tractable solutions with guarantees like asymptotic optimality or constant-factor approximation \citep{Whittle1988RestlessBA, guha2010approximation}.  
In online learning settings, even under partial observability, the RMAB community has achieved sublinear regret. 
A notable early algorithm is Colored UCRL2 \citep{ortner2012regret}, which adapts the UCRL2 framework for restless bandits and provides $\tilde{\mathcal{O}}(\sqrt{T})$ regret under mild mixing assumptions, nearly matching the fundamental lower bound of $\Omega(\sqrt{ST})$ established in \citep{ortner2012regret}, where $S$ is the total number of states.  
More recent methods refine this paradigm using UCB-driven optimism—e.g., UCWhittle \citep{wang2023optimistic} and Restless-UCB \citep{wang2020restless}—to improve computational efficiency and adaptability. 
These developments supply both strong theoretical guarantees and broaden the practical applicability of RMAB solutions. These methods assume stationary environments and do not address change detection or adaptation to piecewise-stationary dynamics, which is the focus of our work.

Recent works have begun to address non-stationary RMABs under smoothly varying environments. For instance, \citet{shisher2025onlinelearningwhittleindices} propose a sliding-window Whittle index policy that adapts to time-varying transition kernels under a bounded variation budget, achieving sublinear dynamic regret. Similarly, \citet{hung2025nonstationaryrestlessmultiarmedbandits} develop an optimistic Whittle-based algorithm with arm-specific sliding windows and establish regret guarantees under a global variation budget constraint. These approaches model non-stationarity through gradual drift in transition dynamics and rely on continuous adaptation mechanisms such as sliding windows.

Another line of work considers more general non-stationarity under adversarial environments. For example, \citet{xiong2024provablyefficientreinforcementlearning} study adversarial RMABs with unknown transitions and bandit feedback, where rewards can change arbitrarily across episodes, and develop reinforcement learning algorithms with $\tilde{O}(\sqrt{T})$ regret guarantees. These settings capture worst-case non-stationarity and require fundamentally different techniques based on online learning and occupancy measures.

In contrast, our work focuses on piecewise-stationary environments with abrupt changes, where the main challenge lies in detecting change points and coordinating exploration with detection. Compared to variation-budget settings that assume gradual drift, and adversarial settings that allow arbitrary changes, our formulation imposes a structured non-stationarity model that enables sharper regret guarantees while still capturing realistic abrupt changes. 

\textbf{Piecewise-Stationary Bandits.} 
Existing approaches to PS-MAB can be broadly categorized into \emph{passive} and \emph{active} methods. 
Passive methods adapt to changes through forgetting mechanisms such as Discounted UCB \citep{kocsis2006discounted}, Sliding-Window UCB \citep{garivier2011upper}, and their Thompson Sampling counterparts \citep{raj2017taming,qi2023discounted}, all achieving $\gO(\sqrt{KMT}\,\mathrm{polylog}(T))$ regret. 
Active methods, in contrast, integrate change detection with standard bandit algorithms. 
Examples include CD-UCB \citep{liu2018change} and M-UCB \citep{cao2019nearly}, which combine UCB with detectors like CUSUM or simple window-based tests to obtain $\gO(\sqrt{KMT\log T})$ regret. 
More recent efforts, such as AdSwitch \citep{auer2019adaptively} and GLR-klUCB \citep{besson2022efficient}, avoid assuming prior knowledge of the number of changes $M$, advancing practical applicability. Among active methods, complexity is largely dictated by the change detector: M-UCB \citep{cao2019nearly} uses a lightweight window-based test, while GLR-klUCB adopts a more statistically refined GLR detector, incurring moderately higher cost but achieving nearly optimal guarantees and strong empirical performance. AdSwitch, in contrast, features a substantially more complex detection scheme with elegant theory but no empirical validation.
For example, \citet{manegueu2021generalized} design detectors based on empirical gaps between arms; 
\citet{suk2022tracking} quantify significant shifts at each step to avoid reliance on prior non-stationarity knowledge; 
and \citet{abbasi2023new} achieve comparable guarantees under abrupt changes. 
Complementary to these, multi-scale approaches such as \citet{pmlr-v134-wei21b} maintain multiple competing instances of the base algorithm to strengthen robustness. 
Although these works do not explicitly target the restless bandit problem, their designs—typically involving change detection or adaptive re-weighting—suggest potential applicability. 
However, their theoretical analyses usually rely on specific assumptions about the base algorithm (e.g., optimism or structural properties), 
which narrows the range of admissible base solvers that can be used while still guaranteeing strong regret bounds.

Our framework differs from these works in two key aspects: (i) it introduces a diminishing exploration mechanism to systematically supply evidence for detection without prior knowledge of change frequency, and (ii) it extends these ideas to restless bandits, where state evolution can confound detection.

%% file: 2-problem.tex
\vspace{-5pt}
\section{Problem Formulation}
\vspace{-5pt}
\label{sec:prelim}
We study the piecewise-stationary RMAB (PS-RMAB) problem. 
The system consists of a single learner interacting with $K$ arms. 
Each arm $k \in \{1, \ldots, K\}$ is modeled as a finite-state Markov chain 
$\mathcal{M}_k = (\mathcal{S}, P_k, R_k)$, 
where $\mathcal{S} = \{1, \ldots, S\}$ is the state space, 
$P_k : \mathcal{S} \times \mathcal{S} \to [0,1]$ is the transition kernel, 
and $R_k : \mathcal{S} \to [0,1]$ is the reward distribution. 

At each round $t=1,2,\ldots,T$, the learner selects an arm $A_t$, 
observes its current state $s_{A_t,t}$, 
and receives a random reward $R_{A_t}(s_{A_t,t}) \in [0,1]$. 
The expectation of this reward is denoted by 
$\mu_{A_t,s_{A_t,t}} = \mathbb{E}[R_{A_t}(s_{A_t,t})]$, 
which we call the \emph{per-state mean reward}. 

We assume the environment is \emph{piecewise-stationary}. 
There exist change points $0=\nu_0 < \nu_1 < \cdots < \nu_{M} = T$ such that, 
within each segment $i \in \{1,\ldots,M\}$, 
arm $k$ is characterized by a Markov chain 
$\mathcal{M}_k^{(i)} = (P_k^{(i)}, R_k^{(i)})$ 
with per-state mean rewards $\mu_{k,s}^{(i)}$. 

The object of interest is the \emph{arm-mean} vector 
$\bar{\boldsymbol{\mu}}^{(i)} = (\bar{\mu}_1^{(i)}, \ldots, \bar{\mu}_K^{(i)})$, 
which represents the long-run average rewards of all arms in segment $i$. 
This vector remains fixed within each segment $(\nu_{i-1}, \nu_i]$ 
but may change across consecutive segments 
(i.e., $\bar{\boldsymbol{\mu}}^{(i)} \neq \bar{\boldsymbol{\mu}}^{(i+1)}$ for at least one arm). 
Any modification that alters an arm mean is regarded as a change point. Figure~\ref{fig:PSRMAB} illustrates the instant reward of one of the arms in a PS-RMAB environment. 
\begin{figure}[ht]
    \centering
    \includegraphics[width=0.8\linewidth]{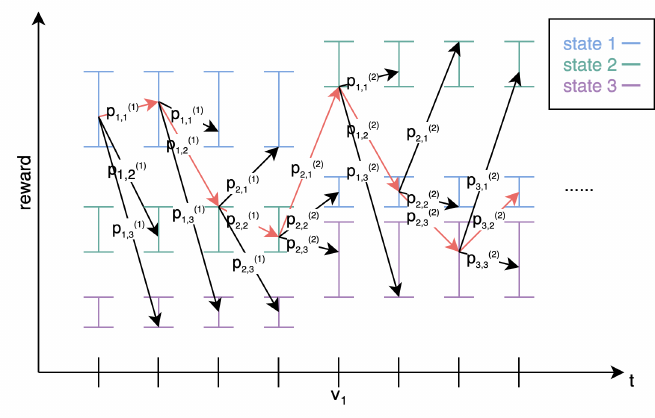}
    \caption{Instant reward of Arm $k$ in PS-RMAB environment}
    \label{fig:PSRMAB}
\end{figure}

Fix a stationary RMAB base solver $\mathcal{B}$. The segment-wise oracle $\pi^{\mathrm{seg}}(\mathcal{B})$ knows the true change points $\nu_1,\ldots,\nu_{M-1}$ and simply restarts $\mathcal{B}$ at the beginning of each segment (all other details—model class, observability, and action constraints—are identical to those faced by the learner). Let $r_t^{\pi}$ denote the expected reward at time $t$ under the learner’s policy and $r_t^{\mathrm{seg}}$ the counterpart under $\pi^{\mathrm{seg}}(\mathcal{B})$

We define the excess regret of a policy $\pi$ (relative to $\pi^{\mathrm{seg}}(\mathcal{B})$) as
\begin{equation}
    \Rex(T)
    \;=\;
    \mathbb{E}\!\Big[\sum_{t=1}^{T} r_t^{\mathrm{seg}}\Big]
    \;-\;
    \mathbb{E}\!\Big[\sum_{t=1}^{T} r_t^{\pi}\Big].
\end{equation}

This metric depends only on the overhead induced by exploration and change detection (and any reset misalignment), while factoring out the stationary performance of the chosen base solver.

%% file: 3-algorithm.tex

\section{Proposed Framework: Diminishing Exploration}

Our framework builds upon the active method paradigm~\citep{yu2009piecewise,liu2018change,
cao2019nearly}, which combines a 
stationary bandit algorithm with a change-detection module. 
Whenever a change is detected, the base algorithm is restarted in the new segment. 
To ensure sufficient environmental sampling for change detection, 
we further integrate a novel exploration scheme, \emph{diminishing exploration}, 
which dynamically balances the trade-off between exploration cost 
and detection delay. This modular design allows arbitrary stationary solvers and detectors to be flexibly combined without altering their internal structures.

\subsection{Diminishing Exploration}

Most existing works adopt a uniform exploration scheme that allocates a 
constant fraction of time to exploration~\citep{liu2018change,
cao2019nearly}. 
While effective for detecting changes, the approach results in regret that scales linearly with the exploration rate and therefore requires prior knowledge of the number of change points $M$ in order to appropriately set this rate. To address this limitation, we introduce \emph{diminishing exploration}, 
which adaptively reduces the frequency of forced exploration over time, thereby eliminating the need for prior knowledge of \(M\).

Let us define $\tau_{i}$ as the $i$-th time when the algorithm alarms a change. In the proposed method, a uniform exploration round starts at $u_{i-1}^{(j)}$ for $j\in\{1, 2, \ldots\}$ with $u_{i-1}^{(1)}=\tau_{i-1}+1$. i.e., the learner chooses to pull the arm $1, 2, \ldots, K$ at time $u_{i-1}^{(j)}, u_{i-1}^{(j)}+1, \ldots, u_{i-1}^{(j)}+K-1$, respectively. The process restarts whenever a new change $\tau_{i}$ is detected. 

\begin{figure}[H]
    \centering
    \includegraphics[width=\linewidth]{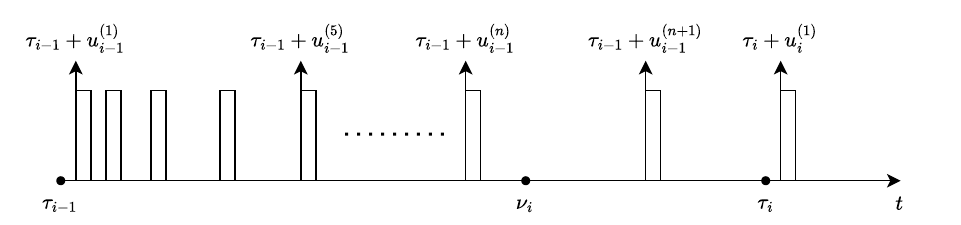}
    \caption{Diminishing exploration.}
    \label{fig:diminishing}
\end{figure}

We aim to balance the regret resulting from exploration and that associated with the performance of change detection by dynamically adjusting the exploration rate \textit{within a segment}. Let $u_{i-1}^{(j)}$ be the start time of the $j$-th uniform exploration session between two consecutive alarms $\tau_{i}$ and $\tau_{i-1}$. 
In our approach, these sessions are designed in such a way that $u_{i-1}^{(j+1)}-u_{i-1}^{(j)}$ is greater than $u_{i-1}^{(j)}-u_{i-1}^{(j-1)}$. This means that the inter-session time within the same time segment increases with $j$, which in turn results in a reduction in the exploration rate. Specifically, for the $i$-th segment, we choose $u_{i}^{(1)}=\left\lceil\left(\alpha-K/4\alpha\right)^{2}\right\rceil$ and
$u_{i}^{(j)} = \left\lceil u_{i}^{(j-1)}+\frac{K}{\alpha}\sqrt{u_{i}^{(j-1)}}+\frac{K^{2}}{4\alpha^{2}} \right\rceil,\quad \forall~1\leq i\leq M \textrm,~j\geq 2$, without the knowledge of $M$ and the parameter $\alpha$ will be chosen later. Clearly, we have $u_{i}^{(j-1)}+K < u_{i}^{(j)}$ for every $j\geq 2$; thus, these exploration phases will not overlap. Moreover, it is obvious that the duration between two exploration phases $u_{i}^{(j)}-u_{i}^{(j-1)}=\mathcal{O}(\sqrt{u_{i}^{(j-1)}})$ increases with time as Figure~\ref{fig:diminishing}; hence, the exploration rate decreases. Thus, we term this mechanism {\it diminishing exploration}.
\subsection{Integration of Modular Components}
The diminishing exploration scheme is orthogonal to the choice of 
change-detection algorithm. 
Any off-the-shelf detector (e.g., CUSUM, GLR, or window-based tests) 
can be used. 
Our framework therefore consists of three modular components:
(i) A \textbf{base algorithm} $\mathcal{B}$ for stationary RMABs (e.g., RestlessUCB, Colored-UCRL2).
(ii) A \textbf{change detector} $\mathcal{D}$ that raises alarms upon detecting arm-mean changes.
(iii) The \textbf{diminishing exploration} scheme $\mathcal{E}$ that schedules exploration rounds to feed $\mathcal{D}$.

Whenever $\mathcal{D}$ signals a change, the base algorithm $\mathcal{B}$ 
is restarted from scratch. 
The modular design ensures forward compatibility, allowing future advancements in any of the three components (e.g., stronger detectors or more efficient base solvers) to be seamlessly incorporated into the framework.

\subsection{Main Algorithm}

\begin{algorithm}[ht]
    \caption{CD-RMAB with Diminishing Exploration (Modular)}\label{alg:main_alg}
    \begin{algorithmic}[1]
        \REQUIRE Horizon $T$, number of arms $K$, exploration parameter $\alpha$; \\
                 \textbf{Base} solver $\mathcal{B}$ (for stationary RMABs); \\
                 \textbf{Change Detector} $\mathcal{D}$ (operates on arm-level reward streams)
        \STATE Initialize segment start $\tau \gets 0$, exploration cursor $u \gets \left\lceil\left(\alpha-\frac{K}{4\alpha}\right)^{2}\right\rceil$, and per-arm counters $n_k \gets 0$ for all $k \in \mathcal{K}$
        \FOR{$t = 1,2,\ldots,T$}
            \IF{$u \leq t-\tau < u+K$}
                \STATE $A_t \gets (t-\tau) - u + 1$
            \ELSE
                \IF{$t-\tau = u+K$} 
                    \STATE $u \gets \left\lceil u + \frac{K}{\alpha}\sqrt{u} + \frac{K^{2}}{4\alpha^{2}} \right\rceil$
                \ENDIF
                \STATE $A_t \gets \mathcal{B}.\textsc{Select}(t, \mathcal{H}_{t-1})$ 
            \ENDIF
            \STATE Pull arm $A_t$; observe current state $s_{A_t,t}$ and reward $r_t \in [0,1]$
            \STATE $\mathcal{B}.\textsc{Update}(A_t, s_{A_t,t}, r_t)$ 
            \STATE $n_{A_t} \gets n_{A_t} + 1$;\quad $Z_{A_t, n_{A_t}} \gets r_t$ 
            \IF{$\mathcal{D}.\textsc{Alarm}(\{Z_{k,\cdot}\}, t) = \textsc{True}$} 
                \STATE $\tau \gets t$;\quad $u \gets 1$;\quad $n_k \gets 0$ for all $k \in \mathcal{K}$
                \STATE $\mathcal{B}.\textsc{Reset}()$ 
            \ENDIF
        \ENDFOR
    \end{algorithmic}
\end{algorithm}

\vspace{-10pt}
The proposed framework is described in Algorithm~\ref{alg:main_alg}, 
which interleaves diminishing exploration (lines 3--7) with actions taken by a 
stationary RMAB base solver $\mathcal{B}$ (line 8). $\mathcal{H}_{t-1}$ denotes the available history for arm selection, which can be either base-only ($\mathcal{H}^{\mathcal{B}}_{t-1}$), where the base algorithm relies solely on its own interaction history, or shared (${\mathcal{H}^{\mathcal{B}}_{t-1},\mathcal{H}^{\mathcal{E}}_{t-1}}$), where it additionally incorporates information gathered during the exploration phase (line 9).
After each step, observed rewards are passed to the change detector $\mathcal{D}$ 
(lines 12--14). 
If $\mathcal{D}$ raises an alarm, the algorithm resets both the exploration 
schedule and the base solver (lines 15--16). 

Although Algorithm~\ref{alg:main_alg} is presented in a modular form, 
our framework is compatible with a wide range of instantiations. 
For example, $\mathcal{B}$ can be any stationary RMAB solver 
(e.g., RestlessUCB~\citep{wang2020restless}, Colored-UCRL2~\citep{ortner2012regret}), 
while $\mathcal{D}$ can be instantiated by standard change-detection methods such as 
CUSUM~\citep{liu2018change}, GLR~\citep{besson2022efficient}, 
or window-based tests~\citep{cao2019nearly}. 
This design highlights the plug-and-play nature of our approach: 
diminishing exploration can be seamlessly integrated with any $\mathcal{B}$ and $\mathcal{D}$.

%% file: 4-analysis.tex
\vspace{-5pt}
\section{Regret Analysis}
\label{sec:analysis}
\vspace{-5pt}

Since the framework allows flexible choices of base solvers $\mathcal{B}$ and change detectors $\mathcal{D}$ at the algorithmic level, our first main result (Theorem~\ref{thm:general_extra}) establishes a general excess regret bound that applies to any such combination. Our analysis treats the base solver as a black box and does not rely on structural properties specific to particular algorithms (e.g., indexability, monotonicity, or other problem-dependent assumptions). To obtain concrete regret guarantees, we then instantiate the framework with specific choices of base solvers and change detectors that satisfy the required conditions. This separation highlights the generality of our framework while ensuring that the final regret bounds remain valid.

To establish an upper bound on the regret, we first introduce the necessary 
notations and event definitions that characterize the behavior of the change 
detection module. These events allow us to disentangle the errors introduced by 
false alarms, detection delay, and the inherent stochasticity of the Markovian 
reward processes. 

Let $\tau_i$ denote the time at which the $i$-th change point is detected by the 
change detector, where $0 \leq i \leq M-1$. 
We define the set of \emph{false alarm events} as 
$F_i := \{\tau_i < \nu_i\}$ for $1 \leq i \leq M-1$, and $F_0 := \{\tau_0 = 0\}$. 
That is, $F_i$ indicates that the $i$-th declared change occurs strictly before the 
true change point $\nu_i$. 
Similarly, we define the set of \emph{successful detection events} as 
$D_i := \{\nu_i \leq \tau_i \leq \nu_i + h_i\}$ for $1 \leq i \leq M-2$, 
where $h_i$ is a detection delay parameter determined by the underlying change 
detector. For completeness, we also set $D_0 := \{\tau_0 = 0\}$ and 
$D_{M-1} := \{\tau_{M-1} \leq T\}$.

In addition, for the regret decomposition we define two gap quantities. 
For each arm $k$ and segment $i$, 
$\Delta_k^{(i)} := \max_{k' \in \mathcal{K}} (\bar{\mu}_{k'}^{(i)} - \bar{\mu}_k^{(i)})$ 
and $\delta_k^{(i)} := |\bar{\mu}_k^{(i+1)} - \bar{\mu}_k^{(i)}|$. 
Finally, we define $\delta^{(i)} := \max_{k \in \mathcal{K}} \delta_k^{(i)}$ as the 
maximum mean shift across arms at change point $\nu_i$.

Recall the excess regret $\Rex(T) := \mathbb{E}\!\big[\sum_{t=1}^{T} r_t^{\mathrm{seg}}\big]-\mathbb{E}\!\big[\sum_{t=1}^{T} r_t^{\pi}\big]$, where $r_t^{\mathrm{seg}}$ corresponds to the segment–oracle that restarts the same base solver at true change points. We decompose $\Rex(T)$ into an \emph{exploration cost} and a \emph{change–detection cost}. Throughout the analysis, we focus on the \emph{fully modular} setting, where arm selection relies solely on the base algorithm’s history $\mathcal{H}^{\mathcal{B}}_{t-1}$. Incorporating $\mathcal{H}^{\mathcal{E}}_{t-1}$ would make the analysis dependent on the specific base algorithm.

\begin{theorem}[Extra–regret bound]\label{thm:general_extra}
For any base solver $\mathcal{B}$ and change detector $\mathcal{D}$, the excess regret of Algorithm~\ref{alg:main_alg} satisfies
\begin{multline} 
\Rex(T)\leq \underbrace{2\alpha\sqrt{MT}}_{(a)}
    +\underbrace{\sum^{M-1}_{i=1}\E\left[\tau_{i}-\nu_{i}\middle | D_{i}\overline{F}_{i}D_{i-1}\overline{F}_{i-1}\right]}_{(b)}\\
+\underbrace{T\sum^{M}_{i=1}\mathbb{P}\left(F_{i}\middle| \overline{F}_{i-1}D_{i-1}\right)+T\sum^{M-1}_{i=1}\mathbb{P}\left(\overline{D}_{i}\middle|\overline{F}_{i}\overline{F}_{i-1}D_{i-1}\right)}_{(c)}, \label{eqn:regret_bound}
\end{multline}
\end{theorem}
\vspace{-5pt}
\begin{proof}
We rewrite the excess regret as the difference between the total regret of our policy and that of the segment oracle, i.e. $\Rex(T)=(\E[\textstyle\sum_{t=1}^T r_t^{\star}]-\E[\textstyle\sum_{t=1}^T r_t^{\pi}])-(\E[\textstyle\sum_{t=1}^T r_t^{\star}]-\E[\textstyle\sum_{t=1}^T r_t^{\mathrm{seg}}]).$
The first term corresponds to the regret of running the base algorithm within each segment, i.e., the sum of the base solver’s regrets across all $M$ segments.
For the second term, we expand it recursively across change points. In the ideal case of correct detections, this adds the detection delays together with the sum of the base solver’s regrets over all $M$ segments.  
If a false alarm or missed detection occurs, we treat these as worst-case events and charge them up to $T$ regret each.  
Thus, once a concrete change detector is chosen, the analysis ensures that the probabilities of false alarms and missed detections are sufficiently small, leaving the main contributions from the forced exploration and the detection delay.  
Finally, since the base solver’s regret cancels out between the two terms, we arrive exactly at the decomposition stated in Theorem~\ref{thm:general_extra}. The detailed proof is deferred to Supplementary Materials.
\end{proof}
\vspace{-5pt}
The bound in Theorem~\ref{thm:general_extra} is entirely independent of the 
choice of base solver~$\mathcal{B}$. 
Indeed, the \emph{excess regret} arises solely from two sources: 
(i) the forced exploration introduced by the diminishing exploration schedule 
(Algorithm~\ref{alg:main_alg}), contributing term~(a), and 
(ii) the behavior of the change detection module~$\mathcal{D}$, which governs 
the detection delay and false-alarm probabilities, contributing terms~(b)--(c). 
\vspace{-5pt}
\subsection{One-State Special Case} \label{subsec:one-state}
\vspace{-5pt}
In this subsection, we consider the one-state setting, which reduces the PS-RMAB
to the classical piecewise-stationary bandit problem. 
We first present a general upper bound on the excess regret that does not assume
any particular base solver or change detector. 
We then instantiate this bound with specific algorithms to obtain concrete order guarantees. To make the general extra-regret bound in Theorem~\ref{thm:general_extra} 
concrete, we now instantiate both the base solver and the change detector. 

\textbf{Base Solver.} 
For the one-state case, the base solver $\mathcal{B}$ can be chosen as the classical 
UCB algorithm, which is known to achieve near-optimal regret in stationary MABs. 
Since the one-state PS-RMAB reduces to the piecewise-stationary MAB setting, 
this choice is natural and ensures that the base solver itself does not inflate 
the overall regret order.

\textbf{Change Detector.} 
We instantiate the CD module with the M-UCB Change Detector~\citep{cao2019nearly}, 
a sliding-window mean-shift test that raises an alarm whenever the empirical 
means of two consecutive halves of a window differ by more than a threshold. 
Its parameters $(w,b)$ are chosen according to
\begin{equation}
    w \;=\; \frac{4}{\delta^2}\Bigl(\sqrt{\log(2KT^2)} + \sqrt{\log(2T)}\Bigr)^2, \label{eqn:w}
\end{equation}
\begin{equation}
b \;=\; \sqrt{w \log(2KT^2/2)}.  \label{eqn:b}
\end{equation}
where $\delta$ is a known lower bound on the magnitude of mean shifts. 
The full pseudocode is in Appendix~\ref{app:cd_alg}. 

\begin{remark}
Assumptions such as the existence of a known gap parameter $\delta$ 
(Assumption~\ref{ass:minimum_gap} in \cite{cao2019nearly}) and minimum segment lengths 
(Assumption~\ref{asm:seg_length}) are required specifically for the M-UCB detector to guarantee 
its statistical properties. These are \emph{not} intrinsic to our framework: 
diminishing exploration itself does not impose additional requirements. 
This distinction highlights the modularity of our approach, which allows the 
plug-and-play use of different CD algorithms under their respective conditions. 
\end{remark}

Combining the base solver $\mathcal{B}$ with the M-UCB Change Detector, 
we obtain the following corollary of Theorem~\ref{thm:general_extra}.

\begin{corollary}\label{cor:onestate_ucb_mucb}
Combining Algorithm~\ref{alg:main_alg} (diminishing exploration with resets) 
with the base solver $\mathcal{B}$ and Algorithm~\ref{alg:mucb_cd} with parameters $(w,b)$ given in Equations~\ref{eqn:w} and~\ref{eqn:b}, 
the excess regret is upper-bounded as
\begin{equation}
    \Rex(T)
    \;\le\;
    \mathcal{O}\!\left(\sqrt{KMT\log T}\right).
\end{equation}
\end{corollary}
\vspace{-10pt}
\textbf{Proof Outline of Corollary~\ref{cor:onestate_ucb_mucb}.}
We decompose the excess regret into three sources: (i) the cost of diminishing exploration within stationary segments, (ii) the overhead caused by false alarms, and (iii) the delay incurred in detecting true changes. The following lemmas bound each component. Finally, combining these bounds with Theorem~\ref{thm:general_extra} yields the desired order result.

In what follows, the first lemma establishes a bound on the regret incurred during the diminishing exploration phase of the algorithm. 
We denote by $R_{\mathrm{DE}}(\tau_{i-1}, \nu_i)$ the regret accumulated due to exploration between the previous alarm time $\tau_{i-1}$ and the next change point $\nu_i$, 
and by $N_{\mathrm{DE},k}(\tau_{i-1}, \nu_i)$ the number of times arm $k$ is selected during this exploration phase.
\begin{lemma}[Diminishing exploration regret]\label{lemma:de_regret}
If the mean values of the arms remain the same during the time interval $[\tau_{i-1}, \nu_{i})$, then we have
\begin{equation}
    N_{\mathrm{DE},k}\left(\tau_{i-1}, \nu_{i}\right)\leq \frac{2\alpha\sqrt{\nu_{i}-\tau_{i-1}}}{K}+\frac{3}{2},
\end{equation}
and
\begin{equation}
    \E\left[R_{\mathrm{DE}}\left(\tau_{i-1}, \nu_{i}\right)\right]\leq 2\alpha\sqrt{\nu_{i}-\tau_{i-1}}+\frac{3}{2}K.
\end{equation}
\end{lemma}
\vspace{-10pt}
This lemma shows that, on any stationary interval, the regret incurred purely from diminishing exploration is controlled by a term proportional to the square root of the interval length. Summing over all $M{+}1$ stationary segments and applying Cauchy–Schwarz, we obtain a total contribution of order $\sqrt{MT}$ up to logarithmic factors.

\begin{lemma}[Probability of false alarm]\label{lemma:prob_fa} 
Under Algorithm~\ref{alg:main_alg} with parameter in Equations~\ref{eqn:w} and \ref{eqn:b}, we have
\vspace{-5pt}
\begin{multline}
    \mathbb{P}\left(F_{i}\middle|\overline{F}_{i-1}D_{i-1}\right)\leq \\ 
    wK\left(1-\left(1-\exp\left(-2b^{2}/w\right)\right)^{\left\lfloor T/w \right\rfloor }\right)\leq \frac{1}{T}.    
\end{multline}
\end{lemma}
\vspace{-10pt}
The above ensures that false alarms occur with vanishing probability. Therefore, the expected regret overhead due to spurious resets is negligible, contributing at most $\mathcal{O}(1)$ in expectation.

\begin{lemma}[Probability of successful detection]\label{lemma:prob_delay} 
Consider a piecewise-stationary bandit environment. For any $\boldsymbol{\mu}^{(i)},\boldsymbol{\mu}^{(i+1)}\in\left[0,1\right]^{K}$ with parameters chosen in \eqref{eqn:w} and \eqref{eqn:b} 
and
\begin{equation}
h_{i} = \left\lceil w\left(\frac{K}{2\alpha}+1\right)\sqrt{s_{i}+1}+\frac{w^{2}}{4}\left(\frac{K}{2\alpha}+1\right)^{2} \right\rceil,
\end{equation}
for some $k\in\gK, i\geq 1$ and $c>0$, under the \Algref{alg:main_alg}, we have 
\begin{equation}
    \mathbb{P}\left(D_{i}\middle|\overline{F}_{i}\overline{F}_{i-1}D_{i-1}\right)\geq 1-\frac{1}{T}.
\end{equation}
\end{lemma}
\vspace{-10pt}
This lemma guarantees that each genuine change is detected with high probability, which prevents the algorithm from staying too long in a mismatched segment.

\begin{lemma}[Expected detection delay]\label{lemma:exp_delay} 
Consider a piecewise-stationary bandit environment. For any $\boldsymbol{\mu}^{(i)},\boldsymbol{\mu}^{(i+1)}\in\left[0,1\right]^{K}$ with parameters chosen in \eqref{eqn:w} and \eqref{eqn:b} 
and
\begin{equation}
h_{i} = \left\lceil w\left(\frac{K}{2\alpha}+1\right)\sqrt{s_{i}+1}+\frac{w^{2}}{4}\left(\frac{K}{2\alpha}+1\right)^{2} \right\rceil,
\end{equation}
for some $K\in\gK, i\geq 1$ and $c>0$, under the \Algref{alg:main_alg}, we have
\begin{equation}
    \E\left[\tau_{i}-\nu_{i}\middle| \overline{F}_{i}D_{i}\overline{F}_{i-1}D_{i-1}\right]\leq h_{i}.
\end{equation}
\end{lemma}

In other words, once a change occurs, the delay before detection is bounded in expectation by $h_i$, which scales sublinearly with $s_i$ under our parameter choice. This ensures that the regret accumulated during delays is well controlled.

\textbf{Putting everything together.}
The regret from diminishing exploration (Lemma~\ref{lemma:de_regret}) is $\mathcal{O}(\sqrt{KMT})$; the false alarm contribution (Lemma~\ref{lemma:prob_fa}) is negligible; and the regret during detection delays (Lemmas~\ref{lemma:prob_delay}–\ref{lemma:exp_delay}) is bounded by $\tilde{\mathcal{O}}(\sqrt{KMT})$. Substituting these into the general bound in Theorem~\ref{thm:general_extra}, we obtain
\begin{equation}
\Rex(T) \le \mathcal{O}\!\left(\sqrt{KMT\log T}\right),
\end{equation}
which concludes the proof sketch.

In the one-state setting, PS-RMAB reduces to the classical piecewise-stationary MAB. 
In addition to the extra-regret perspective, our framework naturally extends to a 
standard regret analysis against a clairvoyant oracle that always pulls the best arm 
in each segment. 

If the base solver $\mathcal{B}$ is chosen as UCB, then under bounded rewards and a 
change detector with controlled false-alarm probability and expected delay, the total 
regret of Algorithm~\ref{alg:main_alg} remains near-optimal, scaling as 
$\tilde{\mathcal{O}}(\sqrt{KMT})$. 
The detailed proof is deferred to Supplementary Materials.
\vspace{-5pt}
\subsection{General Case}
\vspace{-5pt}
Compared to the one-state case, the general PS-RMAB setting introduces additional challenges 
due to the temporal dependence of rewards induced by the Markovian dynamics. 
To handle this dependence, our analysis leverages the mixing time of each arm’s Markov chain, 
which allows us to approximate independence and apply concentration inequalities. 
To make this argument rigorous, we require a technical assumption ensuring that each arm’s 
Markov chain admits well-defined steady-state behavior. 
In particular, ergodicity guarantees the existence of unique stationary distributions and 
enables the use of concentration tools in our regret analysis.

\begin{assumption}\label{assump:ergodic}
All Markov chains $\mathcal{M}_k^{(i)} = (P_k^{(i)},R_k^{(i)})$ are ergodic. 
This ensures the existence of a unique steady-state distribution $d_k^{(i)}$ for 
each arm $k$ in segment $i$, satisfying 
$d_k^{(i)} = d_k^{(i)} P_k^{(i)}$ with $\sum_{s \in \mathcal{S}} d_k^{(i)}(s) = 1$. 
We then define the steady-state arm mean as 
$\bar{\mu}_k^{(i)} := d_k^{(i)} R_k^{(i)}$.
\end{assumption}
\vspace{-5pt}
This assumption is necessary because the reward sequence of each arm arises from a 
Markov chain, which induces temporal dependence. 
Ergodicity allows us to approximate independence through the mixing properties 
of the chains and to apply concentration inequalities (e.g., Hoeffding-type bounds 
for Markov processes). 
Without this assumption, the notion of a stationary arm mean would not be 
well defined, and regret guarantees could not be rigorously established.

To accommodate this dependence within the change detection module, we refine the choice of its parameters. Specifically, the window size $w$ and threshold $b$ are tuned in accordance with the mixing-time bound, ensuring that the detector achieves the required statistical guarantees while remaining compatible with the diminishing exploration schedule.
Formally, we select

\begin{multline}\label{eqn:w_gen}
w \;=\; \frac{4}{\delta_{\min}^{2}}\!\left(
  \sqrt{2\ln(2KT^{2})}
  \;\right.\\ +\left.\; \sqrt{144\,L\,\ln(2KT^{2})}
  \;+\; \sqrt{144\,L\,\ln(2T)}
\right)^{2}, 
\end{multline}
\vspace{-0.5cm}
\begin{equation}\label{eqn:b_gen}
b \;=\; \sqrt{\frac{w}{2}\,\ln(2KT^{2})}
\;+\; \sqrt{144\,w\,L\,\ln(2KT^{2})}.
\end{equation}
where $\delta_{\min}$ is a known lower bound on the mean shift, and $L = \max_{i \in \{1,\dots,M\}} \max_{k \in \mathcal{K}} L_k^{(i)}$ 
denotes the maximum mixing time of the Markov chains associated with all arms and segments.

\begin{corollary}\label{cor:general_mucb}
Combining Algorithm~\ref{alg:main_alg} (diminishing exploration with resets) 
with the base solver $\mathcal{B}$ and Algorithm~\ref{alg:mucb_cd} with parameters $(w,b)$ given in Equations~\ref{eqn:w_gen} and \ref{eqn:b_gen}, 
the excess regret is upper-bounded as
\begin{equation}
    \Rex(T) \;\le\; \mathcal{O}\!\left(\sqrt{MKLT\log T}\right),
\end{equation}
where $L = \max_{i \in \{1,\dots,M\}} \max_{k \in \mathcal{K}} L_k^{(i)}$ 
denotes the maximum mixing time of the Markov chains associated with all arms and segments.
\end{corollary}
\vspace{-10pt}
\textbf{Proof Outline of Corollary~\ref{cor:general_mucb}.}
The overall proof follows the same decomposition as in Corollary~\ref{cor:onestate_ucb_mucb}, namely by bounding the regret incurred from (i) diminishing exploration within stationary segments, (ii) regret caused by false alarms, and (iii) regret due to delays in detecting true changes. The bounds from Lemma~\ref{lemma:de_regret} apply directly here as well. 

For the false alarm probability, we now invoke Lemma~\ref{lemma:prob_fa}, which states that under Algorithm~\ref{alg:mucb_cd} with parameters chosen according to Equations~\ref{eqn:w_gen} and~\ref{eqn:b_gen}, the probability of a false alarm in each stationary period is at most
\begin{equation}
\mathbb{P}\!\left(F_i \,\middle|\, \overline F_{i-1} D_{i-1}\right) \;\le\; \frac{\epsilon\|\varphi\|_{d}+1}{T}.
\end{equation}
This again ensures that the contribution of false alarms to the overall regret is negligible in expectation. 

For the detection probability and delay, Lemma~\ref{lemma:prob_delay} shows that for each change, with the choice of parameters $w,b$ as in Equations~\ref{eqn:w_gen} and~\ref{eqn:b_gen} and the delay threshold $h_i$,
\begin{equation}
\mathbb{P}\!\left(\overline{D}_i \,\middle|\, \overline F_{i-1}D_{i-1}\right)\;\le\; \frac{\epsilon \|\varphi\|_{d}}{T},
\end{equation}
so that missed detections occur with vanishing probability. Moreover, the expected detection delay is bounded by $h_i$, which scales sublinearly with the segment length. These results play the same role as Lemmas~\ref{lemma:prob_delay} and~\ref{lemma:exp_delay} in the one-state case. Here, $\epsilon > 0$ is an arbitrarily small constant, and $\varphi$ denotes the initial state distribution of the Markov chains.

Therefore, the regret contributions from diminishing exploration, false alarms, and detection delays remain of the same order as before, and substituting these bounds into Theorem~\ref{thm:general_extra} yields
\begin{equation}
\Rex(T) \le \mathcal{O}\!\left(\sqrt{KMLT \log T}\right),
\end{equation}
which establishes Corollary~\ref{cor:general_mucb}.

From Theorem~\ref{thm:general_extra}, the additional cost introduced by 
diminishing exploration and change detection, captured by $\Rex$, 
is strictly controlled and does not alter the order of regret. 
To argue near-optimality, we turn to the \emph{one-state case}, 
which reduces to the classical PS-MAB setting. 
Here, our framework combined with UCB attains the known 
$\tilde {\mathcal{O}}(\sqrt{KMT})$ bound, confirming that the overhead is indeed 
compatible with near-optimal performance.

The regret bound depends on problem-specific parameters such as the mixing time of the underlying Markov chains, reward gaps, and detector parameters. These quantities influence constant factors but do not alter the asymptotic rate. In particular, larger mixing times or smaller mean shifts may increase detection difficulty, leading to longer delays, yet the overall regret order remains unchanged.

Motivated by this evidence, we extend the conclusion to the general PS-RMAB setting:

\begin{principle}[Transfer Principle]\label{princ:transfer}
If the base solver $\mathcal{B}$ achieves near-optimal regret in the 
stationary RMAB setting, then Algorithm~\ref{alg:main_alg} equipped with 
$\mathcal{B}$ and a suitable change detector $\mathcal{D}$ 
inherits this near-optimality in the PS-RMAB setting.
\end{principle}
\vspace{-10pt}
\begin{proof}
Since the stationary RMAB already admits a minimax lower bound of 
$\Omega(\sqrt{ST})$ \citep{ortner2012regret}, and our framework in the 
piecewise-stationary case introduces only an additional 
$\tilde{\mathcal O}(\sqrt{MKT})$ overhead, the overall regret rate 
remains $\tilde{\mathcal O}(\sqrt{T})$. 
Thus, the near-optimality of stationary solvers is preserved under 
piecewise-stationary extensions.
\end{proof}

%% file: 6-simulate.tex
\vspace{-15pt}
\section{Simulation Results}
\vspace{-5pt}
\label{sec:sim}

We evaluate our framework using \texttt{RestlessUCB} and \texttt{colorUCRL2} as base algorithms.  
For each case, we report both the \emph{excess regret} defined in our framework and the standard regret benchmarked against a value-iteration oracle baseline.  
The simulations adopt the \emph{partial modular} setting, where information is shared between the exploration module and the base algorithm, enabling more accurate and comprehensive estimation for arm selection; this may explain why, as shown in Figures~\ref{fig:extra_regret}, the observed excess regret can occasionally be smaller than that of the segment oracle.
The results are summarized in Figures~\ref{fig:extra_regret} and~\ref{fig:standard_regret}.

\textbf{Environment Configuration.} 
We simulate a PS-RMAB environment with a horizon of $T=100{,}000$ time steps, partitioned into $M=5$ stationary segments. 
The problem consists of $K=3$ arms, each modeled as a Markov chain with $S=3$ states. 
The stationary mean rewards of the arms are configured as 
$\bar{\mu}^{(i)}_k \in \{0.2,\,0.5,\,0.8\}$, 
assigned according to the rule $(i+k)\bmod 3 = 2,0,1$, respectively, for segment $i$ and arm $k$.

\textbf{Extra vs. Standard Regret.} 
As shown in Figure~\ref{fig:extra_regret}, the excess regret exhibits only minor variation across different choices of base algorithms.  
This indicates that the additional overhead induced by our mechanism is essentially independent of the base solver.  
Comparing with the standard regret in Figure~\ref{fig:standard_regret}, we observe that the overhead introduced by our framework is negligible relative to the dominant regret incurred by the base algorithms themselves.  
In particular, the \texttt{UE} curve represents the optimal exploration amount achievable when the number of segments is known in advance, while \texttt{DE} corresponds to our proposed diminishing exploration scheme without prior knowledge of the number of changes.

\begin{figure}[ht]
\centering
\begin{subfigure}{0.7\textwidth}
    \includegraphics[width=0.7\linewidth]{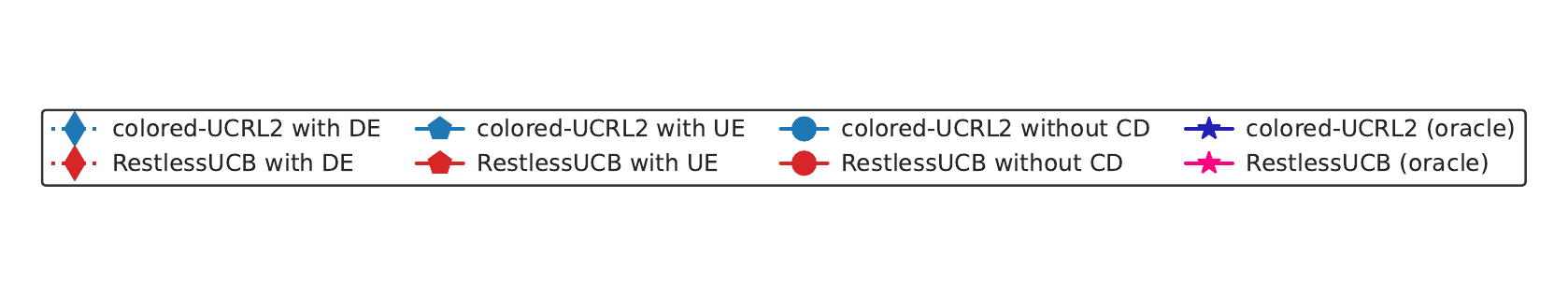}
    \vspace{-0.5cm}
\end{subfigure}
\begin{subfigure}{0.24\textwidth}
    \includegraphics[width=\textwidth]{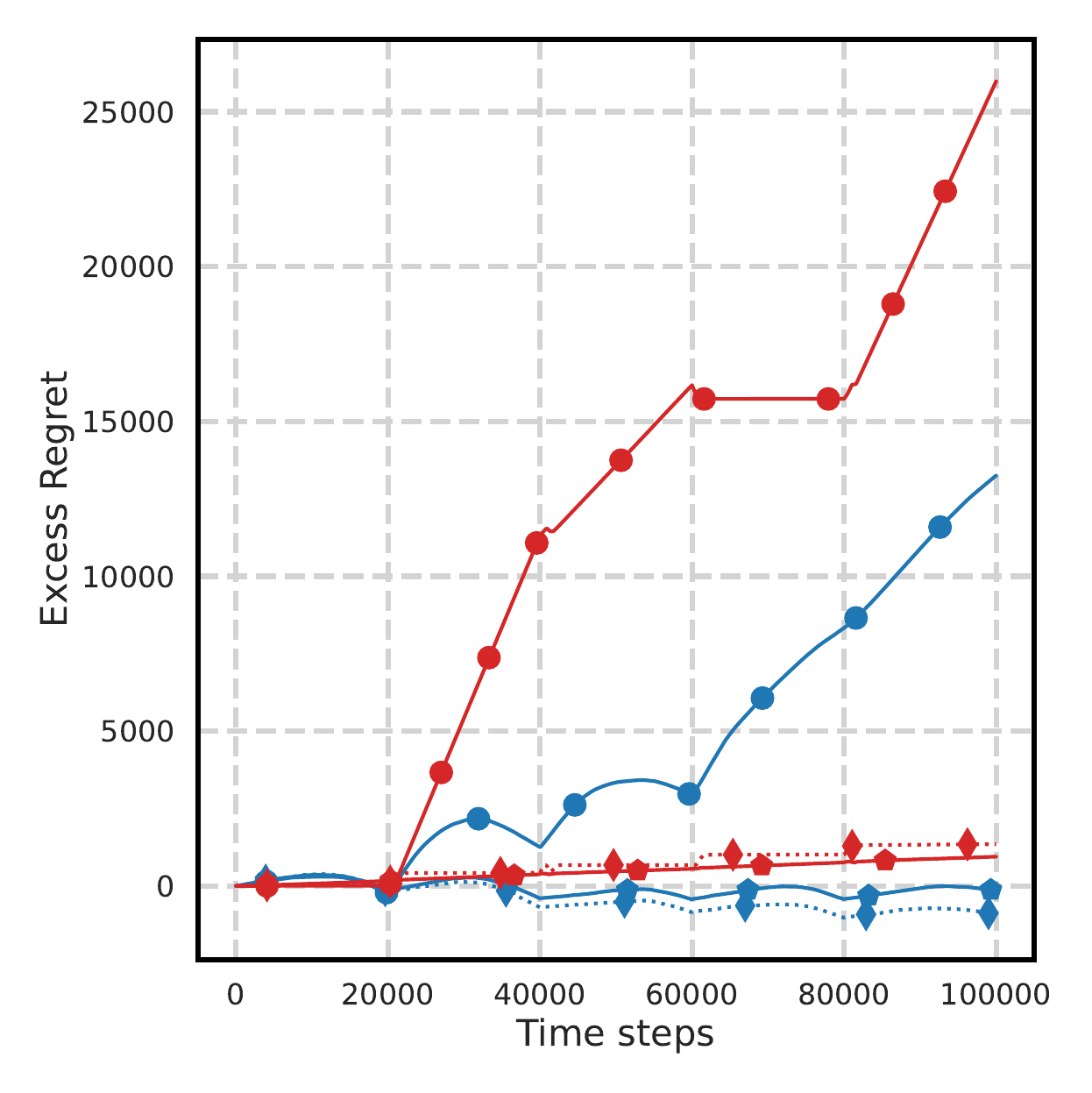}
    \vspace{-15pt}
    \caption{Excess Regret.}
    \label{fig:extra_regret} 
\end{subfigure}
\hspace{-6pt}
\begin{subfigure}{0.24\textwidth}
    \includegraphics[width=\textwidth]{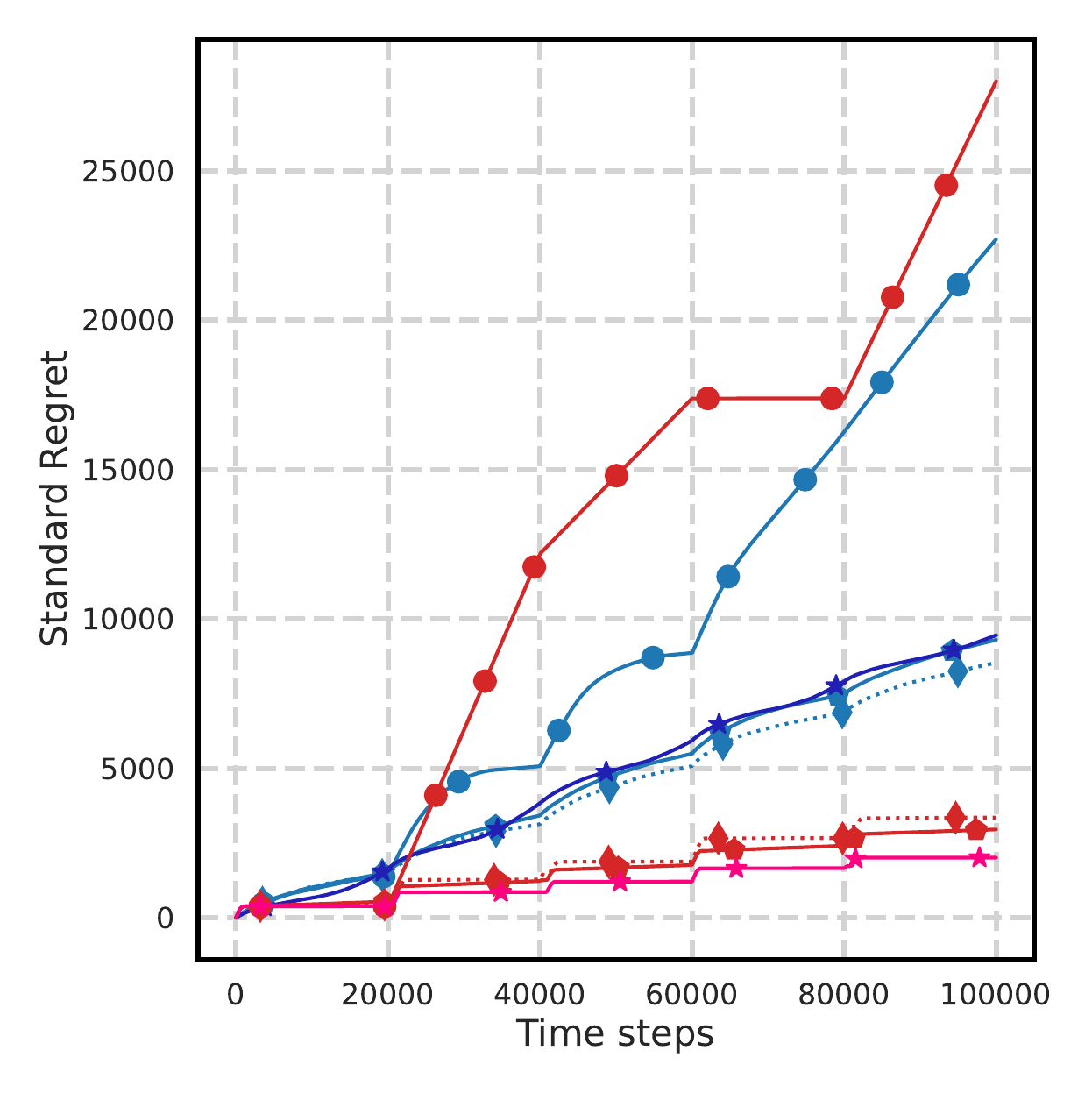}
    \vspace{-15pt}
    \caption{Standard Regret.}
    \label{fig:standard_regret}
\end{subfigure}
\caption{}\label{fig:3}
\end{figure}

\textbf{Significance of Change Detection.} 
One may wonder: since the regret of the base algorithms already dominates, is it sufficient to simply let them adapt on their own, without any change-detection or exploration mechanism?  
To examine this, we also simulate the scenario where no change-detection module is used.  
As shown in Figures~\ref{fig:extra_regret} and~\ref{fig:standard_regret}, the regret exhibits a much sharper growth when the environment changes are ignored, indicating that the base algorithms alone struggle to adapt promptly to non-stationarity.  
This illustrates the potential performance degradation caused by neglecting environmental shifts.

%% file: 7-conclude.tex
\vspace{-5pt}
\section{Conclusion}
\label{sec:conclusion}
We addressed the PS-RMAB problem by introducing a modular framework that integrates stationary solvers, change detection, and a novel diminishing exploration schedule. Central to our analysis is a refined notion of excess regret, which isolates the costs of exploration and detection delay. We proved that this overhead matches the minimax order known for piecewise-stationary bandits, and established a transfer principle showing that the near-optimality of stationary solvers carries over to the PS-RMAB setting. This provides both a unifying theoretical perspective and a practical recipe for extending stationary algorithms to dynamic environments.

%% file: 1-changealg.tex

\section{Change Detection Algorithm}\label{app:cd_alg}
{\bf Change Detector of M-UCB } \\
The following algorithm is the change detection algorithm for M-UCB \citep{cao2019nearly}
\begin{center}
    \begin{minipage}{0.8\linewidth}
        \begin{algorithm}[H]
            \caption{Change Detection of M-UCB: CD$\left(w,b,Z_{1},\ldots,Z_{w}\right)$}\label{alg:mucb_cd}
            \begin{algorithmic}[1]
                \REQUIRE An even integer $w$, $w$ observations $Z_{1},\ldots,Z_{w}$, and a prescribed threshold $b > 0$
                \IF{$\left\lvert\sum^{w}_{\ell=w/2+1}Z_{\ell}-\sum^{w/2}_{\ell=1}Z_{\ell}\right\rvert>b$}
                \STATE Return True
                \ELSE
                \STATE Return False
                \ENDIF
            \end{algorithmic}
        \end{algorithm}
    \end{minipage}
\end{center}

In this algorithm, one requires $w$ observations as input and check whether the difference between the sample average of the first half and that of the second half exceeds a prescribed threshold $b$ (line 1).

{\bf Integration with the change detector of M-UCB.} 

The following two assumptions are required to establish the analytical results when integrating M-UCB into our modularized framework.
\begin{assumption}\label{ass:minimum_gap}
    The algorithm knows a lower bound $\delta>0$ such that $\delta\leq\min_{i}\max_{k\in\mathcal{K}}\delta^{\left(i\right)}_{k}$.
\end{assumption}
Note that Assumption~\ref{ass:minimum_gap} is Assumption 1(b) of \cite{cao2019nearly}, which is required for the M-UCB detector to determine good $w$ and $b$ in regret analysis. It is worth noting that almost all schemes that actively detect changes share similar assumptions; however, different algorithms may impose distinct sets of assumptions. This assumption is mild since $\delta$ may be statistically derived from historical information. Furthermore, even if the lower bound does not hold, and we occasionally encounter changes with expected reward gaps smaller than the assumed $\delta$, such changes may be perceived as too minor to result in significant regret. 

\begin{assumption}\label{asm:seg_length}
    We assume
    \begin{equation}s_{i}= \Omega\left((\log{KT}+\sqrt{K\log{KT}})\sqrt{s_{i-1}}\right).
    \end{equation}
\end{assumption}

In particular, if $s_i = \Theta\left(\left(\log{KT}+\sqrt{K\log{KT}}\right)^{2(1+\epsilon)}\right)$ for every $i$, Assumption~\ref{asm:seg_length} holds.
This assumption essentially posits that the changes are not overly dense, a condition that holds in many practical scenarios. Simple math would then show that given $D_{i-1}$ is true, with this assumption and the proposed diminishing exploration, each arm will have at least $w/2$ observations before and after a change point. Again, we note that similar assumptions are imposed in other algorithms that actively detect changes, while different algorithms may impose different assumptions. This assumption is necessary with our proof technique, which requires every change to be successfully detected with high probability. 

%% file: 2-proof.tex
\section{Proof Detail}
\label{app:pf_detail}

In this appendix, we provide detailed proofs of the results presented in Section 4. Each proof is carefully elaborated to ensure clarity and rigor.

\subsection{Proof of Section 4}\label{app:pf_detail:not_extend}
In this subsection, we present the proofs of theorems in Section 4. In what follows, the first lemma bounds the regret accumulated during the diminishing exploration part of the algorithm. We denote by $R_{\mathrm{DE}}\left(\tau_{i-1}, \nu_{i}\right)$ as the regret caused by the exploration part of the algorithm and by $N_{\mathrm{DE},k}\left(\tau_{i-1}, \nu_{i}\right)$ the number of times that the arm $k$ is selected in the exploration phase from the previous alarm time to the next change point.

\begin{lemma}[Diminishing exploration regret]\label{applemma:de_regret}
    If the mean values of the arms remain the same during the time interval $[\tau_{i-1}, \nu_{i})$, then we have
    \begin{equation}
        N_{\mathrm{DE},k}\left(\tau_{i-1}, \nu_{i}\right)\leq \frac{2\alpha\sqrt{\nu_{i}-\tau_{i-1}}}{K}+\frac{3}{2},
    \end{equation}
    and
    \begin{equation}
        \E\left[R_{\mathrm{DE}}\left(\tau_{i-1}, \nu_{i}\right)\right]\leq 2\alpha\sqrt{\nu_{i}-\tau_{i-1}}+\frac{3}{2}K.
    \end{equation}
\end{lemma}

\begin{proof}
    Recall that $u_{i}^{(j)}$ is the beginning of the $j$-th uniform exploration session in the $i$-th segment. In Algorithm \ref{alg:main_alg}, the initial time of the first exploration session after each $\tau_{i}$ is given by:
    \begin{equation}
        u_{i}^{(1)}=\left\lceil\left(\alpha-\frac{K}{4\alpha}\right)^{2}\right\rceil, \label{eqn:algscheme_initial}
    \end{equation}
    and subsequent times follow the recursive equation:
    \begin{equation}
        u_{i}^{(j)}=\left\lceil u_{i}^{(j-1)}+\frac{K}{\alpha}\sqrt{u_{i}^{(j-1)}}+\frac{K^{2}}{4\alpha^{2}}\right\rceil\geq u_{i}^{(j-1)}+\frac{K}{\alpha}\sqrt{u_{i}^{(j-1)}}+\frac{K^{2}}{4\alpha^{2}}. \label{eqn:algscheme}
    \end{equation}
   
    Based on \eqref{eqn:algscheme_initial} and \eqref{eqn:algscheme}, one could easily check that the sequence $u_{i}^{(n)}$ satisfies that for every natural number $n$,
    \begin{equation}\label{eqn:u_i_bound}
        u_{i}^{(n)} \geq \left(\frac{\left(2n-3\right)K}{4\alpha}+\alpha\right)^{2}.
    \end{equation}

    Let $u_{i}^{(m)}$ be the last exploration start time in time interval $[\tau_{i-1}, \nu_{i})$. Then, we have

    \begin{equation}
        \E\left[R_{\mathrm{DE}}\left(\nu_{i},\tau_{i-1}\right)\right]\leq mK. \label{eqn:R_DE}
    \end{equation}

    Additionally, we have:

    \begin{equation}
        \nu_{i}-\tau_{i-1}\geq u_{i}^{(m)} \geq \left(\frac{\left(2m-3\right)K}{4\alpha}+\alpha\right)^{2}\geq \left(\frac{2\E\left[R_{DE}\left(\nu_{i}-\tau_{i-1}\right)\right]-3K}{4\alpha}+\alpha\right)^{2}. \label{eqn:T_mK}
    \end{equation}

    Finally, based on the \eqref{eqn:R_DE} and \eqref{eqn:T_mK}, we can conclude that:              
    \begin{equation}
        \E\left[R_{\textrm{DE}}\left(\nu_{i},\tau_{i-1}\right)\right]\leq 2\alpha\sqrt{\nu_{i}-\tau_{i-1}}-2\alpha^2+\frac{3}{2}K
        \leq 2\alpha\sqrt{\nu_{i}-\tau_{i-1}}+\frac{3}{2}K,
    \end{equation}
    and
    \begin{equation}
        N_{\textrm{DE},k}\left(\nu_{i},\tau_{i-1}\right)\leq \frac{2\alpha\sqrt{\nu_{i}-\tau_{i-1}}}{K}+\frac{3}{2}.
    \end{equation}

\end{proof}

In the following lemma, we aim to explore how many time steps it takes for a given arm to reach a certain number of samples through diminishing exploration.

\begin{lemma}[Samples-time steps transform]\label{applemma:samples-time}
    When each arm has accumulated \(n\) samples, the required time is as follows: 
    If the counting of the \(n\) samples begins immediately after a reset, the required time is given by
    \begin{equation}\label{eqn:Treset}
        T_{reset}\leq \left(\alpha+\frac{\left(2n-3\right)K}{4\alpha}+n\right)^2+K.
    \end{equation}
    However, if there is a delay of \(t_{d}\) time steps after the reset before the counting begins, the required time is given by 
    \begin{equation}\label{eqn:Ttd}
        T_{t_{d}}\leq 2n\left(\frac{K}{2\alpha}+1\right)\sqrt{t_{d}+1}+n^{2}\left(\frac{K}{2\alpha}+1\right)^{2}.
    \end{equation}
    
\end{lemma}

\begin{proof}

    We can derive the following from Equation~\ref{eqn:algscheme} in the proof of Lemma~\ref{applemma:de_regret}:
    
    \begin{equation}\label{eqn:u_iter}
    u^{(j)} \leq u^{(j-1)} + \frac{k}{\alpha} \sqrt{u^{(j-1)}} + \frac{k^2}{4\alpha^2} + 1 
    \leq \left( \sqrt{u^{(j-1)}} + \frac{k}{2\alpha} \right)^2 + 1 
    \leq \left( \sqrt{u^{(j-1)}} + \frac{k}{2\alpha} + 1 \right)^2.
    \end{equation}
    
    First, let us consider the case where the counting of \(n\) samples begins immediately after the reset. From Equation~\ref{eqn:algscheme_initial}, we can derive the following:
    
    \begin{equation}
    u^{(1)} \leq \left( \alpha - \frac{K}{4\alpha} + 1 \right)^2.
    \end{equation}
    
    Using Equation~\ref{eqn:u_iter}, we can further derive:
    
    \begin{equation}
    u^{(2)} \leq \left( \alpha - \frac{K}{4\alpha} + 1 + \frac{K}{2\alpha} + 1 \right)^2.
    \end{equation}
    
    Finally, we obtain:
    
    \begin{equation}
    u^{(n)} \leq \left( \alpha - \frac{K}{4\alpha} + 1 + (n-1)\left( \frac{K}{2\alpha} + 1 \right) \right)^2.
    \end{equation}
    
    Here, \(u^{(n)}\) represents the starting time of the \(n\)-th exploration block. The total time required to guarantee that all \(K\) arms have been sampled \(n\) times is therefore:
    
    \begin{equation}
    T_{\text{reset}} = u^{(n)} + K \leq \left( \alpha + \frac{(2n-3)K}{4\alpha} + n \right)^2 + K.
    \end{equation}
    
    Next, we consider how long it takes for each arm to be sampled \(n\) times after \(t_d\) time steps following the reset. We first assume that, prior to \(t_d\), each arm has already been sampled \(x\) times. For simplicity, we assume an ideal case where the exploration block starts exactly at time \(t_d + 1\). Therefore, we have:
    
    \begin{equation}
    u^{(x+1)} = t_d + 1.
    \end{equation}
    
    Since, in reality, the exploration block start time may not exactly coincide with \(t_d + 1\), we account for the possibility that it could begin at a later time by considering the next exploration block's start time. This allows us to bound the non-ideal case.
    
    Following the same approach as in the first part of the proof, we eventually obtain:
    
    \begin{equation}
    u^{(x+n+1)} \leq \left( \sqrt{t_d + 1} + n\left(\frac{K}{2\alpha} + 1\right) \right)^2.
    \end{equation}
    
    Finally, we derive that the total time required for each arm to be sampled \(n\) times after \(t_d\) time steps is:
    
    \begin{equation}
    T_{t_d} = u^{(x+n+1)} - u^{(x+1)} \leq 2n\left( \frac{K}{2\alpha} + 1 \right) \sqrt{t_d + 1} + n^2 \left( \frac{K}{2\alpha} + 1 \right)^2.
    \end{equation}

\end{proof}

Define $R_{\mathrm{excess}}(r,s):=\E[\textstyle\sum_{t=r}^s r_t^{\mathrm{seg}}]-\textstyle\sum_{t=r}^s r_t^{\pi}$ be the regret accumulated during $r$ and $s$. In the next lemma, we provide an upper bound on the regret accumulated from the $\left(i-1\right)$-th alarm time to the end of $\left(i-1\right)$-th segment, given that the previous change was successfully detected.

\begin{lemma}[Excess Regret bound with stationary bandit]\label{applemma:excess_regret_stat} 
Consider a stationary bandit interval with $\nu_{i-1}<\tau_{i-1}<\nu_{i}$. Condition on the successful detection events $\overline{F}_{i-1}$ and $D_{i-1}$, the expected regret accumulated during $\left(\tau_{i-1},\nu_{i}\right)$ can be bounded by
\begin{equation}
    \E\left[R_{\mathrm{excess}}\left(\tau_{i-1},\nu_{i}\right)\middle|\overline{F}_{i-1}D_{i-1} \right]\leq 2\alpha\sqrt{s_{i}}+T\cdot \mathbb{P}\left(F_{i}\middle|\overline{F}_{i-1}D_{i-1}\right),\label{eqn:one_seg_regret}
\end{equation}
\end{lemma}

\begin{proof}
We begin by rewriting the definition of the excess regret over a stationary segment $(\tau_{i-1}, \nu_i]$ as
\[
R_{\mathrm{excess}}\!\left(\tau_{i-1},\nu_{i}\right)
= \Bigg(\mathbb{E}\Big[\textstyle\sum_{t=\tau_{i-1}}^{\nu_{i}} r_t^{\star}\Big] - \textstyle\sum_{t=\tau_{i-1}}^{\nu_{i}} r_t^{\pi}\Bigg)
 - \Bigg(\mathbb{E}\Big[\textstyle\sum_{t=\tau_{i-1}}^{\nu_{i}} r_t^{\star}\Big] - \mathbb{E}\Big[\textstyle\sum_{t=\tau_{i-1}}^{\nu_{i}} r_t^{\mathrm{seg}}\Big]\Bigg),
\]
where the first term corresponds to the standard regret and the second to the standard regret of the segment oracle.  
We denote the standard regret over this interval by $R(\tau_{i-1},\nu_{i})$ and analyze it first.

For every $i$, conditioning on the previous successful detection $\overline{F}_{i-1} D_{i-1}$, we decompose the regret depending on whether a false alarm occurs:
\begin{subequations}
\begin{align}
    \mathbb{E}\!\left[R\left(\tau_{i-1},\nu_{i}\right)\middle|\overline{F}_{i-1}D_{i-1} \right]
    &=\mathbb{E}\!\left[R\left(\tau_{i-1},\nu_{i}\right)\middle|F_{i}\overline{F}_{i-1}D_{i-1}\right]\mathbb{P}\!\left(F_{i}\middle|\overline{F}_{i-1}D_{i-1}\right) \\
    &\quad+\mathbb{E}\!\left[R\left(\tau_{i-1},\nu_{i}\right)\middle|\overline{F}_{i}\overline{F}_{i-1}D_{i-1}\right]\mathbb{P}\!\left(\overline{F}_{i}\middle|\overline{F}_{i-1}D_{i-1}\right)\\
    &\leq T\cdot \mathbb{P}\!\left(F_{i}\middle|\overline{F}_{i-1}D_{i-1}\right)+\mathbb{E}\!\left[R\left(\tau_{i-1},\nu_{i}\right)\middle|\overline{F}_{i}\overline{F}_{i-1}D_{i-1}\right]\\
    &\leq T\cdot \mathbb{P}\!\left(F_{i}\middle|\overline{F}_{i-1}D_{i-1}\right)+2\alpha\sqrt{s_{i}}+\sum_{i=1}^M C_i, \label{eqn:regret_seg}
\end{align}
\end{subequations}

Next, subtracting the segment–oracle standard regret on the same interval cancels out the base regret terms $\sum_{i=1}^M C_i$, leaving only the exploration and detection contributions. This yields the desired excess regret bound stated in Lemma~\ref{applemma:excess_regret_stat}.
    
\end{proof}

\begin{theorem}[Excess-regret bound]\label{appthm:general_extra}
For any base solver $\mathcal{B}$ and change detector $\mathcal{D}$, the excess regret of Algorithm~\ref{alg:main_alg} satisfies
\begin{multline} 
\Rex(T)\leq 2\alpha\sqrt{MT}
    +\sum^{M-1}_{i=1}\E\left[\tau_{i}-\nu_{i}\middle | D_{i}\overline{F}_{i}D_{i-1}\overline{F}_{i-1}\right]+T\sum^{M}_{i=1}\mathbb{P}\left(F_{i}\middle| \overline{F}_{i-1}D_{i-1}\right)+T\sum^{M-1}_{i=1}\mathbb{P}\left(\overline{D}_{i}\middle|\overline{F}_{i}\overline{F}_{i-1}D_{i-1}\right), \label{eqn:exregret_bound}
\end{multline}
\end{theorem}

\begin{proof}
    Recall that $R_{\mathrm{excess}}(r,s):=\E[\textstyle\sum_{t=1}^T r_t^{\mathrm{seg}}]-\textstyle\sum_{t=1}^T r_t^{\pi}$, then  $\Rex\left(T\right)=\E\left[R_{\mathrm{excess}}\left(1,T\right)\right]$. We have
    \begin{subequations}
        \label{10}
        \begin{align}
            \Rex\left(T\right)&=\E\left[R_{\mathrm{excess}}\left(1,T\right)\right]\\
            &=\E\left[R_{\mathrm{excess}}\left(1,T\right)\middle|\overline{F}_{0}D_{0}\right] \label{eqn:exR_init}\\
            &\leq \E\left[R_{\mathrm{excess}}\left(1,\nu_{1}\right)\middle|\overline{F}_{1}\overline{F}_{0}D_{0}\right]+\E\left[R_{\mathrm{excess}}\left(\nu_{1},T\right)\middle|\overline{F}_{1}\overline{F}_{0}D_{0}\right]+ T\cdot \mathbb{P}\left(F_{1}\middle| \overline{F}_{0}D_{0}\right)\label{eqn:exR_conE1}\\
            &\leq  2\alpha\sqrt{\left(\nu_{1}-\nu_{0}\right)}+\E\left[R\left(\nu_{1},T\right)\middle| \overline{F}_{1}\overline{F}_{0}D_{0}\right]+T\cdot \mathbb{P}\left(F_{1}\middle| \overline{F}_{0}D_{0}\right),\label{eqn:exR_conE2}
        \end{align}
    \end{subequations}
    where \eqref{eqn:exR_init} holds because $\tau_{0}=0$, \eqref{eqn:exR_conE1} is due to the law of total expectation and some trivial bounds, and \eqref{eqn:exR_conE2} follows from Lemmas~\ref{applemma:regret_stat}. The third term in \eqref{eqn:exR_conE2} is then further bounded as follows:

    \begin{subequations}
        \begin{align}
            \hspace{-20pt}\E\left[R_{\mathrm{excess}}\left(\nu_{1},T\right)\middle| \overline{F}_{1}\overline{F}_{0}D_{0}\right] 
            &\leq\E\left[R_{\mathrm{excess}}\left(\nu_{1},T\right)\middle| D_{1}\overline{F}_{1}\overline{F}_{0}D_{0}\right]+T\cdot \left(1-\mathbb{P}\left(D_{1}\middle| \overline{F}_{1}\overline{F}_{0}D_{0}\right)\right)\label{eqn:exR_conF1}\\
            &\hspace{-60pt}\leq \E\left[R_{\mathrm{excess}}\left(\nu_{1},T\right)\middle| D_{1}\overline{F}_{1}\overline{F}_{0}D_{0}\right] + T\cdot\mathbb{P}\left(\overline{D}_{1}\middle| \overline{F}_{1}\overline{F}_{0}D_{0}\right)\label{eqn:exR_conF2}\\     
            &\hspace{-60pt}= \E\left[R_{\mathrm{excess}}\left(\tau_{1},T\right)\middle| D_{1}\overline{F}_{1}\overline{F}_{0}D_{0}\right]+\E\left[R\left(\nu_{1},\tau_{1}\right)\middle| D_{1}\overline{F}_{1}\overline{F}_{0}D_{0}\right]+T\cdot\mathbb{P}\left(\overline{D}_{1}\middle| \overline{F}_{1}\overline{F}_{0}D_{0}\right)\label{eqn:exreg_split1}\\
            &\hspace{-60pt}\leq\E\left[R_{\mathrm{excess}}\left(\tau_{1},T\right)\middle| \overline{F}_{1}D_{1}\right] + \E\left[\tau_{1}-\nu_{1}\middle| \overline{F}_{1}D_{1}\overline{F}_{0}D_{0}\right]\label{eqn:reg_split2}+T\cdot\mathbb{P}\left(\overline{D}_{1}\middle| \overline{F}_{1}\overline{F}_{0}D_{0}\right)\\
            &\hspace{-60pt}\leq\E\left[R_{\mathrm{excess}}\left(\tau_{1},T\right)\middle| \overline{F}_{1}D_{1}\right] + \E\left[\tau_{1}-\nu_{1}\middle| \overline{F}_{1}D_{1}\right]+T\cdot\mathbb{P}\left(\overline{D}_{1}\middle| \overline{F}_{1}\overline{F}_{0}D_{0}\right)\label{eqn:exreg_split3},
        \end{align}
    \end{subequations}
    where \eqref{eqn:exR_conF1} applies the law of total expectation and some trivial bounds. 
    From here, we can set up the following recursion:

    \begin{subequations}
        \begin{align}
            &\E\left[R_{\mathrm{excess}}\left(1,T\right)\right] =\E\left[R_{\mathrm{excess}}\left(1,T\right)\middle|\overline{F}_{0}D_{0}\right] \\
            & \leq \E\left[R_{\mathrm{excess}}\left(\tau_{1},T\right)\middle|\overline{F}_{1}D_{1}\right]+2\alpha\sqrt{s_{1}-1}+\E\left[\tau_{1}-\nu_{1}\middle| \overline{F}_{1}D_{1}\right]\\
            &\hspace{+60pt}+T\cdot\mathbb{P}\left(F_{1}\middle| \overline{F}_{0}D_{0}\right)+T\cdot\mathbb{P}\left(\overline{D}_{1}\middle| \overline{F}_{1}\overline{F}_{0}D_{0}\right) \\
            & \leq \E\left[R_{\mathrm{excess}}\left(\tau_{2},T\right)\middle|\overline{F}_{2}D_{2}\right]+2\alpha\sum^{2}_{i=1}\sqrt{s_{i}-1}\\ 
            &+\sum^{2}_{i=1}\E\left[\tau_{i}-\nu_{i}\middle| \overline{F}_{i-1}D_{i-1}\right]+T\sum^{2}_{i=1}\mathbb{P}\left(F_{i}\middle| \overline{F}_{i-1}D_{i-1}\right)+T\sum^{2}_{i=1}\mathbb{P}\left(\overline{D}_{i}\middle| \overline{F}_{i}\overline{F}_{i-1}D_{i-1}\right) \\
            &\hspace{+150pt}\vdots\nonumber\\
            & \leq 2\alpha\sum^{M}_{i=1}\sqrt{s_{i}}+\sum^{M-1}_{i=1}\E\left[\tau_{i}-\nu_{i}\middle| \overline{F}_{i-1}D_{i-1}\right]\\
            &\hspace{+60pt}+T\sum^{M}_{i=1}\mathbb{P}\left(F_{i}\middle| \overline{F}_{i-1}D_{i-1}\right)+T\sum^{M-1}_{i=1}\mathbb{P}\left(\overline{D}_{i}\middle| \overline{F}_{i}\overline{F}_{i-1}D_{i-1}\right)\\
            & \leq 2\alpha\sqrt{MT}+\sum^{M-1}_{i=1}\E\left[\tau_{i}-\nu_{i}\middle| \overline{F}_{i-1}D_{i-1}\right]\label{eqn:exR_segment}\\
            &\hspace{+60pt}+T\sum^{M}_{i=1}\mathbb{P}\left(F_{i}\middle| \overline{F}_{i-1}D_{i-1}\right)+T\sum^{M-1}_{i=1}\mathbb{P}\left(\overline{D}_{i}\middle| \overline{F}_{i}\overline{F}_{i-1}D_{i-1}\right)
        \end{align}
    \end{subequations}
    where \eqref{eqn:exR_segment} follows from the Cauchy–Schwarz inequality
    \begin{subequations}
        \begin{align}
            \left(\sum^{M}_{i=1}\sqrt{s_{i}}\right)^{2}\leq \left(\sum^{M}_{i=1}s_{i}\right)\left(\sum^{M}_{i=1}1\right)= M\sum^{M}_{i=1}s_{i}= MT,
        \end{align}
    \end{subequations}
    and the fact that $\sum^{M}_{i=1}s_{i}=T$. This completes the proof.
\end{proof}

\subsection{Proof of One-State Special Case}

With the general regret bound in Theorem~\ref{appthm:general_extra}, a regret bound of the change detector with the proposed diminishing exploration can be obtained by bounding $\mathbb{P}\left(F_{i}\middle|\overline{F}_{i-1}D_{i-1}\right)$, $\mathbb{P}\left(D_{i}\middle|\overline{F}_{i}\overline{F}_{i-1}D_{i-1}\right)$, and $\E\left[\tau_{i}-\nu_{i}\middle| \overline{F}_{i}D_{i}\overline{F}_{i-1}D_{i-1}\right]$. In what follows, we adopt window-based change detectors (Algorithm~\ref{app:cd_alg}) as the basis for our subsequent analysis.

First, in Lemma~\ref{applemma:prob_fa}, we show that the probability of false alarm is very small; thereby, its contribution to the regret is negligible. 

\begin{lemma}[Probability of false alarm]\label{applemma:prob_fa} 
Under Algorithm~\ref{alg:main_alg} with parameter in Equations~\ref{eqn:w} and \ref{eqn:b}, we have
\begin{equation}
    \mathbb{P}\left(F_{i}\middle|\overline{F}_{i-1}D_{i-1}\right)\leq wK\left(1-\left(1-\exp\left(-2b^{2}/w\right)\right)^{\left\lfloor T/w \right\rfloor }\right)\leq \frac{1}{T}.
\end{equation}
\end{lemma}

\begin{proof}
    Suppose that at time $t$, we have gathered $w$ samples of arm $k\in\gK$, namely $Y_{k,1}, Y_{k,2},\ldots, Y_{k,w}$, for change detection in line 17 of \Algref{alg:main_alg}, and we define
    \begin{equation}
        S_{k,t}= \sum^{w}_{\ell=w/2 +1}Y_{k,\ell}-\sum^{w/2}_{\ell=1}Y_{k,\ell}.            \label{eqn:s_def}
    \end{equation}
    Note that $S_{k,t}=0$ when there is insufficient (less than $w$) samples to trigger the change detection algorithm. By definition, we have
    \begin{equation}
        \tau_{k,i}=\inf \{t\geq \tau_{i-1}+w:\left\lvert S_{k,t}\right\rvert >b\}.
    \end{equation}
    Given that the events $D_{i-1}$ and $\bar{F}_{i-1}$ hold, we define $\tau_{k,i}$ as the first detection time of the $k$-th arm after $\nu_{i}$. Clearly, $\tau_{i}=\min _{k\in\gK} \left\{\tau_{k,i}\right\}$ as \Algref{alg:main_alg} would reset every time a change is detected. Using the union bound, we have
    \begin{subequations}
        \begin{align}
            \mathbb{P}\left(F_{i}\middle| \overline{F}_{i-1}D_{i-1}\right)=&\mathbb{P}\left(\max_{k\in\mathcal{K}} \sum_{t=\tau_{i-1}+1}^{\nu_{i}}\1_{\left\{A_{t}=k\right\}}\geq w,F_{i}\middle| \overline{F}_{i-1}D_{i-1}\right)\\
            &+\mathbb{P}\left(\max_{k\in\mathcal{K}} \sum_{t=\tau_{i-1}+1}^{\nu_{i}}\1_{\left\{A_{t}=k\right\}}<w,F_{i}\middle| \overline{F}_{i-1}D_{i-1}\right) \label{eqn:false_al_p_b}\\
            =&\mathbb{P}\left(F_{i}\middle| \overline{F}_{i-1}D_{i-1},\max_{k\in\mathcal{K}} \sum_{t=\tau_{i-1}+1}^{\nu_{i}}\1_{\left\{A_{t}=k\right\}}\geq w\right)\label{eqn:false_al_p_c}\\
            &\cdot \mathbb{P}\left(\max_{k\in\mathcal{K}} \sum_{t=\tau_{i-1}+1}^{\nu_{i}}\1_{\left\{A_{t}=k\right\}}\geq w\middle| \overline{F}_{i-1}D_{i-1}\right)\label{eqn:false_al_p_d}\\
            \leq &\mathbb{P}\left(F_{i}\middle|\overline{F}_{i-1}D_{i-1},\max_{k\in\mathcal{K}} \sum_{t=\tau_{i-1}+1}^{\nu_{i}}\1_{\left\{A_{t}=k\right\}}\geq w\right)\label{eqn:false_al_p_e}\\
            \leq & \sum^{K}_{k=1}\mathbb{P}\left(\tau_{k,i}\leq \nu_{i}\middle|\overline{F}_{i-1}D_{i-1}, \max_{k^{\prime}\in\mathcal{K}} \sum_{t=\tau_{i-1}+1}^{\nu_{i}}\1_{\left\{A_{t}=k^{\prime}\right\}}\geq w\right)\label{eqn:max_1}\\
            \leq & \sum^{K}_{k=1}\mathbb{P}\left(\tau_{k,i}\leq \nu_{i}\middle|\overline{F}_{i-1}D_{i-1}, \sum_{t=\tau_{i-1}+1}^{\nu_{i}}\1_{\left\{A_{t}=k\right\}}\geq w\right), \label{eqn:false_al_p_f}
        \end{align}
    \end{subequations}
    where the term in \eqref{eqn:false_al_p_b} is clearly equal to $0$ as there will be no false alarm if we do not even have sufficiently many observations to trigger the alarm as suggested by Algorithm~\ref{alg:mucb_cd}. \Eqref{eqn:false_al_p_c} and \eqref{eqn:false_al_p_d} hold by the definition of conditional probability, \eqref{eqn:false_al_p_e} is due to the fact that the term in \eqref{eqn:false_al_p_d} is at most one, and \eqref{eqn:max_1} follows from the union bound. In \eqref{eqn:false_al_p_f}, if $k\neq k^{\prime}$, we cannot  guarantee that $\sum_{t=\tau_{i-1}+1}^{\nu_{i}}\1_{\left\{A_{t}=k^{\prime}\right\}}\geq w$. Hence, some $k$ might cause the probability in the \eqref{eqn:max_1} to be zeros. 
 
    For any $0\leq j\leq w-1$, define the stopping time
    \begin{equation}
        \tau_{k,i}^{(j)}:=\inf \{t=\tau_{i-1} + j+nw,n\in\mathbb{Z}^{+}:\left\lvert S_{k,t}\right\rvert >b\}.
    \end{equation}
    Clearly, $\tau_{k,i}=\min\{\tau_{k,i}^{(0)},\ldots,\tau_{k,i}^{(w-1)}\}$. Let us define, for any $0\leq j\leq w-1$,
    \begin{equation}
        \xi_{k,i}^{(j)}=\frac{\left(\tau_{k,i}^{(j)}-j-\tau_{i-1}\right)}{w}.
    \end{equation}
    Note that condition on the events $D_{i-1}$ and $\bar{F}_{i-1}$, $\xi_{k, i}^{(j)}$ is a geometric random variable with parameter $p := \mathbb{P }(\left\lvert S_{k,t}\right\rvert>b)$, because when fixing $j$, there is no overlap between the samples in the current window and the next.
    \begin{multline}
        \mathbb{P}\left(\tau_{k,i}^{(j)}=\tau_{i-1}+nw+j\middle|\overline{F}_{i-1}D_{i-1},\sum_{t=\tau_{i-1}+1}^{\nu_{i}}\1_{\left\{A_{t}=k\right\}}\geq w\right)\\=
            \mathbb{P}\left(\xi_{k,i} = n\middle|\overline{F}_{i-1}D_{i-1},\sum_{t=\tau_{i-1}+1}^{\nu_{i}}\1_{\left\{A_{t}=k\right\}}\geq w\right)
            =p(1-p)^{n-1}.
    \end{multline}

    Here, the inclusion of subsequent events as conditions should not impact the results, as when entering the change detection algorithm, those events have already occurred.
    Moreover, by union bound, we have that for any $k\in\gK$,
    \begin{subequations}
        \begin{align}
            \mathbb{P}\left(\tau_{k,i}\leq \nu_{i}\middle|\overline{F}_{i-1}D_{i-1},\sum_{t=\tau_{i-1}+1}^{\nu_{i}}\1_{\left\{A_{t}=k\right\}}\geq w\right)&\leq w\left( 1-\left(1-p\right)^{\left\lfloor \left(\nu_{i}-\tau_{i-1}\right)/w \right\rfloor}  \right)\\ 
            &\leq w\left( 1-(1-p)^{\left\lfloor T/w \right\rfloor}  \right).  \label{eqn:tau_ki}
        \end{align}
    \end{subequations}
    We further use the McDiarmid's inequality and the union bound to show that

    \begin{subequations}
        \begin{align}
            p&=\mathbb{P}\left(\left\lvert S_{k,t} \right\rvert>b\right)=\mathbb{P}\left(S_{k,t}>b\right)+\mathbb{P}\left(S_{k,t}<-b\right)\\
            &\leq 2\cdot \exp\left(-\frac{2b^{2}}{w}\right). \label{eqn:pp}
        \end{align}
    \end{subequations}
    Using the result in \eqref{eqn:tau_ki} and \eqref{eqn:pp} into \eqref{eqn:false_al_p_f}, 
    \begin{subequations}
        \begin{align}
            &\mathbb{P}\left(F_{i}\middle|\overline{F}_{i-1}D_{i-1}\right)\leq \sum^{K}_{k=1}w\left( 1-\left(1-2 \exp\left(-\frac{2b^{2}}{w}\right)\right)^{\left\lfloor T/w \right\rfloor}  \right) \\
            &= wK\left( 1-\left(1-2 \exp\left(-\frac{2b^{2}}{w}\right)\right)^{\left\lfloor T/w \right\rfloor}  \right).
        \end{align}
    \end{subequations}
    Moreover, applying $\left(1-x\right)^{a}>1-ax$ for any $a>1$ and $0<x<1$ and plugging the choice of $b = \sqrt{w\log\left(2KT^{2} \right)/2}$ as in \eqref{eqn:b} shows the second inequality.
\end{proof}

Lemma~\ref{applemma:samples-time} ensures that, with high probability, the detection delay is confined within a tolerable interval. 
That is, each arm is sampled \(w/2\) times, and using equation~\ref{eqn:Ttd} from lemma~\ref{applemma:samples-time}, we select \(h_{i}\) as
\begin{equation}
h_{i} = \left\lceil w\left(\frac{K}{2\alpha}+1\right)\sqrt{s_{i}+1}+\frac{w^{2}}{4}\left(\frac{K}{2\alpha}+1\right)^{2} \right\rceil.
\end{equation}

\begin{lemma}[Probability of successful detection]\label{applemma:prob_delay} 
Consider a piecewise-stationary bandit environment. For any $\boldsymbol{\mu}^{(i)},\boldsymbol{\mu}^{(i+1)}\in\left[0,1\right]^{K}$ with parameters chosen in \eqref{eqn:w} and \eqref{eqn:b} 
and
\begin{equation}
h_{i} = \left\lceil w\left(\frac{K}{2\alpha}+1\right)\sqrt{s_{i}+1}+\frac{w^{2}}{4}\left(\frac{K}{2\alpha}+1\right)^{2} \right\rceil,
\end{equation}
for some $k\in\gK, i\geq 1$ and $c>0$, under the \Algref{alg:main_alg}, we have 
\begin{equation}
    \mathbb{P}\left(D_{i}\middle|\overline{F}_{i}\overline{F}_{i-1}D_{i-1}\right)\geq 1-\frac{1}{T}.
\end{equation}
\end{lemma}

\begin{proof}
    \begin{subequations}
        \begin{align}
            \mathbb{P}\left(D_{i}\middle|\overline{F}_{i}\overline{F}_{i-1}D_{i-1}\right)
            = &\mathbb{P}\left(\tau_{i}\leq \nu_{i}+h_{i}\middle|\overline{F}_{i}\overline{F}_{i-1}D_{i-1}\right)\\
            \geq &\max_{t \in \left\{ \nu_{i}+1,\ldots,\nu_{i}+h_{i}\right\}}\mathbb{P}\left(S_{\tilde{k},t}>b\middle| \overline{F}_{i}\overline{F}_{i-1}D_{i-1}\right)\label{eqn:alarm_pb3}\\
            \geq &\max_{j \in \left\{ 0,\ldots,w/2\right\}} \left(1-2\exp\left(-\frac{( j\left\lvert\delta^{(i)}_{\tilde{k}}\right\rvert-b)^{2}}{w}\right)\right)\label{eqn:alarm_pb4}\\
            = & 1-2\exp\left(-\frac{(w|\delta^{(i)}_{\tilde{k}}|/2-b)^{2}}{w}\right)\label{eqn:alarm_pb5}\\
            \geq & 1-2\exp\left(-\frac{wc^{2}}{4}\right)\label{eqn:alarm_pb6}.
        \end{align}
    \end{subequations}
    where $S_{\tilde{k},t}$ is defined in \eqref{eqn:s_def}, \eqref{eqn:alarm_pb4} follows from the McDiarmid’s inequality, and \eqref{eqn:alarm_pb5} is due to the fact that the maximum value is attained when $j=w/2$. Last, \eqref{eqn:alarm_pb6} is true for any choice of $w,b$ and $c$ such that $\delta^{(i)}_{\tilde{k}}\geq 2b/w+c$ holds. We thus set $w$ and $b$ as in \eqref{eqn:w} and \eqref{eqn:b}, respectively, and choose $c=2\sqrt{\log\left(2T\right)/w}$, which leads to $\mathbb{P}\left(D_{i}\middle|\overline{F}_{i}\overline{F}_{i-1}D_{i-1}\right)\geq 1-1/T$.
\end{proof}
    
Lemma~\ref{applemma:exp_delay} further bounds the expected detection delay in the situation where the change detection algorithm successfully detects the change within the desired interval.

\begin{lemma}[Expected detection delay]\label{applemma:exp_delay} 
Consider a piecewise-stationary bandit environment. For any $\boldsymbol{\mu}^{(i)},\boldsymbol{\mu}^{(i+1)}\in\left[0,1\right]^{K}$ with parameters chosen in \eqref{eqn:w} and \eqref{eqn:b} 
and
\begin{equation}
h_{i} = \left\lceil w\left(\frac{K}{2\alpha}+1\right)\sqrt{s_{i}+1}+\frac{w^{2}}{4}\left(\frac{K}{2\alpha}+1\right)^{2} \right\rceil,
\end{equation}
for some $k\in\gK, i\geq 1$ and $c>0$, under the \Algref{alg:main_alg}, we have
\begin{equation}
    \E\left[\tau_{i}-\nu_{i}\middle| \overline{F}_{i}D_{i}\overline{F}_{i-1}D_{i-1}\right]\leq h_{i}.
\end{equation}
\end{lemma}

\begin{proof}
    For any $1\leq i\leq M$, we have
    \begin{subequations}
        \begin{align}
            \E\left[\tau_{i}-\nu_{i}\middle| \overline{F}_{i}D_{i}\overline{F}_{i-1}D_{i-1}\right]  
            &=\sum^{h_{i}}_{j=1}\mathbb{P}\left(\tau_{i}\geq\nu_{i}+j\middle| \overline{F}_{i}D_{i}\overline{F}_{i-1}D_{i-1}\right)
            \leq h_{i}.
        \end{align}
    \end{subequations}
\end{proof}

Plugging the bounds in Lemmas ~\ref{applemma:prob_fa},~\ref{applemma:prob_delay} and~\ref{applemma:exp_delay} into Theorem~\ref{appthm:general_extra} shows the following regret bound in Corollary~\ref{appcor:onestate_ucb_mucb}.

\begin{corollary}\label{appcor:onestate_ucb_mucb}
Combining Algorithm~\ref{alg:main_alg} (diminishing exploration with resets) 
with the base solver $\mathcal{B}$ and Algorithm~\ref{alg:mucb_cd} with parameters $(w,b)$ given in Equations~\ref{eqn:w} and~\ref{eqn:b}, 
the excess regret is upper-bounded as
\begin{equation}
     \Rex(T)\leq 2\alpha\sqrt{MT}+w\left(\frac{K}{2\alpha}+1\right)\sqrt{M\left(T+M\right)}+\frac{w^{2}M}{4}\left(\frac{K}{2\alpha}+1\right)^{2}+2M.
\end{equation}\label{eqn:regret_bound_MUCB_p}By setting $\alpha = c\sqrt{K\log{\left(KT\right)}}$ for some constant $c$, the expected regret is upper-bounded by $\mathcal{O}(\sqrt{KMT\log{T}})$.
\end{corollary}

\subsubsection{Proof of standard regret in one-state case}

Define $R\left(r,s\right) :=\sum^{s}_{t=r}\max_{k\in\gK}\E\left[X_{k,t}\right]-X_{A_{t},t}$ be the regret accumulated during $r$ and $s$. In the next lemma, we provide an upper bound on the regret accumulated from the $\left(i-1\right)$-th alarm time to the end of $\left(i-1\right)$-th segment, given that the previous change was successfully detected.

\begin{lemma}[Regret bound with stationary bandit]\label{applemma:regret_stat} 
Consider a stationary bandit interval with $\nu_{i-1}<\tau_{i-1}<\nu_{i}$. Condition on the successful detection events $\overline{F}_{i-1}$ and $D_{i-1}$, the expected regret accumulated during $\left(\tau_{i-1},\nu_{i}\right)$ can be bounded by
\begin{equation}
    \E\left[R\left(\tau_{i-1},\nu_{i}\right)\middle|\overline{F}_{i-1}D_{i-1} \right]\leq \tilde{C}_{i}+2\alpha\sqrt{s_{i}}+T\cdot \mathbb{P}\left(F_{i}\middle|\overline{F}_{i-1}D_{i-1}\right),\label{eqn:one_seg_regret}
\end{equation}
where $\tilde{C}_{i}=8\sum_{\Delta^{\left(i\right)}_{k}>0}\frac{\log T}{\Delta^{\left(i\right)}_{k}}+\left(\frac{5}{2}+\frac{\pi^{2}}{3}+K\right)\sum^{K}_{k=1}\Delta^{\left(i\right)}_{k}$.
\end{lemma}

\begin{proof}
    For every $i$, we have
    \begin{subequations}
        \begin{align}
            \E\left[R\left(\tau_{i-1},\nu_{i}\right)\middle|\overline{F}_{i-1}D_{i-1} \right] &=\E\left[R\left(\tau_{i-1},\nu_{i}\right)\middle|F_{i}\overline{F}_{i-1}D_{i-1}\right]\mathbb{P}\left(F_{i}\middle|\overline{F}_{i-1}D_{i-1}\right)\\
            &\hspace{-20pt}+\E\left[R\left(\tau_{i-1},\nu_{i}\right)\middle|\overline{F}_{i}\overline{F}_{i-1}D_{i-1}\right]\mathbb{P}\left(\overline{F}_{i}\middle|\overline{F}_{i-1}D_{i-1}\right)\\
            &\hspace{-20pt}\leq T\cdot \mathbb{P}\left(F_{i}\middle|\overline{F}_{i-1}D_{i-1}\right)+\E\left[R\left(\tau_{i-1},\nu_{i}\right)\middle|\overline{F}_{i}\overline{F}_{i-1}D_{i-1}\right]. \label{eqn:regret_seg}
        \end{align}
    \end{subequations}
    Now, define $N_{k}\left(t_{1},t_{2}\right) :=\sum^{t_{2}}_{t=t_{1}}\1_{\left\{ A_{t}=k \right\}}$ to be the number of times that arm $k$ is selected by \Algref{alg:main_alg} from $t_{1}$ to $t_{2}$. Note that  
    \begin{equation}
        \E\left[R\left(\tau_{i-1},\nu_{i}\right)\middle|\overline{F}_{i}\overline{F}_{i-1}D_{i-1}\right]=\sum_{\Delta^{\left(i\right)}_{k}>0}\Delta^{\left(i\right)}_{k}\cdot \E\left[N_{k}\left(\tau_{i-1},\nu_{i}\right)\middle|\overline{F}_{i}\overline{F}_{i-1}D_{i-1}\right].  \label{eqn:regret_decomposition}
    \end{equation}
    To bound the second term of \eqref{eqn:regret_seg}, we further bound $N_{k}\left(\tau_{i-1},\nu_{i}\right)$ as follows,
    \begin{subequations}
        \begin{align}
            N_{k}\left(\tau_{i-1},\nu_{i}\right)=&\sum^{\nu_{i}}_{t=\tau_{i-1}+1}\1_{\left\{ A_{t}=k, \tau_{i}>\nu_{i}, N_{k}\left(\tau_{i-1},\nu_{i}\right)<l\right\}}+\sum^{\nu_{i}}_{t=\tau_{i-1}+1}\1_{\left\{ A_{t}=k, \tau_{i}>\nu_{i}, N_{k}\left(\tau_{i-1},\nu_{i}\right)\geq l\right\}}\\
            \leq &l+N_{DE,k}\left(\nu_{i}-\tau_{i-1}\right)+\sum^{\nu_{i}}_{t=\tau_{i-1}+1}\1_{\left\{ k=\argmax_{k\in\gK} \mathrm{UCB}_{k\in\gK}, \tau_{i}>\nu_{i}, N_{k}\left(\tau_{i-1},\nu_{i}\right)\geq l\right\}}\\
            \leq &l+\frac{2\alpha\sqrt{\nu_{i}-\tau_{i-1}}}{K}+\frac{3}{2}+\sum^{\nu_{i}}_{t=\tau_{i-1}+1}\1_{\left\{ k=\argmax_{k\in\gK} \mathrm{UCB}_{k\in\gK}, \tau_{i}>\nu_{i}, N_{k}\left(\tau_{i-1},\nu_{i}\right)\geq l\right\}}\label{eqn:N4}\\
            \leq &l+\frac{2\alpha\sqrt{s_{i}}}{K}+\frac{3}{2}+\sum^{\nu_{i}}_{t=\tau_{i-1}+1}\1_{\left\{ k=\argmax_{k\in\gK} \mathrm{UCB}_{k\in\gK}, \tau_{i}>\nu_{i}, N_{k}\left(\tau_{i-1},\nu_{i}\right)\geq l\right\}},
        \end{align}
    \end{subequations}
    \Eqref{eqn:N4} follows from Lemma \ref{applemma:de_regret}.
    Setting $l=\left\lceil 8\log T/\left(\Delta^{(i)}_{k}\right)^{2} \right\rceil $, and following the same steps as in the proof of Theorem 1 of \cite{auer2002finite}, we arrive at
    \begin{equation}
        \E\left[N_{k}\left(\tau_{i-1},\nu_{i}\right)\middle|\overline{F}_{i}\overline{F}_{i-1}D_{i-1}\right]\leq \frac{2\alpha\sqrt{s_{i}}}{K}+\frac{8\log T}{\left(\Delta^{(i)}_{k}\right)^{2}}+\frac{5}{2}+\frac{\pi^{2}}{3}+K.\label{eqn:N_one_seg}
    \end{equation}
    Putting everything together completes the proof.
\end{proof}

The term $\tilde{C}_{i}$ in equation~\ref{eqn:one_seg_regret} involves  $\Delta_k^{(i)}$. In what follows, we establish an upper bound of $\tilde{C_{i}}$ that avoids this dependence.
\begin{lemma}[Regret bound without involving $\Delta_k^{(i)}$]\label{applemma:rewrite_C}
    Consider a stationary bandit interval with $\nu_{i-1}<\tau_{i-1}<\nu_{i}$. Condition on the successful detection events $\overline{F}_{i-1}$ and $D_{i-1}$, and let $\Delta>0$, the expected regret accumulated during $\left(\tau_{i-1},\nu_{i}\right)$ can be bounded by
    \begin{equation}
        \E\left[R\left(\tau_{i-1},\nu_{i}\right)\middle|\overline{F}_{i-1}D_{i-1} \right]\leq 2\alpha\sqrt{s_{i}}+\frac{8K\log{T}}{\Delta}+\Delta\cdot s_{i}+\left(\frac{5}{2}+\frac{\pi^{2}}{3}+K\right)K+T\cdot \mathbb{P}\left(F_{i}\middle|\overline{F}_{i-1}D_{i-1}\right).\label{eqn:one_seg_regret_no_gap}
    \end{equation}
    Moreover, by selecting $\Delta=\sqrt{MK\log{T}/T}$, we have
    \begin{equation}
        \E\left[R\left(\tau_{i-1},\nu_{i}\right)\middle|\overline{F}_{i-1}D_{i-1}\right]
        \leq 2\alpha\sqrt{s_{i}}+8\sqrt{\frac{KT\log{T}}{M}}+\sqrt{\frac{MK\log{T}}{T}} s_{i}+\left(\frac{5}{2}+\frac{\pi^{2}}{3}+K\right)K+T\cdot \mathbb{P}\left(F_{i}\middle|\overline{F}_{i-1}D_{i-1}\right).
    \end{equation}
\end{lemma}


\begin{proof}
    We rewite the equation~\ref{eqn:regret_decomposition} as follows,
    \begin{subequations}
        \begin{align}
            \E\left[R\left(\tau_{i-1},\nu_{i}\right)\middle|\overline{F}_{i}\overline{F}_{i-1}D_{i-1}\right]&=\sum_{\Delta^{\left(i\right)}_{k}>0}\Delta^{\left(i\right)}_{k}\cdot \E\left[N_{k}\left(\tau_{i-1},\nu_{i}\right)\middle|\overline{F}_{i}\overline{F}_{i-1}D_{i-1}\right]\\
            &=\sum_{\Delta^{\left(i\right)}_{k}\geq\Delta}\Delta^{\left(i\right)}_{k}\cdot \E\left[N_{k}\left(\tau_{i-1},\nu_{i}\right)\middle|\overline{F}_{i}\overline{F}_{i-1}D_{i-1}\right]\label{eqn:decomp1}\\
            +\sum_{\Delta^{\left(i\right)}_{k}<\Delta}\Delta^{\left(i\right)}_{k}\cdot \E\left[N_{k}\left(\tau_{i-1},\nu_{i}\right)\middle|\overline{F}_{i}\overline{F}_{i-1}D_{i-1}\right]\hspace{-8.5cm}\label{eqn:decomp2}\\
            &\leq \sum_{\Delta^{\left(i\right)}_{k}\geq\Delta}\left[\frac{2\alpha\sqrt{s_{i}}}{K}\cdot\Delta^{(i)}_{k}+\frac{8\log{T}}{\Delta^{\left(i\right)}_{k}}+\left(\frac{5}{2}+\frac{\pi^{2}}{3}+K\right)\cdot\Delta^{(i)}_{k}\right]+\Delta\cdot s_{i}\label{eqn:decomp3}\\
            &\leq 2\alpha\sqrt{s_{i}}+\sum_{\Delta^{(i)}_{k}\geq \Delta}\frac{8\log{T}}{\Delta^{\left(i\right)}_{k}}+\left(\frac{5}{2}+\frac{\pi^{2}}{3}+K\right)\sum_{\Delta^{(i)}_{k}\geq \Delta}\Delta^{(i)}_{k}+\Delta\cdot s_{i}\\
            &\leq 2\alpha\sqrt{s_{i}}+\frac{8K\log{T}}{\Delta}+\left(\frac{5}{2}+\frac{\pi^{2}}{3}+K\right)K+\Delta\cdot s_{i}. \label{eqn:to_take_Delta}
        \end{align}
    \end{subequations}
    Equations~\ref{eqn:decomp1} to~\ref{eqn:decomp3} are derived using equation~\ref{eqn:N_one_seg}, while the transition from~\ref{eqn:decomp2} to~\ref{eqn:decomp3} leverages the inequalities $\Delta > \Delta^{(i)}_k$ and $s_{i} > \mathbb{E}\left[N_{k}\left(\tau_{i-1},\nu_{i}\right)\middle|\overline{F}_{i}\overline{F}_{i-1}D_{i-1}\right]$. 
    
    The second part of the lemma follows straightforwardly by plugging $\Delta=\sqrt{MK\log{T}/T}$ into then the equation~\ref{eqn:to_take_Delta}.
\end{proof}

Theorem \ref{appthm:std_regret_mucb} can then be proved by recursively applying Lemma~\ref{applemma:regret_stat}.
\begin{theorem}[Standard Regret bound of M-UCB]\label{appthm:std_regret_mucb}
    Combining Algorithms~\ref{alg:main_alg} and~\ref{alg:mucb_cd} with the parameters in Equation~\ref{eqn:w}, and Equation~\ref{eqn:b} achieves the expected regret upper bound as follows:
    \begin{multline} 
     \E\left[R\left(1,T\right)\right]\leq \underbrace{9\sqrt{MKT\log{T}}+\left(\frac{5}{2}+\frac{\pi^{2}}{3}+K\right)MK}_{(a)}+\underbrace{2\alpha\sqrt{MT}}_{(b)}
    +\underbrace{w\left(\frac{K}{2\alpha}+1\right)\sqrt{M\left(T+M\right)}}_{(c)}\\
    +\underbrace{\frac{w^{2}M}{4}\left(\frac{K}{2\alpha}+1\right)^{2}}_{(c)}+\underbrace{2M}_{(d)}.
    \end{multline}\label{eqn:regret_bound_MUCB_p}
    By setting $\alpha = c\sqrt{K\log{\left(KT\right)}}$ for some constant $c$, the expected regret is upper-bounded by $\mathcal{O}(\sqrt{KMT\log{T}})$.
\end{theorem}
\begin{proof}
    Recall that $R\left(r,s\right)=\sum^{s}_{t=r}\max_{k\in\gK}\E\left[X_{k,t}\right]-X_{A_{t},t}$, then  $\gR\left(T\right)=\E\left[R\left(1,T\right)\right]$. We have
    \begin{subequations}
        \label{10}
        \begin{align}
            \gR\left(T\right)&=\E\left[R\left(1,T\right)\right]\\
            &=\E\left[R\left(1,T\right)\middle|\overline{F}_{0}D_{0}\right] \label{eqn:R_init}\\
            &\leq \E\left[R\left(1,\nu_{1}\right)\middle|\overline{F}_{1}\overline{F}_{0}D_{0}\right]+\E\left[R\left(\nu_{1},T\right)\middle|\overline{F}_{1}\overline{F}_{0}D_{0}\right]+ T\cdot \mathbb{P}\left(F_{1}\middle| \overline{F}_{0}D_{0}\right)\label{eqn:R_conE1}\\
            &\leq  \tilde{C}_{1}+2\alpha\sqrt{\left(\nu_{1}-\nu_{0}\right)}+\E\left[R\left(\nu_{1},T\right)\middle| \overline{F}_{1}\overline{F}_{0}D_{0}\right]+T\cdot \mathbb{P}\left(F_{1}\middle| \overline{F}_{0}D_{0}\right),\label{eqn:R_conE2}
        \end{align}
    \end{subequations}
    where \eqref{eqn:R_init} holds because $\tau_{0}=0$, \eqref{eqn:R_conE1} is due to the law of total expectation and some trivial bounds, and
    \eqref{eqn:R_conE2} follows from Lemmas~\ref{applemma:regret_stat}. The third term in \eqref{eqn:R_conE2} is then further bounded as follows:
    \begin{subequations}
        \begin{align}
            \hspace{-20pt}\E\left[R\left(\nu_{1},T\right)\middle| \overline{F}_{1}\overline{F}_{0}D_{0}\right] 
            &\leq\E\left[R\left(\nu_{1},T\right)\middle| D_{1}\overline{F}_{1}\overline{F}_{0}D_{0}\right]+T\cdot \left(1-\mathbb{P}\left(D_{1}\middle| \overline{F}_{1}\overline{F}_{0}D_{0}\right)\right)\label{eqn:R_conF1}\\
            &\hspace{-60pt}\leq \E\left[R\left(\nu_{1},T\right)\middle| D_{1}\overline{F}_{1}\overline{F}_{0}D_{0}\right] + T\cdot\mathbb{P}\left(\overline{D}_{1}\middle| \overline{F}_{1}\overline{F}_{0}D_{0}\right)\label{eqn:R_conF2}\\     
            &\hspace{-60pt}= \E\left[R\left(\tau_{1},T\right)\middle| D_{1}\overline{F}_{1}\overline{F}_{0}D_{0}\right]+\E\left[R\left(\nu_{1},\tau_{1}\right)\middle| D_{1}\overline{F}_{1}\overline{F}_{0}D_{0}\right]+T\cdot\mathbb{P}\left(\overline{D}_{1}\middle| \overline{F}_{1}\overline{F}_{0}D_{0}\right)\label{eqn:reg_split1}\\
            &\hspace{-60pt}\leq\E\left[R\left(\tau_{1},T\right)\middle| \overline{F}_{1}D_{1}\right] + \E\left[\tau_{1}-\nu_{1}\middle| \overline{F}_{1}D_{1}\overline{F}_{0}D_{0}\right]\label{eqn:reg_split2}+T\cdot\mathbb{P}\left(\overline{D}_{1}\middle| \overline{F}_{1}\overline{F}_{0}D_{0}\right)\\
            &\hspace{-60pt}\leq\E\left[R\left(\tau_{1},T\right)\middle| \overline{F}_{1}D_{1}\right] + \E\left[\tau_{1}-\nu_{1}\middle| \overline{F}_{1}D_{1}\right]+T\cdot\mathbb{P}\left(\overline{D}_{1}\middle| \overline{F}_{1}\overline{F}_{0}D_{0}\right)\label{eqn:reg_split3},
        \end{align}
    \end{subequations}
    where \eqref{eqn:R_conF1} applies the law of total expectation and some trivial bounds.%
    From here, we can set up the following recursion:
    \begin{subequations}
        \begin{align}
            &\E\left[R\left(1,T\right)\right] =\E\left[R\left(1,T\right)\middle|\overline{F}_{0}D_{0}\right] \\
            & \leq \E\left[R\left(\tau_{1},T\right)\middle|\overline{F}_{1}D_{1}\right]+\tilde{C}_{1}+2\alpha\sqrt{s_{1}-1}+\E\left[\tau_{1}-\nu_{1}\middle| \overline{F}_{1}D_{1}\right]\\
            &\hspace{+60pt}+T\cdot\mathbb{P}\left(F_{1}\middle| \overline{F}_{0}D_{0}\right)+T\cdot\mathbb{P}\left(\overline{D}_{1}\middle| \overline{F}_{1}\overline{F}_{0}D_{0}\right) \\
            & \leq \E\left[R\left(\tau_{2},T\right)\middle|\overline{F}_{2}D_{2}\right]+\sum^{2}_{i=1}\tilde{C}_{i}+2\alpha\sum^{2}_{i=1}\sqrt{s_{i}-1}\\ 
            &+\sum^{2}_{i=1}\E\left[\tau_{i}-\nu_{i}\middle| \overline{F}_{i-1}D_{i-1}\right]+T\sum^{2}_{i=1}\mathbb{P}\left(F_{i}\middle| \overline{F}_{i-1}D_{i-1}\right)+T\sum^{2}_{i=1}\mathbb{P}\left(\overline{D}_{i}\middle| \overline{F}_{i}\overline{F}_{i-1}D_{i-1}\right) \\
            &\hspace{+150pt}\vdots\nonumber\\
            & \leq \sum^{M}_{i=1}\tilde{C}_{i}+2\alpha\sum^{M}_{i=1}\sqrt{s_{i}}+\sum^{M-1}_{i=1}\E\left[\tau_{i}-\nu_{i}\middle| \overline{F}_{i-1}D_{i-1}\right]\\
            &\hspace{+60pt}+T\sum^{M}_{i=1}\mathbb{P}\left(F_{i}\middle| \overline{F}_{i-1}D_{i-1}\right)+T\sum^{M-1}_{i=1}\mathbb{P}\left(\overline{D}_{i}\middle| \overline{F}_{i}\overline{F}_{i-1}D_{i-1}\right)\\
            & \leq \sum^{M}_{i=1}\tilde{C}_{i}+2\alpha\sqrt{MT}+\sum^{M-1}_{i=1}\E\left[\tau_{i}-\nu_{i}\middle| \overline{F}_{i-1}D_{i-1}\right]\label{eqn:R_segment}\\
            &\hspace{+60pt}+T\sum^{M}_{i=1}\mathbb{P}\left(F_{i}\middle| \overline{F}_{i-1}D_{i-1}\right)+T\sum^{M-1}_{i=1}\mathbb{P}\left(\overline{D}_{i}\middle| \overline{F}_{i}\overline{F}_{i-1}D_{i-1}\right)
        \end{align}
    \end{subequations}
    where \eqref{eqn:R_segment} follows from the Cauchy–Schwarz inequality
    \begin{subequations}
        \begin{align}
            \left(\sum^{M}_{i=1}\sqrt{s_{i}}\right)^{2}\leq \left(\sum^{M}_{i=1}s_{i}\right)\left(\sum^{M}_{i=1}1\right)= M\sum^{M}_{i=1}s_{i}= MT,
        \end{align}
    \end{subequations}
    Finally, by lemma~\ref{applemma:rewrite_C}, we know that 
    \begin{subequations}
        \begin{align}
            \sum^{M}_{i}\tilde{C}_{i}&\leq \sum^{M}_{i=1}\left[8\sqrt{\frac{KT\log{T}}{M}}+\left(\frac{5}{2}+\frac{\pi^{2}}{3}+K\right)K+\sqrt{\frac{MK\log{T}}{T}}\cdot s_{i}\right]\\
            &= 8\sqrt{MKT\log{T}}+\left(\frac{5}{2}+\frac{\pi^{2}}{3}+K\right)MK+\sqrt{\frac{MK\log{T}}{T}}\cdot\sum^{M}_{i=1}s_{i}. \label{eqn:to choose Delta}\\
            &=8\sqrt{MKT\log{T}}+\left(\frac{5}{2}+\frac{\pi^{2}}{3}+K\right)MK+\sqrt{\frac{MK\log{T}}{T}}\cdot T.\\
            &=8\sqrt{MKT\log{T}}+\left(\frac{5}{2}+\frac{\pi^{2}}{3}+K\right)MK+\sqrt{MKT\log{T}}\\
            &=9\sqrt{MKT\log{T}}+\left(\frac{5}{2}+\frac{\pi^{2}}{3}+K\right)MK\label{eqn:sum_C}.
        \end{align}
    \end{subequations}
    Then, plugging equation~\ref{eqn:sum_C} into equation~\ref{eqn:R_segment} completes the proof of Theorem~\ref{appthm:std_regret_mucb}.
\end{proof}

\subsection{Proof of General Case}

To handle the temporal dependence induced by the Markovian dynamics, we leverage a concentration inequality for ergodic Markov chains.
Assumption~\ref{appassump:ergodic} ensures the existence of unique stationary distributions for each arm in each segment, allowing us to apply the following Hoeffding-type bound for Markovian rewards~\cite{Hoeffding_bound_Markovian_reward}, which plays a key role in bounding the probabilities of false alarms $\mathbb{P}\left(F_{i}\middle|\overline{F}_{i-1}D_{i-1}\right)$ and successful detections $\mathbb{P}\left(D_{i}\middle|\overline{F}_{i}\overline{F}_{i-1}D_{i-1}\right)$.

\begin{assumption}\label{appassump:ergodic}
All Markov chains $\mathcal{M}_k^{(i)} = (P_k^{(i)},R_k^{(i)})$ are ergodic. 
This ensures the existence of a unique steady-state distribution $d_k^{(i)}$ for 
each arm $k$ in segment $i$, satisfying 
$d_k^{(i)} = d_k^{(i)} P_k^{(i)}$ with $\sum_{s \in \mathcal{S}} d_k^{(i)}(s) = 1$. 
We then define the steady-state arm mean as 
$\bar{\mu}_k^{(i)} := d_k^{(i)} R_k^{(i)}$.
\end{assumption}

\begin{lemma}[Hoeffding Bound for Markovian Reward \cite{Hoeffding_bound_Markovian_reward}]
    \label{applem:Hoeffding bound for Markov Chain}
    Let $\mathcal{M}$ be an ergodic Markov chain with state space $\mathcal{S}$ and stationary distribution $d$. Let $L=T(\epsilon)$ be its $\varepsilon$-mixing time for $\varepsilon \leq 1/8$. Let $(V_1, \cdots, V_t)$ denote a $t$-step random walk on $\mathcal{M}$ starting from an initial distribution $\phi$. Let $f_\ell: \mathcal{S} \mapsto [0,1]$ be a weight function at $\ell$ time step such that expected weighted $\mathbb{E}_{v\leftarrow d}=\mu$ for all $\ell$. Define total weight walk $X\triangleq \sum^{t}_{\ell=1} f_\ell(V_\ell)$. There exists some constant $c$ (which is independent to $d, \delta, \epsilon$) and let $0\leq \delta\leq 1$ such that
    \begin{equation}
        \begin{aligned}
            \mathbb{P}[X\geq\left(1+\delta\right)\mu t]&\leq c {\lVert\phi\rVert}_{d} \exp{\left(-\delta^2\mu t/ \left(72L\right)\right)}\\
            \mathbb{P}[X\leq\left(1-\delta\right)\mu t]&\leq c {\lVert\phi\rVert}_{d} \exp{\left(-\delta^2\mu t/ \left(72L\right)\right)}.
        \end{aligned}
    \end{equation}
\end{lemma}

Using Lemma~\ref{applem:Hoeffding bound for Markov Chain}, we now control the probability of false alarms within each stationary segment. 
The following lemma shows that, with appropriate choices of the detection window and threshold parameters, 
the probability of triggering a false alarm is negligible.

\begin{lemma}[Probability of false alarm]
    \label{applem:FA}
    Under Algorithm~\ref{alg:mucb_cd} with the parameter choices specified in Equations~\ref{eqn:w_gen} and~\ref{eqn:b_gen}, the probability of a false alarm in each stationary segment satisfies
    \begin{equation}    
        \mathbb{P}\!\left(F_i \,\middle|\, \overline{F}_{i-1} D_{i-1}\right)
        \;\leq\; \frac{\epsilon \lVert \varphi \rVert_{d} + 1}{T},
    \end{equation}
    where $\epsilon > 0$ is an arbitrarily small constant, $\varphi$ is the initial state distribution, and $\lVert \varphi \rVert_{d}$ is the corresponding $d$-norm.
\end{lemma}

\begin{proof}
We formalize this analysis by defining an equivalent function $S_{k,t}$ and establishing the stopping conditions necessary for detecting such events. The function $S_{k,t}$ serves as a proxy for the behavior of Algorithm~\ref{alg:mucb_cd} and is defined as follows:
\begin{equation}
    S_{k,t}=\sum_{\ell=w/2+1}^{w}Z_{k,\ell}-\sum_{\ell=1}^{w/2}Z_{k,\ell}, \quad\text{for any}\quad k\in K, t\geq w,
\end{equation}
where $w$ is window size.  

To formalize the conditions for triggering a change detection, we define the stopping time: $\lvert S_{k,t}\rvert\geq b$, which is equalivent to the change alerts in Algorithm~\ref{alg:mucb_cd}. To facilitate analysis, we divide the stopping time events into $w$ subsequences indexed by $j$ with each subsequence corresponding to a distinct offset. The stopping time for the $j$-th subsequence is given by:
\begin{equation}
    \tau_{k, i}^{(j)}:=\inf \left\{t=\tau_{i-1}+j+n w, n \in \mathbb{Z}^{+}:\left|S_{k, t}\right|>b\right\},
\end{equation}
where $b$ is the detection threshold. By dividing the stopping time events into these $w$ subsequences, we ensure that the samples within each subsequences are independent over $w$-time steps. This independence is achieved because the samples are aggregated over disjoint time intervals defined by the periodic offset $j$. Such independence is crucial for applying concentration inequalities later in the analysis.
Clearly, $\tau_{k,i}=\min \left\{\tau_{k,i}^{(0)}, \ldots, \tau_{k,i}^{(w-1)}\right\}$. Let us define $\xi_{k,i}^{(j)}$
\begin{equation}
    \xi_{k,i}^{(j)}=\frac{\left(\tau_{k,i}^{(j)}-j-\tau_{k,i-1}\right)}{w}, \quad \text{for any}\quad 0 \leq j \leq w-1. 
\end{equation}

Note that condition on the events $D_{i-1}$ and $F_{i-1}, \xi_{k,i}^{(j)}$ is a geometric random variable with parameter $p:=\mathbb{P}\left(\left|S_{k, t}\right|>b\right)$, because when fixing $j$, there is no overlap between the samples in the current window and the next.
\begin{subequations}
    \begin{align}
        \mathbb{P}\left(\tau_{k,i}^{(j)}=\tau_{i-1}+n w+j \middle| \bar{F}_{i-1} D_{i-1},\sum_{t=\tau_i+1}^{\nu_i}\1_{A_t=k}\geq w\right)\nonumber
        = \mathbb{P}\left(\xi_{k,i}=n \middle| \bar{F}_{i-1} D_{i-1}, \sum_{t=\tau_i+1}^{\nu_i}\1_{A_t=k}\geq w\right)\nonumber=p(1-p)^{n-1}.
    \end{align}
\end{subequations}

Here, the inclusion of subsequent events as conditions should not impact the results, as when entering the change detection algorithm, those events have already occurred. Moreover, by union bound, we have that for any $i \in \mathbb{N}$,

\begin{subequations}
    \begin{align}
        \mathbb{P}\left(\tau_{k,i} \leq \nu_i \mid \bar{F}_{i-1} D_{i-1}\right) 
        & \leq w\left(1-(1-p)^{\left\lfloor\left(\nu_i-\tau_i-1\right) / w\right\rfloor}\right)\\
        & \leq w\left(1-(1-p)^{\lfloor T / w\rfloor}\right).
    \end{align}    
\end{subequations}

We further upper bound the probability \( p_k \) for each arm \( k \) by decomposing 
the reward into two parts:  
(i) a \emph{state-dependent mean component}, determined by the Markovian state of the arm; and  
(ii) an \emph{independent noise term} \( n_{t} \in [-1,1] \) with zero mean. 

The key idea is to separately bound the effects of the \emph{distributional fluctuations} 
induced by the Markov chain and the \emph{reward noise}.  
Accordingly, we split the detection threshold into two components, \( b = b_d + b_n \), 
where \( b_d \) handles deviations due to distributional shifts and \( b_n \) accounts for noise fluctuations. 

Let \( V_{k,t} \) denote the state of arm \( k \) at time \( t \), and let 
\( f_{k,t}: \mathcal{S} \mapsto [0,1] \) be the corresponding reward mean function.  
Then the probability of the detection statistic exceeding the threshold can be bounded as follows:
\begin{subequations}
    \begin{align}
        p
        &= \mathbb{P}\left[\left\lvert S_{k,\ell}\right\rvert \geq b \right]\label{eqn:p1}\\
        &= \mathbb{P}\left[\left\lvert 
        \sum_{\ell = w/2+1}^{w} \left(f_{k,\ell}\left(V_{k,\ell}\right) + n_{\ell}\right)
        - \sum_{\ell = 1}^{w/2} \left(f_{k,\ell}\left(V_{k,\ell}\right) + n_{\ell}\right)
        \right\rvert \geq b_d + b_n \right] \label{eqn:p2}\\
        &\leq 1 - 
        \mathbb{P}\left[\left\lvert 
        \sum_{\ell = w/2+1}^{w} f_{k,\ell}(V_{k,\ell}) 
        - \sum_{\ell = 1}^{w/2} f_{k,\ell}(V_{k,\ell})
        \right\rvert \leq b_d \right]
        \cdot
        \mathbb{P}\left[\left\lvert 
        \sum_{\ell = 1}^{w} n_{\ell} 
        \right\rvert \leq b_n \right]. \label{eqn:p}
    \end{align}
\end{subequations}

To apply the concentration inequality, we focus on the case that induces the largest deviation from the steady-state distribution. This occurs when the rewards within the detection window are collected consecutively in time, which leads to the strongest temporal dependence among samples. We then derive the concentration probability of this distributional deviation. 
To fit the requirement of Lemma~\ref{applem:Hoeffding bound for Markov Chain}, 
we shift the original reward function $f_\ell$ to a new function $f^{\prime}_{k,\ell}$ such that $f_{k,\ell}^{\prime}:\mathcal{S}\mapsto[0,1]$:
\begin{align}
        f^{\prime}_{k,\ell}=
        \begin{cases}
        \frac{1-f_{k,\ell}}{2} & 1\leq \ell \leq \frac{w}{2}, \\
        \frac{1+f_{k,\ell}}{2} & \frac{w}{2}+1\leq \ell \leq w. \label{eqn:f_map}
        \end{cases}
\end{align}

We first focus on the first probability term in Equation~\ref{eqn:p}.
By substituting Equation~\ref{eqn:f_map} into it, we obtain

\begin{subequations}
    \begin{align}
        &\mathbb{P}\left[\left\lvert 
        \sum_{\ell = w/2+1}^{w} f_{k,\ell}(V_{k,\ell}) 
        - \sum_{\ell = 1}^{w/2} f_{k,\ell}(V_{k,\ell})
        \right\rvert \leq b_d \right]\label{eqn:bound_f1} \\
        &= \mathbb{P}\left[\left\lvert 
        2\sum_{\ell = w/2+1}^{w} f^{\prime}_{k,\ell}(V_{k,\ell}) 
        + 2\sum_{\ell = 1}^{w/2} f^{\prime}_{k,\ell}(V_{k,\ell})
        - w
        \right\rvert \leq b_d \right]\label{eqn:bound_f2} \\
        &= \mathbb{P}\left[\left\lvert 
        \sum_{\ell = 1}^{w} f^{\prime}_{k,\ell}(V_{k,\ell}) 
        - \frac{w}{2}
        \right\rvert \leq \frac{b_d}{2} \right]\label{eqn:bound_f3} \\
        &= 1 - 
        \mathbb{P}\left[\sum_{\ell = 1}^{w} f^{\prime}_{k,\ell}(V_{k,\ell}) 
        - \frac{w}{2} > \frac{b_d}{2} \right]
        - 
        \mathbb{P}\left[\sum_{\ell = 1}^{w} f^{\prime}_{k,\ell}(V_{k,\ell}) 
        - \frac{w}{2} < -\frac{b_d}{2} \right]. \label{eqn:bound_f4}\\
        &= 1 - 
        \mathbb{P}\left[\sum_{\ell = 1}^{w} f^{\prime}_{k,\ell}(V_{k,\ell}) 
        > \frac{w}{2} + \frac{b_d}{2} \right]
        - 
        \mathbb{P}\left[\sum_{\ell = 1}^{w} f^{\prime}_{k,\ell}(V_{k,\ell}) 
        < \frac{w}{2} - \frac{b_d}{2} \right]. \label{eqn:bound_f5}
    \end{align}
\end{subequations}
Specifically, Equation~\ref{eqn:bound_f2} follows from the definition of $f^{\prime}_{k,\ell}$ in Equation~\ref{eqn:f_map}.
Next, by simple algebraic rearrangement, we obtain Equation~\ref{eqn:bound_f3}.
Equation~\ref{eqn:bound_f4} is derived by decomposing the absolute value event into two separate one-sided probability terms, and Equation~\ref{eqn:bound_f5} then follows from further rearranging these terms.

By applying Lemma~\ref{applem:Hoeffding bound for Markov Chain} to the two probability terms above, we obtain
\begin{subequations}
    \begin{align}
        &\mathbb{P}\left[\sum_{\ell = 1}^{w} f^{\prime}_{k,\ell}(V_{k,\ell}) 
        > \frac{w}{2} + \frac{b_d}{2} \right] \\
        &\leq \epsilon \lVert \varphi \rVert_{d} \exp\left(-\frac{b_d^2}{144 L_{k}}\right), \\
        &\mathbb{P}\left[\sum_{\ell = 1}^{w} f^{\prime}_{k,\ell}(V_{k,\ell}) 
        < \frac {w}{2} - \frac{b_d}{2} \right] \\
        &\leq \epsilon \lVert \varphi \rVert_{d} \exp\left(-\frac{b_d^2}{144 L_{k}}\right).
    \end{align}
\end{subequations}
Combining the above two inequalities, we have
\begin{equation}
    \mathbb{P}\left[\left\lvert 
    \sum_{\ell = w/2+1}^{w} f_{k,\ell}(V_{k,\ell}) 
    - \sum_{\ell = 1}^{w/2} f_{k,\ell}(V_{k,\ell})
    \right\rvert \leq b_d \right]
    \geq 1 - 2 \epsilon \lVert \varphi \rVert_{d} \exp\left(-\frac{b_d^2}{144 L_{k}}\right),
\end{equation}
where \( L_{k} \) is the mixing time of arm \( k \). 

If we set \( b_d = \sqrt{144 L \log\left(2KT^2\lVert \varphi \rVert_{d}^{-1}\right)} \), where \( L = \max_{k} L_{k} \), then we have
\begin{equation}
    \mathbb{P}\left[\left\lvert 
    \sum_{\ell = w/2+1}^{w} f_{k,\ell}(V_{k,\ell}) 
    - \sum_{\ell = 1}^{w/2} f_{k,\ell}(V_{k,\ell})
    \right\rvert \leq b_d \right]
    \geq 1 - \frac{\epsilon \lVert \varphi \rVert_{d}}{KT^2}.
\end{equation}
Next, we bound the second probability term in Equation~\ref{eqn:p}. Since the noise terms \( n_{\ell} \) are independent and bounded within \([-1,1]\), we can apply McDiarmid's inequality to obtain:
\begin{equation}
    \mathbb{P}\left[\left\lvert 
    \sum_{\ell = 1}^{w} n_{\ell} 
    \right\rvert \leq b_n \right]
    \geq 1 - \mathbb{P}\left[\left\lvert 
    \sum_{\ell = 1}^{w} n_{\ell} 
    \right\rvert > b_n \right]
    \geq 1 - \mathbb{P}\left[\sum_{\ell = 1}^{w} n_{\ell} > b_n \right] - \mathbb{P}\left[\sum_{\ell = 1}^{w} n_{\ell} < -b_n \right]
    \geq 1 - 2\exp\left(-\frac{2b_n^2}{w}\right).
\end{equation}

By setting $b_{n}=\sqrt{w\log\left(2KT^2\right)/2}$, we have
\begin{equation}
    \mathbb{P}\left[\left\lvert 
    \sum_{\ell = 1}^{w} n_{\ell} 
    \right\rvert \leq b_n \right]
    \geq 1 - \frac{1}{KT^2}.  \label{eqn:n_bound}
\end{equation}
Combining the bounds for both probability terms, we obtain
\begin{equation}
    p
    \leq 1 - \left(1 - \frac{\epsilon \lVert \varphi \rVert_{d}}{KT^2}\right) \left(1 - \frac{1}{KT^2}\right)
    \leq \frac{\epsilon \lVert \varphi \rVert_{d} + 1}{KT^2}.
\end{equation}
Finally, we can bound the probability of false alarm as follows:
\begin{subequations}
    \begin{align}
        \mathbb{P}\left(F_i \middle| \overline{F}_{i-1} D_{i-1}\right) &\leq \sum_{k=1}^{K} \mathbb{P}\left(\tau_{k,i} \leq \nu_i \mid \overline{F}_{i-1} D_{i-1}\right)\label{eqn:false_alarm_union_1}\\
        &\leq \sum_{k=1}^{K} w\left(1-\left(1-\frac{\epsilon \lVert \varphi \rVert_{d} + 1}{KT^2}\right)^{\lfloor T / w\rfloor}\right)\label{eqn:false_alarm_union_2}\\
        &\leq \sum_{k=1}^{K} w \cdot \frac{\epsilon \lVert \varphi \rVert_{d} + 1}{KT^2} \cdot \lfloor T / w\rfloor\label{eqn:false_alarm_union_3}\\
        &\leq \frac{\epsilon \lVert \varphi \rVert_{d} + 1}{T}.\label{eqn:false_alarm_union_4}
    \end{align}
\end{subequations}
Here, \eqref{eqn:false_alarm_union_1} follows from the union bound; \eqref{eqn:false_alarm_union_2} follows from the earlier derived bound on \( p \); \eqref{eqn:false_alarm_union_3} uses the inequality \( (1 - x)^n \geq 1 - nx \) for \( x \in [0,1] \) and \( n \geq 1 \); and \eqref{eqn:false_alarm_union_4} simplifies the expression to yield the final result.
\end{proof}

We next use the same concentration tool to establish that, with high probability, each true change is successfully detected within a sublinear delay. 

\begin{lemma}[Probability of successful detection]
    \label{applem:SD}
    Under Algorithm~\ref{alg:mucb_cd} with the parameter choices specified in \eqref{eqn:w_gen} and \eqref{eqn:b_gen}, and with the delay threshold
    \begin{equation}
        h_{i} = \left\lceil 
        w\left(\frac{K}{2\alpha}+1\right)\sqrt{s_i+1}
        +\frac{w^2}{4}\left(\frac{K}{2\alpha}+1\right)^2
        \right\rceil,
    \end{equation}
    for any $i \ge 1$, the probability of successful detection satisfies
    \begin{equation}
        \mathbb{P}\!\left(D_i \,\middle|\, \overline{F}_{i}\overline{F}_{i-1} D_{i-1}\right)
        \;\ge\; 1 - \frac{\epsilon \lVert \varphi \rVert_{d}}{T},
    \end{equation}
    where $\epsilon > 0$ is an arbitrarily small constant, $\varphi$ is the initial state distribution, and $d$ is the stationary distribution.
\end{lemma}

\begin{proof}
\begin{subequations}
    \begin{align}
        \mathbb{P}\left(D_i \middle| \overline{F}_{i}\overline{F}_{i-1} D_{i-1}\right) &= 1 - \mathbb{P}\left(\tau_i > \nu_i + h_i \middle| \overline{F}_{i}\overline{F}_{i-1} D_{i-1}\right)\label{eqn:sd_1}\\
        &\hspace{-40pt} \geq \max_{t\in\left\{\nu_{i}+1,\ldots,\nu_{i}+h_{i}\right\}} \mathbb{P}\left(\exists k\in\gK, \left|S_{k,t}\right| > b \middle| \overline{F}_{i}\overline{F}_{i-1} D_{i-1}\right)\label{eqn:sd_2}\\
        &\hspace{-40pt} \geq \max_{j\in\left\{0,\ldots,w/2\right\}} \mathbb{P}\left(\left\lvert\sum^{w}_{\ell=j+1}\left(f_{k,\ell}\left(V_{k,\ell}\right)+n_{\ell}\right)-\sum^{j}_{\ell=1}\left(f_{k,\ell}\left(V_{k,\ell}\right)+n_{\ell}\right)\right\rvert > b \middle| \overline{F}_{i}\overline{F}_{i-1} D_{i-1}\right)\label{eqn:sd_3}\\
        &\hspace{-40pt} \geq \mathbb{P}\left(\left\lvert\sum^{w}_{\ell=w/2+1}\left(f_{k,\ell}\left(V_{k,\ell}\right)+n_{\ell}\right)-\sum^{w/2}_{\ell=1}\left(f_{k,\ell}\left(V_{k,\ell}\right)+n_{\ell}\right)\right\rvert > b \middle| \overline{F}_{i}\overline{F}_{i-1} D_{i-1}\right)\label{eqn:sd_4}
    \end{align}
\end{subequations}
Here, \eqref{eqn:sd_1} follows from the definition of successful detection; 
\eqref{eqn:sd_2} follows from the definition of the stopping time; 
\eqref{eqn:sd_3} mirrors the derivation from \eqref{eqn:p1} to \eqref{eqn:p2}; 
and \eqref{eqn:sd_4} follows from noting that the maximum is attained at $j = w/2$ 
and decomposing the absolute value event into two separate probability terms, 
one capturing the distributional deviation and the other the noise fluctuation.

First, we bound the first probability term in \eqref{eqn:sd_4}.
\begin{subequations}
    \begin{align}
        &\mathbb{P}\left(\left\lvert\sum^{w}_{\ell=w/2+1}f_{k,\ell}\left(V_{k,\ell}\right)-\sum^{w/2}_{\ell=1}f_{k,\ell}\left(V_{k,\ell}\right)\right\rvert > b + b_d \middle| \overline{F}_{i}\overline{F}_{i-1} D_{i-1}\right)\label{eqn:sd_bound_f1}\\
        &= \mathbb{P}\left(\sum^{w}_{\ell=w/2+1}f_{k,\ell}\left(V_{k,\ell}\right)-\sum^{w/2}_{\ell=1}f_{k,\ell}\left(V_{k,\ell}\right) > b + b_d \middle| \overline{F}_{i}\overline{F}_{i-1} D_{i-1}\right)\label{eqn:sd_bound_f2}\\
        &\hspace{20pt} + \mathbb{P}\left(\sum^{w/2}_{\ell=1}f_{k,\ell}\left(V_{k,\ell}\right)-\sum^{w}_{\ell=w/2+1}f_{k,\ell}\left(V_{k,\ell}\right) > b + b_d \middle| \overline{F}_{i}\overline{F}_{i-1} D_{i-1}\right)\label{eqn:sd_bound_f3}\\
        &\geq \mathbb{P}\left(\sum^{w}_{\ell=w/2+1}\left(f_{k,\ell}\left(V_{k,\ell}\right)-\bar{\mu}_{k}^{(i)}\right)-\sum^{w/2}_{\ell=1}\left(f_{k,\ell}\left(V_{k,\ell}\right)-\bar{\mu}_{k}^{(i-1)}\right) > b + b_d - \frac{w}{2}\left(\bar{\mu}^{(i)}_{k}-\bar{\mu}^{(i-1)}_{k}\right) \middle| \overline{F}_{i}\overline{F}_{i-1} D_{i-1}\right)\label{eqn:sd_bound_f4}\\
        &\hspace{20pt} + \mathbb{P}\left(\sum^{w/2}_{\ell=1}\left(f_{k,\ell}\left(V_{k,\ell}\right)-\bar{\mu}_{k}^{(i-1)}\right)-\sum^{w}_{\ell=w/2+1}\left(f_{k,\ell}\left(V_{k,\ell}\right)-\bar{\mu}_{k}^{(i)}\right) > b + b_d - \frac{w}{2}\left(\bar{\mu}^{(i-1)}_{k}-\bar{\mu}^{(i)}_{k}\right) \middle| \overline{F}_{i}\overline{F}_{i-1} D_{i-1}\right)\label{eqn:sd_bound_f5}\\
        &\geq \mathbb{P}\left(\bar{\mu}^{(i)}_{k}>\bar{\mu}^{(i-1)}_{k}\right)\mathbb{P}\left(\sum^{w/2}_{\ell=1}\left(f_{k,\ell}\left(V_{k,\ell}\right)-\bar{\mu}_{k}^{(i-1)}\right)< \frac{\left(\bar{\mu}^{(i)}_{k}-\bar{\mu}^{(i-1)}_{k}\right)w}{4}-\frac{\left(b+b_{n}\right)}{2} \middle| \overline{F}_{i}\overline{F}_{i-1} D_{i-1}\right)\label{eqn:sd_bound_f6}\\
        &\hspace{20pt} \cdot \mathbb{P}\left(\sum^{w}_{\ell=w/2+1}\left(f_{k,\ell}\left(V_{k,\ell}\right)-\bar{\mu}_{k}^{(i)}\right) > \frac{\left(b+b_{n}\right)}{2} - \frac{\left(\bar{\mu}^{(i)}_{k}-\bar{\mu}^{(i-1)}_{k}\right)w}{4} \middle| \overline{F}_{i}\overline{F}_{i-1} D_{i-1}\right)\label{eqn:sd_bound_f7}\\
        &\hspace{20pt} + \mathbb{P}\left(\bar{\mu}^{(i)}_{k}<\bar{\mu}^{(i-1)}_{k}\right)\mathbb{P}\left(\sum^{w}_{\ell=w/2+1}\left(f_{k,\ell}\left(V_{k,\ell}\right)-\bar{\mu}_{k}^{(i)}\right)< -\frac{\left(\bar{\mu}^{(i-1)}_{k}-\bar{\mu}^{(i)}_{k}\right)w}{4}-\frac{\left(b+b_{n}\right)}{2} \middle| \overline{F}_{i}\overline{F}_{i-1} D_{i-1}\right)\label{eqn:sd_bound_f8}\\
        &\hspace{20pt} \cdot \mathbb{P}\left(\sum^{w/2}_{\ell=1}\left(f_{k,\ell}\left(V_{k,\ell}\right)-\bar{\mu}_{k}^{(i-1)}\right) > \frac{\left(b+b_{n}\right)}{2} + \frac{\left(\bar{\mu}^{(i-1)}_{k}-\bar{\mu}^{(i)}_{k}\right)w}{4} \middle| \overline{F}_{i}\overline{F}_{i-1} D_{i-1}\right)\label{eqn:sd_bound_f9}\\
        &\geq \mathbb{P}\left(\bar{\mu}^{(i)}_{k}>\bar{\mu}^{(i-1)}_{k}\right)\mathbb{P}\left(\sum^{w/2}_{\ell=1}\left(f_{k,\ell}\left(V_{k,\ell}\right)-\bar{\mu}_{k}^{(i-1)}\right)< \frac{\delta w}{4}-\frac{\left(b+b_{n}\right)}{2} \middle| \overline{F}_{i}\overline{F}_{i-1} D_{i-1}\right)\label{eqn:sd_bound_f10}\\
        &\hspace{20pt} \cdot \mathbb{P}\left(\sum^{w}_{\ell=w/2+1}\left(f_{k,\ell}\left(V_{k,\ell}\right)-\bar{\mu}_{k}^{(i)}\right) > \frac{\left(b+b_{n}\right)}{2} - \frac{\delta w}{4} \middle| \overline{F}_{i}\overline{F}_{i-1} D_{i-1}\right)\label{eqn:sd_bound_f11}\\
        &\hspace{20pt} + \mathbb{P}\left(\bar{\mu}^{(i)}_{k}<\bar{\mu}^{(i-1)}_{k}\right)\mathbb{P}\left(\sum^{w}_{\ell=w/2+1}\left(f_{k,\ell}\left(V_{k,\ell}\right)-\bar{\mu}_{k}^{(i)}\right)< -\frac{\delta w}{4}-\frac{\left(b+b_{n}\right)}{2} \middle| \overline{F}_{i}\overline{F}_{i-1} D_{i-1}\right)\label{eqn:sd_bound_f12}\\
        &\hspace{20pt} \cdot \mathbb{P}\left(\sum^{w/2}_{\ell=1}\left(f_{k,\ell}\left(V_{k,\ell}\right)-\bar{\mu}_{k}^{(i-1)}\right) > \frac{\left(b+b_{n}\right)}{2} + \frac{\delta w}{4} \middle| \overline{F}_{i}\overline{F}_{i-1} D_{i-1}\right)\label{eqn:sd_bound_f13}
    \end{align}
\end{subequations}

Here, \eqref{eqn:sd_bound_f2} and \eqref{eqn:sd_bound_f3} follow from decomposing the absolute value event into two separate one-sided probability terms;
\eqref{eqn:sd_bound_f4} and \eqref{eqn:sd_bound_f5} follow from rearranging the terms and adding and subtracting the respective means;
\eqref{eqn:sd_bound_f6} and \eqref{eqn:sd_bound_f8} follow from conditioning on whether the mean reward has increased or decreased;
\eqref{eqn:sd_bound_f7} and \eqref{eqn:sd_bound_f9} follow from rearranging the terms;
and \eqref{eqn:sd_bound_f10} and \eqref{eqn:sd_bound_f12} follow from the definition of the minimum change size $\delta$.

Let $\zeta_{k}^{(i)}=\frac{\delta/2-\left(b+b_n\right)/w}{\bar{\mu}^{(i-1)}_{k}}$ and $\zeta_{k}^{(i-1)}=\frac{\delta/2-\left(b+b_n\right)/w}{\bar{\mu}^{(i-1)}_{k}}$. Then, \eqref{eqn:sd_bound_f10}--\eqref{eqn:sd_bound_f13} become
\begin{subequations}
    \begin{align} 
        &\mathbb{P}\left(\left\lvert\sum^{w}_{\ell=w/2+1}f_{k,\ell}\left(V_{k,\ell}\right)-\sum^{w/2}_{\ell=1}f_{k,\ell}\left(V_{k,\ell}\right)\right\rvert > b + b_d \middle| \overline{F}_{i}\overline{F}_{i-1} D_{i-1}\right)\label{eqn:sd_bound_f14}\\
        &\geq \mathbb{P}\left(\bar{\mu}^{(i)}_{k}>\bar{\mu}^{(i-1)}_{k}\right)\mathbb{P}\left(\sum^{w/2}_{\ell=1}\left(f_{k,\ell}\left(V_{k,\ell}\right)-\bar{\mu}_{k}^{(i-1)}\right)< \frac{1}{2}\left(\frac{\delta/2-\left(b+b_n\right)/w}{\bar{\mu}^{(i-1)}_{k}}\right)\bar{\mu}^{(i-1)}_{k}w \middle| \overline{F}_{i}\overline{F}_{i-1} D_{i-1}\right)\label{eqn:sd_bound_f15}\\
        &\hspace{20pt} \cdot \mathbb{P}\left(\sum^{w}_{\ell=w/2+1}\left(f_{k,\ell}\left(V_{k,\ell}\right)-\bar{\mu}_{k}^{(i)}\right) > -\frac{1}{2}\left(\frac{\delta/2-\left(b+b_n\right)/w}{\bar{\mu}^{(i)}_{k}}\right)\bar{\mu}^{(i)}_{k}w \middle| \overline{F}_{i}\overline{F}_{i-1} D_{i-1}\right)\label{eqn:sd_bound_f16}\\
        &\hspace{20pt} + \mathbb{P}\left(\bar{\mu}^{(i)}_{k}<\bar{\mu}^{(i-1)}_{k}\right)\mathbb{P}\left(\sum^{w}_{\ell=w/2+1}\left(f_{k,\ell}\left(V_{k,\ell}\right)-\bar{\mu}_{k}^{(i)}\right)< \frac{1}{2}\left(\frac{\delta/2-\left(b+b_n\right)/w}{\bar{\mu}^{(i)}_{k}}\right)\bar{\mu}^{(i)}_{k}w \middle| \overline{F}_{i}\overline{F}_{i-1} D_{i-1}\right)\label{eqn:sd_bound_f17}\\
        &\hspace{20pt} \cdot \mathbb{P}\left(\sum^{w/2}_{\ell=1}\left(f_{k,\ell}\left(V_{k,\ell}\right)-\bar{\mu}_{k}^{(i-1)}\right) > -\frac{1}{2}\left(\frac{\delta/2-\left(b+b_n\right)/w}{\bar{\mu}^{(i-1)}_{k}}\right)\bar{\mu}^{(i-1)}_{k}w \middle| \overline{F}_{i}\overline{F}_{i-1} D_{i-1}\right)\label{eqn:sd_bound_f18}\\
        &\geq \mathbb{P}\left(\bar{\mu}^{(i)}_{k}>\bar{\mu}^{(i-1)}_{k}\right)\mathbb{P}\left(\sum^{w/2}_{\ell=1}\left(f_{k,\ell}\left(V_{k,\ell}\right)-\bar{\mu}_{k}^{(i-1)}\right)< \frac{\zeta_{k}^{(i-1)}\bar{\mu}^{(i-1)}_{k}w}{2} \middle| \overline{F}_{i}\overline{F}_{i-1} D_{i-1}\right)\label{eqn:sd_bound_f16}\\
        &\hspace{20pt} \cdot \mathbb{P}\left(\sum^{w}_{\ell=w/2+1}\left(f_{k,\ell}\left(V_{k,\ell}\right)-\bar{\mu}_{k}^{(i)}\right) > -\frac{\zeta_{k}^{(i)}\bar{\mu}^{(i)}_{k}w}{2} \middle| \overline{F}_{i}\overline{F}_{i-1} D_{i-1}\right)\label{eqn:sd_bound_f17}\\
        &\hspace{20pt} + \mathbb{P}\left(\bar{\mu}^{(i)}_{k}<\bar{\mu}^{(i-1)}_{k}\right)\mathbb{P}\left(\sum^{w}_{\ell=w/2+1}\left(f_{k,\ell}\left(V_{k,\ell}\right)-\bar{\mu}_{k}^{(i)}\right)< \frac{\zeta_{k}^{(i)}\bar{\mu}^{(i)}_{k}w}{2} \middle| \overline{F}_{i}\overline{F}_{i-1} D_{i-1}\right)\label{eqn:sd_bound_f18}\\
        &\hspace{20pt} \cdot \mathbb{P}\left(\sum^{w/2}_{\ell=1}\left(f_{k,\ell}\left(V_{k,\ell}\right)-\bar{\mu}_{k}^{(i-1)}\right) > -\frac{\zeta_{k}^{(i-1)}\bar{\mu}^{(i-1)}_{k}w}{2} \middle| \overline{F}_{i}\overline{F}_{i-1} D_{i-1}\right)\label{eqn:sd_bound_f19}
    \end{align}
\end{subequations}
Here, \eqref{eqn:sd_bound_f14}--\eqref{eqn:sd_bound_f17} follow from substituting the definitions of $\zeta_{k}^{(i)}$ and $\zeta_{k}^{(i-1)}$; and \eqref{eqn:sd_bound_f18} and \eqref{eqn:sd_bound_f19} follow from rearranging the terms.
Before applying Lemma~\ref{applem:Hoeffding bound for Markov Chain}, we rewrite the above as
\begin{subequations}
    \begin{align}
        &\mathbb{P}\left(\left\lvert\sum^{w}_{\ell=w/2+1}f_{k,\ell}\left(V_{k,\ell}\right)-\sum^{w/2}_{\ell=1}f_{k,\ell}\left(V_{k,\ell}\right)\right\rvert > b + b_d \middle| \overline{F}_{i}\overline{F}_{i-1} D_{i-1}\right)\label{eqn:sd_bound_f20}\\
        &\geq \mathbb{P}\left(\bar{\mu}^{(i)}_{k}>\bar{\mu}^{(i-1)}_{k}\right)\left(1-\mathbb{P}\left(\sum^{w/2}_{\ell=1}\left(f_{k,\ell}\left(V_{k,\ell}\right)-\bar{\mu}_{k}^{(i-1)}\right)\geq \frac{\left(1+\zeta_{k}^{(i-1)}\right)\bar{\mu}^{(i-1)}_{k}w}{2} \middle| \overline{F}_{i}\overline{F}_{i-1} D_{i-1}\right)\right)\label{eqn:sd_bound_f21}\\
        &\hspace{20pt} \cdot \left(1-\mathbb{P}\left(\sum^{w}_{\ell=w/2+1}\left(f_{k,\ell}\left(V_{k,\ell}\right)-\bar{\mu}_{k}^{(i)}\right)\leq \frac{\left(1-\zeta_{k}^{(i)}\right)\bar{\mu}^{(i)}_{k}w}{2} \middle| \overline{F}_{i}\overline{F}_{i-1} D_{i-1}\right)\right)\label{eqn:sd_bound_f22}\\
        &\hspace{20pt} + \mathbb{P}\left(\bar{\mu}^{(i)}_{k}<\bar{\mu}^{(i-1)}_{k}\right)\left(1-\mathbb{P}\left(\sum^{w}_{\ell=w/2+1}\left(f_{k,\ell}\left(V_{k,\ell}\right)-\bar{\mu}_{k}^{(i)}\right)\geq \frac{\left(1+\zeta_{k}^{(i)}\right)\bar{\mu}^{(i)}_{k}w}{2} \middle| \overline{F}_{i}\overline{F}_{i-1} D_{i-1}\right)\right)\label{eqn:sd_bound_f23}\\
        &\hspace{20pt} \cdot \left(1-\mathbb{P}\left(\sum^{w/2}_{\ell=1}\left(f_{k,\ell}\left(V_{k,\ell}\right)-\bar{\mu}_{k}^{(i-1)}\right)\leq \frac{\left(1-\zeta_{k}^{(i-1)}\right)\bar{\mu}^{(i-1)}_{k}w}{2} \middle| \overline{F}_{i}\overline{F}_{i-1} D_{i-1}\right)\right)\label{eqn:sd_bound_f24}
    \end{align}
\end{subequations}

Here, \eqref{eqn:sd_bound_f20}--\eqref{eqn:sd_bound_f23} follow from rearranging the terms.
Now, we can apply Lemma~\ref{applem:Hoeffding bound for Markov Chain} to bound the probability terms in \eqref{eqn:sd_bound_f21}--\eqref{eqn:sd_bound_f24}. For example, for the term in \eqref{eqn:sd_bound_f21}, we have
\begin{subequations}
    \begin{align}
        &\mathbb{P}\left(\left\lvert\sum^{w}_{\ell=w/2+1}f_{k,\ell}\left(V_{k,\ell}\right)-\sum^{w/2}_{\ell=1}f_{k,\ell}\left(V_{k,\ell}\right)\right\rvert > b + b_d \middle| \overline{F}_{i}\overline{F}_{i-1} D_{i-1}\right)\label{eqn:sd_bound_f25}\\
        &\geq \mathbb{P}\left(\bar{\mu}^{(i)}_{k}>\bar{\mu}^{(i-1)}_{k}\right)\left(1-\epsilon \lVert \varphi \rVert_{d}\exp\left(-\frac{\left(c/2-b_{n}/w\right)^{2}}{144L}\right)\right)^2\label{eqn:sd_bound_f26}\\
        &\hspace{20pt} + \mathbb{P}\left(\bar{\mu}^{(i)}_{k}<\bar{\mu}^{(i-1)}_{k}\right)\left(1-\epsilon \lVert \varphi \rVert_{d}\exp\left(-\frac{\left(c/2-b_{n}/w\right)^{2}}{144L}\right)\right)^2\label{eqn:sd_bound_f27}\\
        &= \left(1-\epsilon \lVert \varphi \rVert_{d}\exp\left(-\frac{\left(c/2-b_{n}/w\right)^{2}}{144L}\right)\right)^2\label{eqn:sd_bound_f28}\\
    \end{align}
\end{subequations}
Finally, to satisfy the inequality $\delta^{(i)}_{k} \ge 2b/w + c$, we set $w$ and $b$ as in \eqref{eqn:w_gen} and \eqref{eqn:b_gen}, and choose $c = 2\left(\sqrt{144L\ln T/w} + \sqrt{\ln(2KT^2)/2w}\right)$, which leads to 
\begin{equation}
    \mathbb{P}\left(\left\lvert\sum^{w}_{\ell=w/2+1}f_{k,\ell}\left(V_{k,\ell}\right)-\sum^{w/2}_{\ell=1}f_{k,\ell}\left(V_{k,\ell}\right)\right\rvert > b + b_d \middle| \overline{F}_{i}\overline{F}_{i-1} D_{i-1}\right)\geq 1 - \frac{\epsilon \lVert \varphi \rVert_{d}}{T}.\label{eqn:sd_final}
\end{equation}
The second probability term in \eqref{eqn:sd_4} can be bounded similarly as in \eqref{eqn:n_bound}, the successful detection probability is then obtained by substituting \eqref{eqn:sd_final} and \eqref{eqn:n_bound} back into \eqref{eqn:sd_4}.
\begin{equation}
    \mathbb{P}\left(D_i \middle| \overline{F}_{i}\overline{F}_{i-1} D_{i-1}\right) \geq 1 - \frac{\epsilon \lVert \varphi \rVert_{d}}{T}.\label{eqn:sd_final_2}
\end{equation}

\end{proof}


Lemma~\ref{applemma:gen_exp_delay} further bounds the expected detection delay in the situation where the change detection algorithm successfully detects the change within the desired interval.

\begin{lemma}[Expected detection delay]\label{applemma:gen_exp_delay} 
Consider a piecewise-stationary restless bandit environment. For any $\boldsymbol{\bar{\mu}}^{(i)},\boldsymbol{\bar{\mu}}^{(i+1)}\in\left[0,1\right]^{K}$ with parameters chosen in \eqref{eqn:w} and \eqref{eqn:b} 
and
\begin{equation}
h_{i} = \left\lceil w\left(\frac{K}{2\alpha}+1\right)\sqrt{s_{i}+1}+\frac{w^{2}}{4}\left(\frac{K}{2\alpha}+1\right)^{2} \right\rceil,
\end{equation}
for some $k\in\gK, i\geq 1$ and $c>0$, under the \Algref{alg:main_alg}, we have
\begin{equation}
    \E\left[\tau_{i}-\nu_{i}\middle| \overline{F}_{i}D_{i}\overline{F}_{i-1}D_{i-1}\right]\leq h_{i}.
\end{equation}
\end{lemma}

\begin{proof}
    For any $1\leq i\leq M$, we have
    \begin{subequations}
        \begin{align}
            \E\left[\tau_{i}-\nu_{i}\middle| \overline{F}_{i}D_{i}\overline{F}_{i-1}D_{i-1}\right]  
            &=\sum^{h_{i}}_{j=1}\mathbb{P}\left(\tau_{i}\geq\nu_{i}+j\middle| \overline{F}_{i}D_{i}\overline{F}_{i-1}D_{i-1}\right)
            \leq h_{i}.
        \end{align}
    \end{subequations}
\end{proof}

Plugging the bounds in Lemmas ~\ref{applem:FA},~\ref{applem:SD} and~\ref{applemma:exp_delay} into Theorem~\ref{appthm:general_extra} shows the following regret bound in Corollary~\ref{appcor:gen_mucb}.

\begin{corollary}\label{appcor:gen_mucb}
Combining Algorithm~\ref{alg:main_alg} (diminishing exploration with resets) 
with the base solver $\mathcal{B}$ and Algorithm~\ref{alg:mucb_cd} with parameters $(w,b)$ given in Equations~\ref{eqn:w} and~\ref{eqn:b}, 
the excess regret is upper-bounded as
\begin{equation}
     \Rex(T)\leq 2\alpha\sqrt{MT}+w\left(\frac{K}{2\alpha}+1\right)\sqrt{M\left(T+M\right)}+\frac{w^{2}M}{4}\left(\frac{K}{2\alpha}+1\right)^{2}+2\epsilon \lVert \varphi \rVert _{d}+1.
\end{equation}\label{eqn:regret_bound_MUCB_p}By setting $\alpha = c\sqrt{KL\log{\left(KT\right)}}$ for some constant $c$, the expected regret is upper-bounded by $\mathcal{O}(\sqrt{KLMT\log{T}})$.
\end{corollary}

%% file: 3-parameters_RMAB.tex
\section{Algorithms and Parameters Tuning}
\label{app:para}
In this appendix, we provide an explanation of our parameter selection.

\textbf{Algorithm Setting.} For change detection mechanism, the window size $w$ is set to $\left\lfloor100\sqrt{144L}\right\rfloor$, where $L$ is the mixing time for $\epsilon=1/8$, the threshold $b = \sqrt{w/2\ln{(2KT^2)}} + \sqrt{144wL\ln{(2KT^2)}}$. For uniform sampling scheme, the sample rate $\alpha = \sqrt{M/T\log{(T/M)}}$ is following the setting in~\cite{liu2018change}. For diminishing exploration scheme, the sampling parameter $\alpha = 1$. In RestlessUCB algorithm, the minimum sampling amount is set as $m(T)=100$, and the confidence radius $rad(T)=\sqrt{T/(2Mm(T))}$. For the Colored-UCRL2 algorithm, $\delta=1/\min(10000,T)$ and $B=S$, following the work in~\cite{ortner2012regret}. To improve computational efficiency, we restrict the value of $\delta$ so that the value iteration can run faster without significant loss of accuracy. 

\textbf{Environment Setting.} We consider an experiment consisting of three Markov chains (arms), indexed by $k \in \{1,2,3\}$, each evolving over five segments, indexed by $i \in \{1,2,3,4,5\}$. The transitions between states in each Markov chain are governed by transition probability matrices, and each state is associated with a reward mean function.

\textbf{Markov Chain Transition Kernels. }
The transition kernel for each segment $i$ is given by a set of transition matrices $P^{(i,k)}$, where each row represents the probability distribution over the next states, given the current state.
For each segment $i$, we define:

\begin{equation}
P^{(i)} =
\begin{bmatrix}
P^{(i)}_1 & P^{(i)}_2 & P^{(i)}_3
\end{bmatrix}
\end{equation}

where $P^{(i)}_k$ represents the transition matrix for arm $k$.

The transition matrices for each segment are:

\begin{equation}
P^{(1)} =
\begin{bmatrix}
\begin{bmatrix} 0.50 & 0.50 & 0.00 \\ 0.17 & 0.66 & 0.17 \\ 0.00 & 0.50 & 0.50 \end{bmatrix}, &
\begin{bmatrix} 0.36 & 0.64 & 0.00 \\ 0.39 & 0.22 & 0.39 \\ 0.00 & 0.50 & 0.50 \end{bmatrix}, &
\begin{bmatrix} 0.55 & 0.45 & 0.00 \\ 0.43 & 0.36 & 0.21 \\ 0.00 & 0.50 & 0.50 \end{bmatrix}
\end{bmatrix}
\end{equation}

\begin{equation}
P^{(2)} =
\begin{bmatrix}
\begin{bmatrix} 0.43 & 0.57 & 0.00 \\ 0.45 & 0.30 & 0.25 \\ 0.00 & 0.64 & 0.36 \end{bmatrix}, &
\begin{bmatrix} 0.33 & 0.67 & 0.00 \\ 0.32 & 0.42 & 0.26 \\ 0.00 & 0.50 & 0.50 \end{bmatrix}, &
\begin{bmatrix} 0.50 & 0.50 & 0.00 \\ 0.11 & 0.47 & 0.42 \\ 0.00 & 0.67 & 0.33 \end{bmatrix}
\end{bmatrix}
\end{equation}

\begin{equation}
P^{(3)} =
\begin{bmatrix}
\begin{bmatrix} 0.62 & 0.38 & 0.00 \\ 0.50 & 0.44 & 0.06 \\ 0.00 & 0.62 & 0.38 \end{bmatrix}, &
\begin{bmatrix} 0.50 & 0.50 & 0.00 \\ 0.43 & 0.21 & 0.36 \\ 0.00 & 0.38 & 0.62 \end{bmatrix}, &
\begin{bmatrix} 0.17 & 0.83 & 0.00 \\ 0.50 & 0.28 & 0.22 \\ 0.00 & 0.44 & 0.56 \end{bmatrix}
\end{bmatrix}
\end{equation}

\begin{equation}
P^{(4)} =
\begin{bmatrix}
\begin{bmatrix} 0.57 & 0.43 & 0.00 \\ 0.20 & 0.47 & 0.33 \\ 0.00 & 0.57 & 0.43 \end{bmatrix}, &
\begin{bmatrix} 0.60 & 0.40 & 0.00 \\ 0.18 & 0.47 & 0.35 \\ 0.00 & 0.50 & 0.50 \end{bmatrix}, &
\begin{bmatrix} 0.53 & 0.47 & 0.00 \\ 0.20 & 0.20 & 0.60 \\ 0.00 & 0.73 & 0.27 \end{bmatrix}
\end{bmatrix}
\end{equation}

\begin{equation}
P^{(5)} =
\begin{bmatrix}
\begin{bmatrix} 0.38 & 0.62 & 0.00 \\ 0.31 & 0.13 & 0.56 \\ 0.00 & 0.56 & 0.44 \end{bmatrix}, &
\begin{bmatrix} 0.55 & 0.45 & 0.00 \\ 0.45 & 0.46 & 0.09 \\ 0.00 & 0.64 & 0.36 \end{bmatrix}, &
\begin{bmatrix} 0.33 & 0.67 & 0.00 \\ 0.56 & 0.25 & 0.19 \\ 0.00 & 0.40 & 0.60 \end{bmatrix}
\end{bmatrix}
\end{equation}

\textbf{Reward Mean Functions.} The rewards are randomly generated based on the current state in each segment, denoted as $R^{(i)}_k(s)$, where $k=1,2,3$ represents the arm, $i=1,2,3,4,5$ represents the segment, and $s=1,2,3$ represents the state. 

\begin{equation}
R^{(i)} =
\begin{bmatrix}
R^{(i)}_1(1) & R^{(i)}_1(2) & R^{(i)}_1(3) \\
R^{(i)}_2(1) & R^{(i)}_2(2) & R^{(i)}_2(3) \\
R^{(i)}_3(1) & R^{(i)}_3(2) & R^{(i)}_3(3)
\end{bmatrix}
\end{equation}

The values for each segment are:

\begin{equation}
R^{(1)} =
\begin{bmatrix}
0.24090 & 0.21567 & 0.11303 \\
0.68377 & 0.55163 & 0.29024 \\
0.88837 & 0.87993 & 0.40863
\end{bmatrix} \quad
R^{(2)} =
\begin{bmatrix}
0.76023 & 0.42959 & 0.15431 \\
0.95580 & 0.81347 & 0.63100 \\
0.37325 & 0.24685 & 0.06446
\end{bmatrix}
\end{equation}

\begin{equation}
R^{(3)} =
\begin{bmatrix}
0.83859 & 0.80774 & 0.19539 \\
0.29037 & 0.23361 & 0.08248 \\
0.95163 & 0.45959 & 0.03668
\end{bmatrix} \quad
R^{(4)} =
\begin{bmatrix}
0.52273 & 0.14461 & 0.03640 \\
0.67324 & 0.54894 & 0.31872 \\
0.93209 & 0.88988 & 0.62226
\end{bmatrix}
\end{equation}

\begin{equation}
R^{(5)} =
\begin{bmatrix}
0.64763 & 0.56596 & 0.36023 \\
0.87041 & 0.82968 & 0.08825 \\
0.22632 & 0.21085 & 0.13084
\end{bmatrix}
\end{equation}


\textbf{Mixing Time Calculation.} The mixing time $L$ is calculated based on the transition matrices of the Markov chains. For each arm $k$ in segment $i$, we compute the second largest eigenvalue modulus (SLEM) of the transition matrix $P^{(i,k)}$. The mixing time is then determined using the formula:
\begin{equation}
L^{(i,k)} = \frac{\ln(1/\epsilon)}{1 - \lambda_2^{(i,k)}}
\end{equation}
where $\lambda_2^{(i,k)}$ is the SLEM of $P^{(i,k)}$ and $\epsilon$ is the desired accuracy level (set to $1/8$ in our experiments). The overall mixing time $L$ used in the algorithm is the maximum mixing time across all arms and segments:
\begin{equation}
L = \max_{i,k} L^{(i,k)}
\end{equation}

%% file: 4-one-state_simulation.tex
\section{One-State Special Case Simulation Results}
\label{app:sim}
In this appendix, we assess the effectiveness of the proposed diminishing exploration scheme across various dimensions, encompassing regret scaling in $M$, $K$, and $T$, regrets in synthetic environments, and regrets in a real-world scenario. In addition to evaluating M-UCB \citep{cao2019nearly} with our diminishing exploration, we also examine a variant of CUSUM-UCB \citep{liu2018change} that incorporates diminishing exploration with CUSUM-UCB, further highlighting the efficacy of the proposed exploration method. We will compare our approach with M-UCB, CUSUM-UCB, GLR-UCB, Discounted-UCB, Discounted Thompson Sampling, and MASTER. 
Unless stated otherwise, we report the average regrets over 100 simulation trials. Detailed configuration is provided in Appendix~\ref{app:para_one_state}.

{\bf Regret in Each Time Step.} In this simulation, we consider a multi-armed bandit problem with $T = 20000$ time stpdf and $M=5$. 
Recall that $\mu^{(i)}_{k}$ represents the expected value for arm $k$ in the $i$-th segment. Here, we set $\mu^{(i)}_{k}=0.2,0.5,0.8$ for $i$ with $(i+k)\bmod 3=2, 0, 1$, respectively.

Figure~\ref{fig:t} shows that for both CUSUM-UCB and M-UCB, employing diminishing exploration can effectively reduce the additional regret caused by constant exploration. 
Moreover, M-UCB with the proposed diminishing exploration achieves the lowest regret. In the figure, the change points are clearly evident by observing the breakpoints in each line. The reason for the overall steeper slope of CUSUM-UCB (both with and without diminishing exploration) is due to the heightened sensitivity of the CUSUM detector itself, resulting in more frequent false alarms. 

\begin{figure}[h]
\centering
\begin{subfigure}{0.85\textwidth}
    \centering
    \centering
    \includegraphics[width=0.7\linewidth]{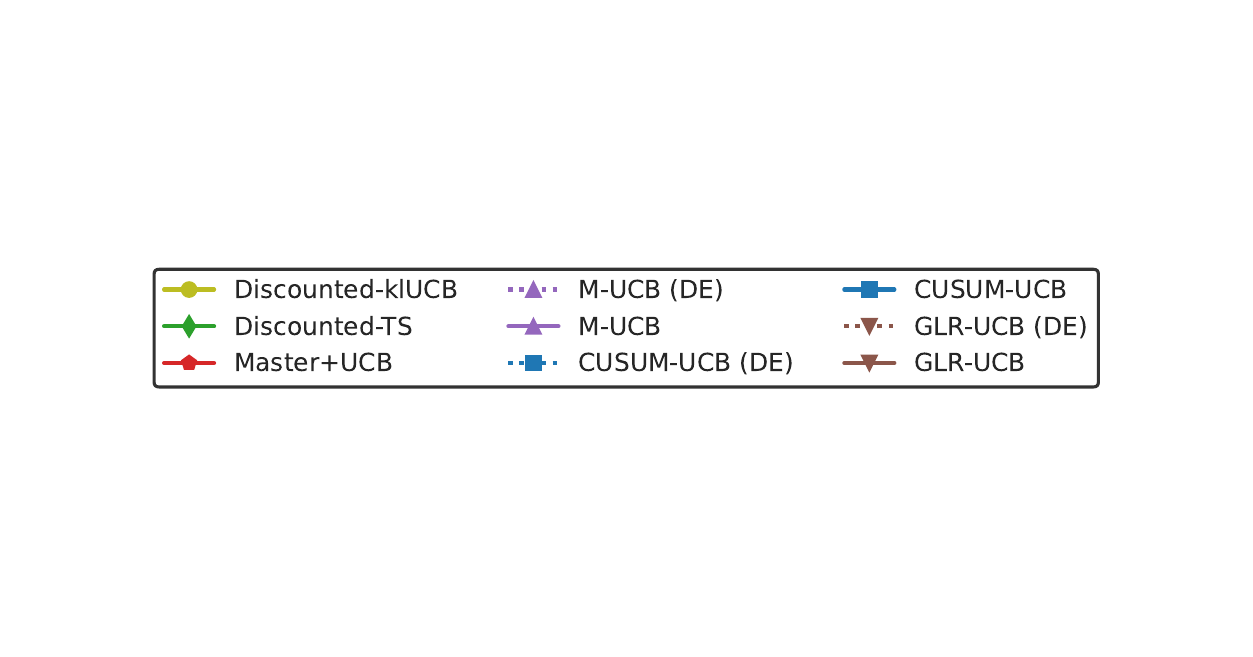}
    \vspace{-2cm}
\end{subfigure}
\hfill
\begin{subfigure}{0.25\textwidth}
    \includegraphics[width=\textwidth]{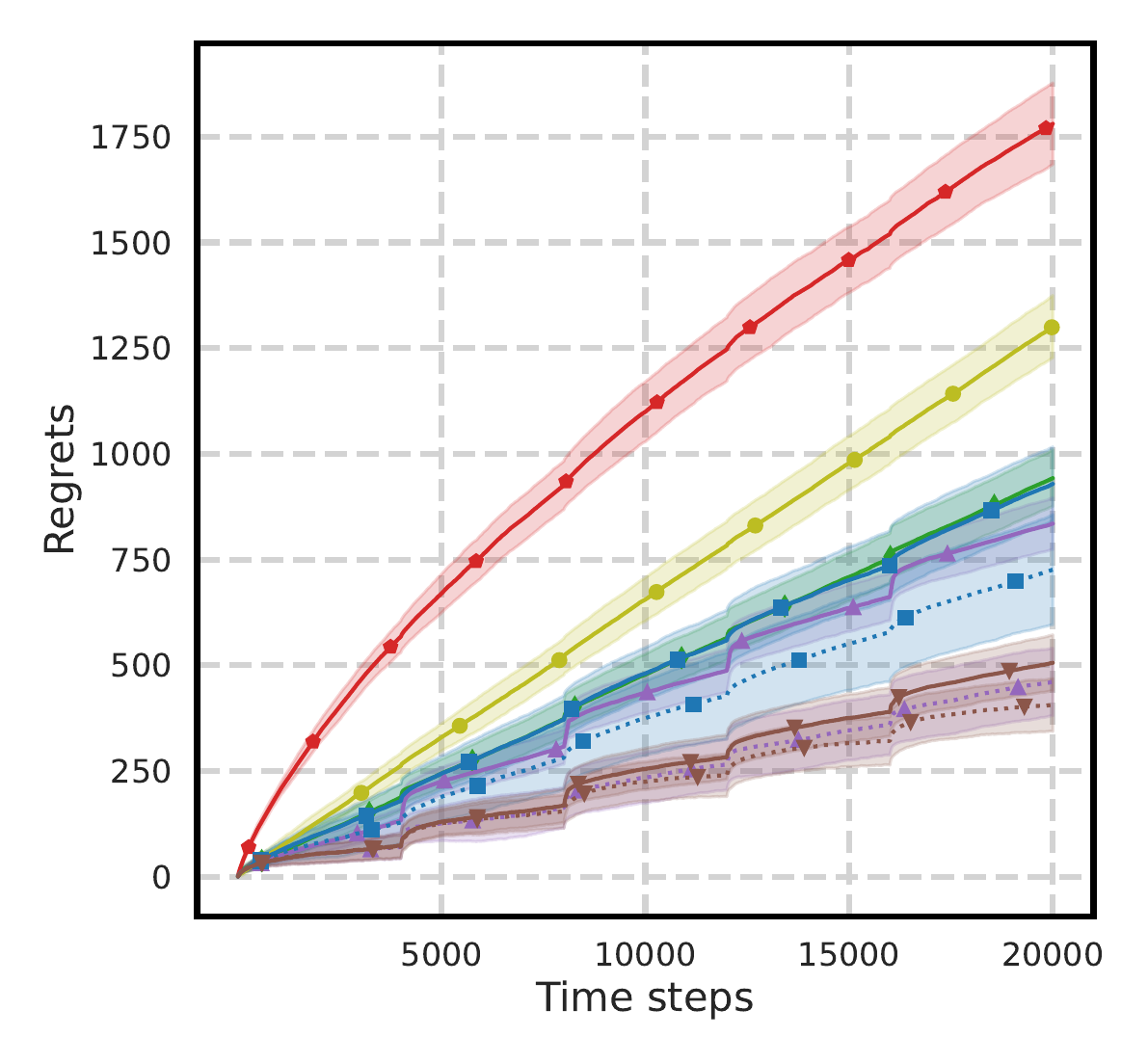}
    \caption{Scaling in $t$.}
    \label{fig:t}
\end{subfigure}
\hspace{-20pt}
\hfill
\begin{subfigure}{0.25\textwidth}
    \includegraphics[width=\textwidth]{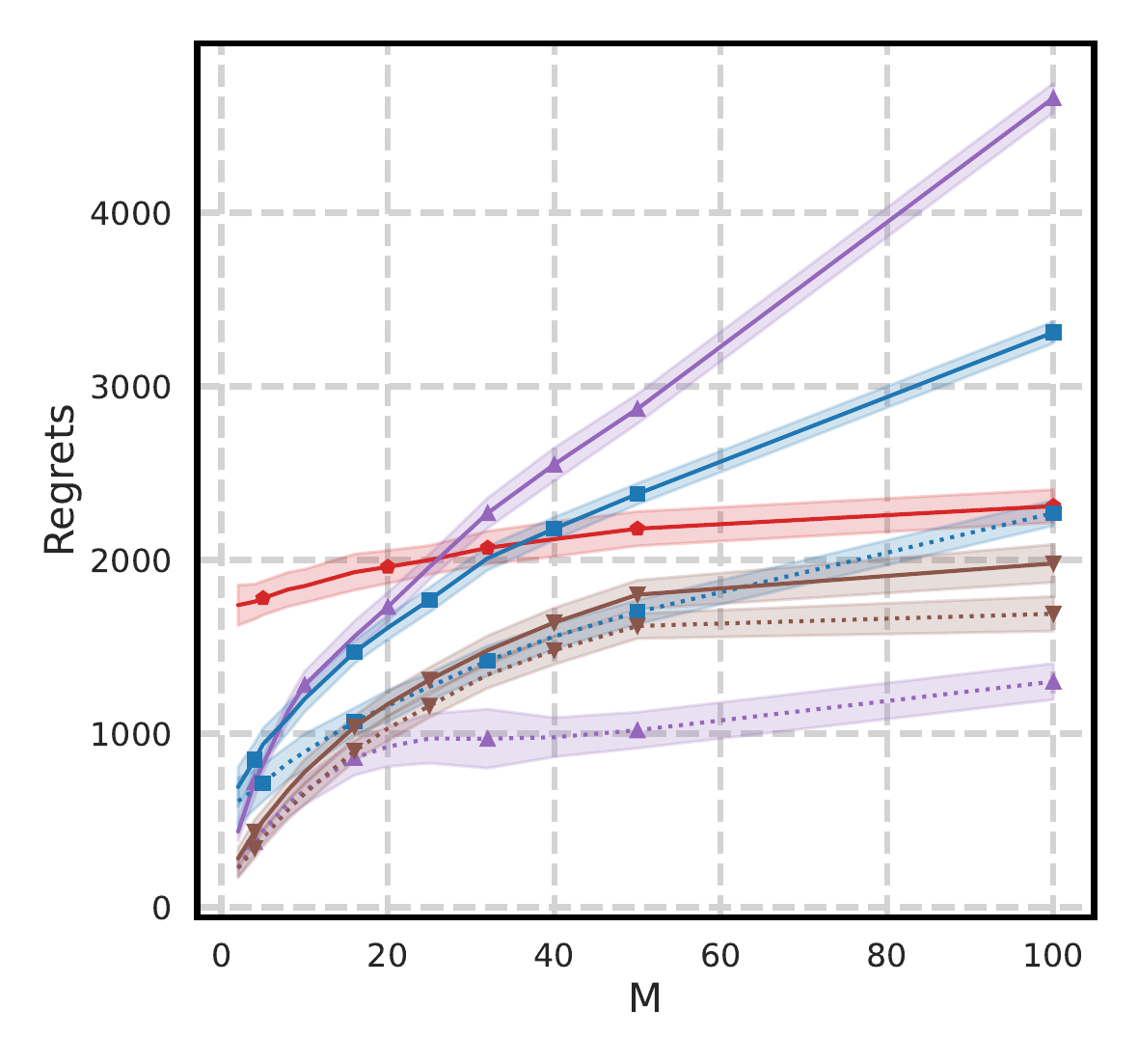}
    \caption{Scaling in $M$.}
    \label{fig:M}
\end{subfigure}
\hspace{-20pt}
\hfill
\begin{subfigure}{0.25\textwidth}
    \includegraphics[width=\textwidth]{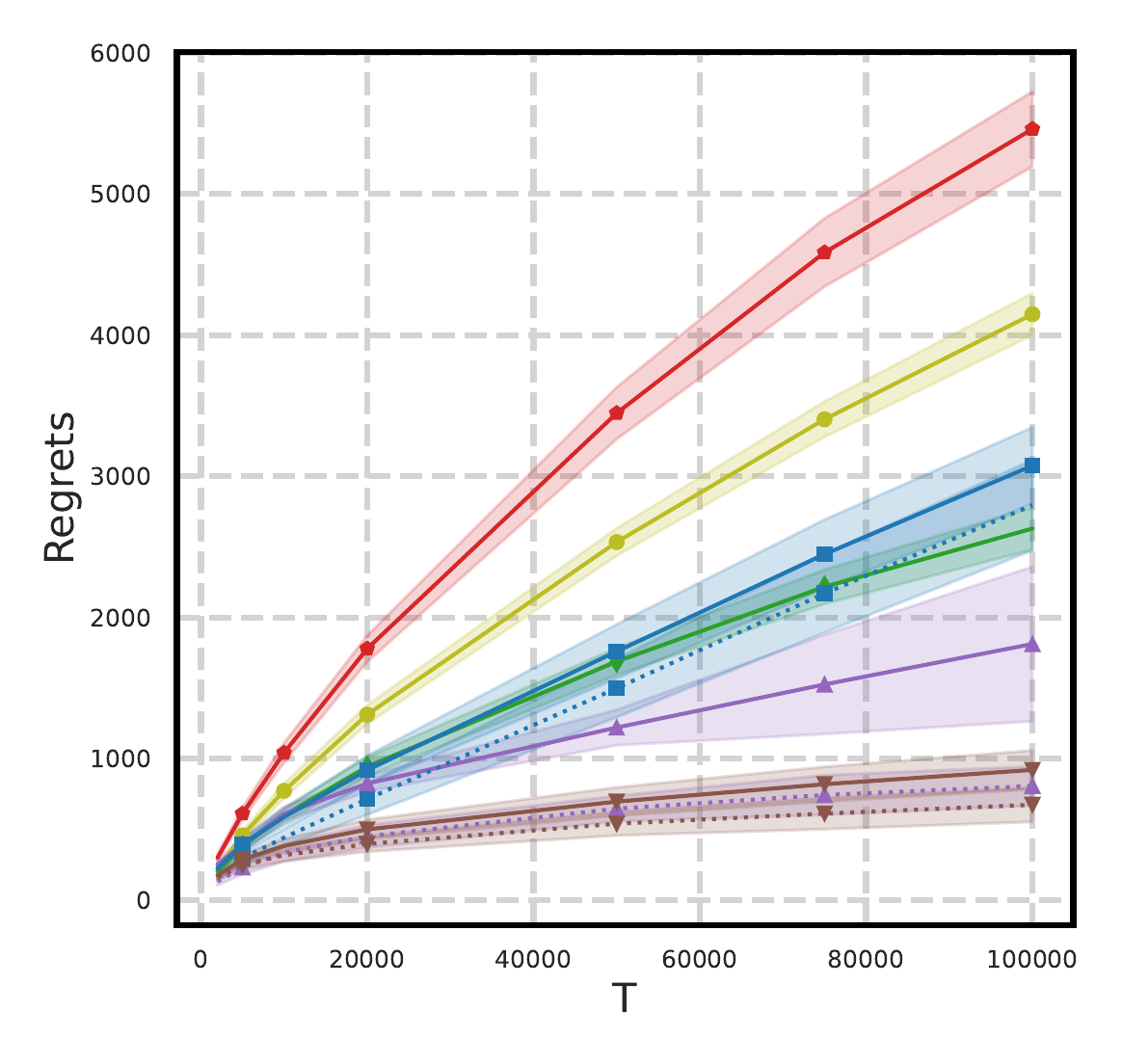}
    \caption{Scaling in $T$.}
    \label{fig:T}
\end{subfigure}
\hspace{-20pt}
\hfill
\begin{subfigure}{0.25\textwidth}
    \includegraphics[width=\textwidth]{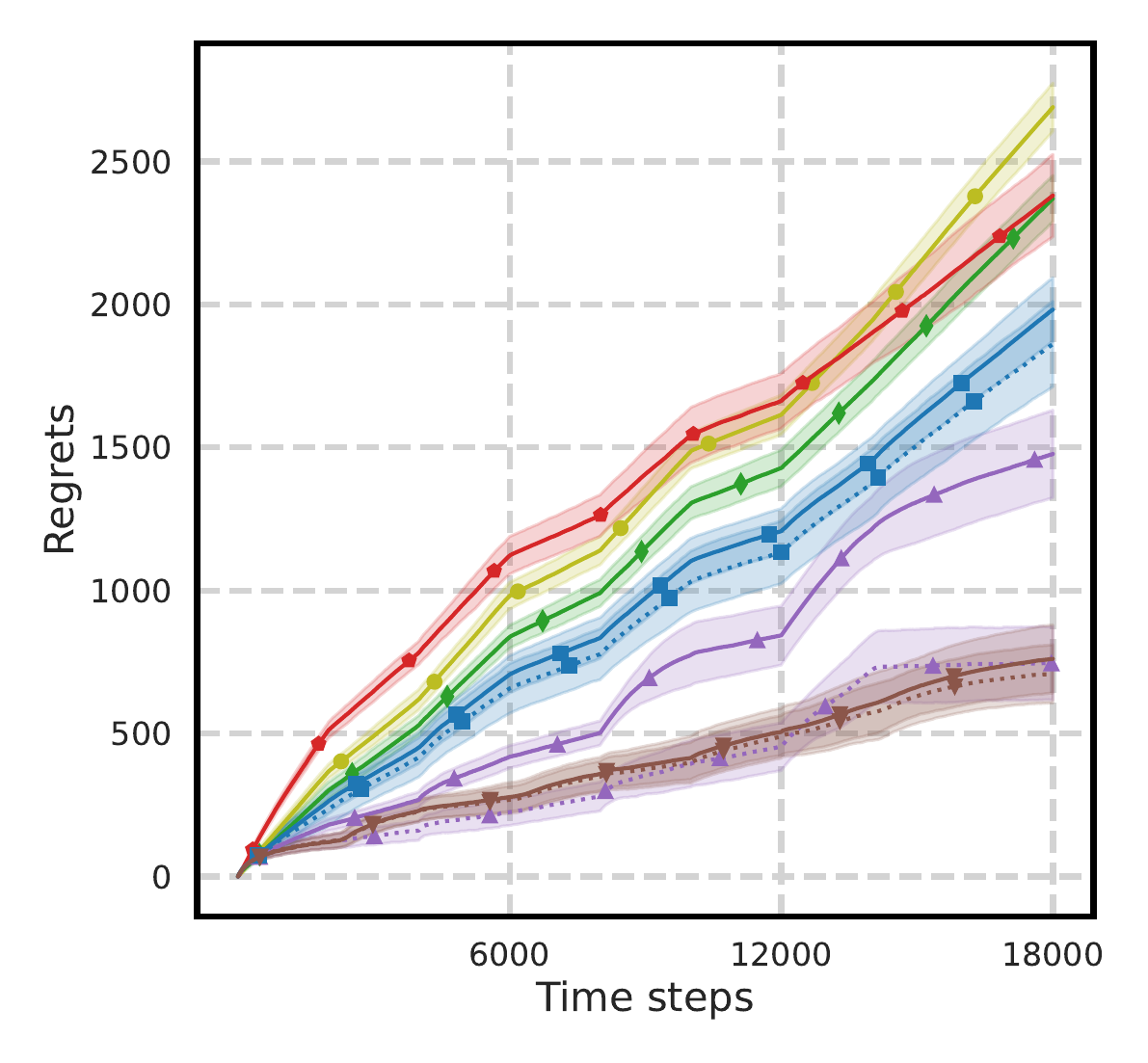}
    \caption{Yahoo Data Set.}
    \label{fig:Yahoo}
\end{subfigure}
\hspace{-20pt}
\vspace{+10pt}
\caption{Regret in the synthetic environments and under the Yahoo data set.}\label{fig:3}
\end{figure}

{\bf Regret Scaling in $M$.} Based on the settings outlined above, with the only variation being in the parameter $M$, Figure~\ref{fig:M} illustrates the dynamic regrets across various values of $M$. In this experiment, adjustments to the exploration parameter settings are required based on the size of $M$ when using a constant exploration rate. However, this is not the case for the proposed diminishing exploration. The result confirms the earlier-discussed rationale that the proposed diminishing exploration can automatically adapt to the environment, resulting in the best regret performance among the algorithms.

{\bf Regret Scaling in $T$.} In line with the setting described above, with the only variation being the parameter $T$, we present the dynamic regrets across different values of $T$. Observations similar to those made above can again be found in Figure~\ref{fig:T}, where the proposed diminishing exploration can effectively reduce the regret.

\textbf{Regret and Execution Time.} Here, we compare the computation time and regret across different algorithms for various choices of $M$ and $T$ with other parameters same as above. Specifically, we conduct experiments under three scenarioso: one where the environment changes rapidly ($M=50$ and $T=20000$), one where the environment changes slowly ($M=5$ and $T=20000$), and one where the considered time horizon is double ($M=5$ and $T=40000$). As shown in Figure~\ref{fig:ComM50} to \ref{fig:ComBigT}, despite oblivious of $M$, our algorithm almost always achieves the lowest computation time and regret in all the scenarios. Moreover, comparing to Master+UCB, another algorithm not requiring the knowledge of $M$, our algorithm is always significantly better as shown in these figures. Figure~\ref{fig:ratio} plots the ratio of average execution time of Master+UCB to that of M-UCB with our DE for various $T$. It is shown that the growth rate is faster than $0.5 \log T$, and it appears to become even linear in $T$ as $T$ increases.

\begin{figure}[h]
\centering
\begin{subfigure}{0.85\textwidth}
    \centering
    \centering
    \includegraphics[width=0.7\linewidth]{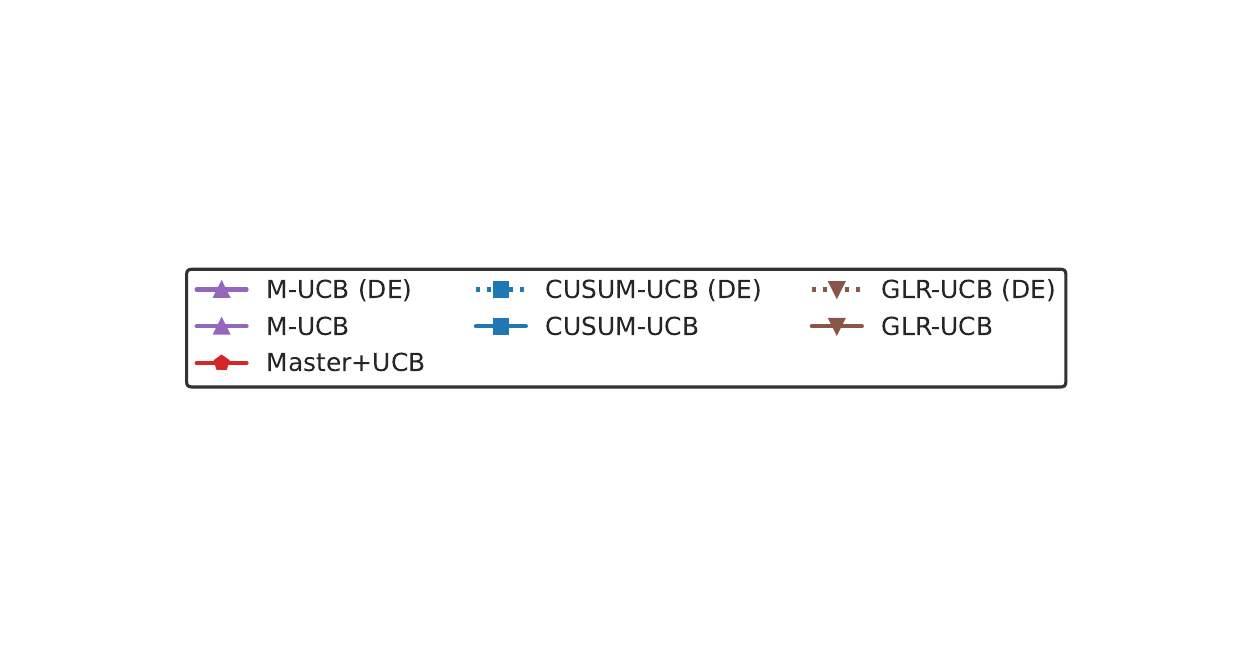}
    \vspace{-2cm}
\end{subfigure}
\hfill
\begin{subfigure}{0.25\textwidth}
    \includegraphics[width=\textwidth]{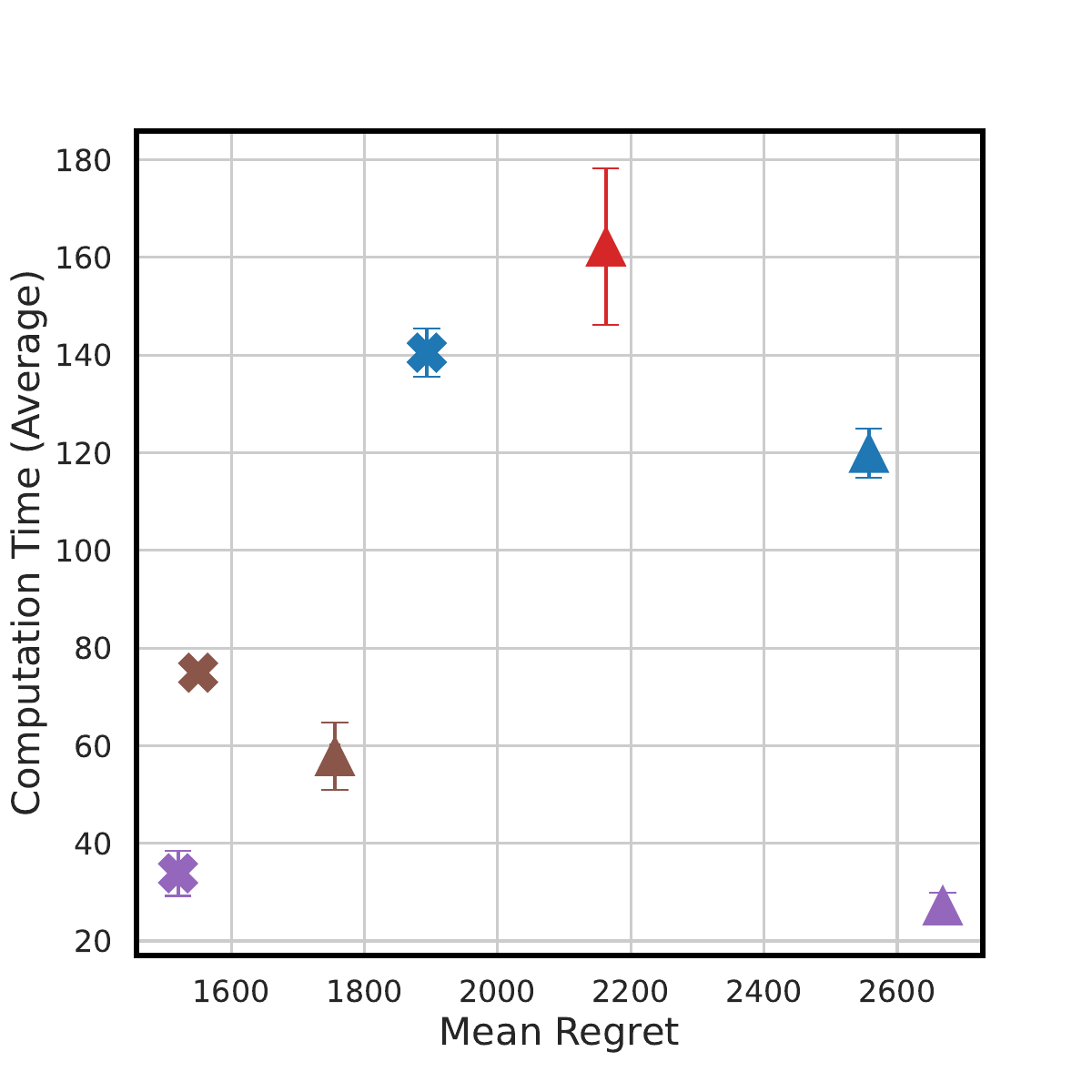}
    \caption{$M=50$, $T=20000$.}
    \label{fig:ComM50}
\end{subfigure}
\hspace{-20pt}
\hfill
\begin{subfigure}{0.25\textwidth}
    \includegraphics[width=\textwidth]{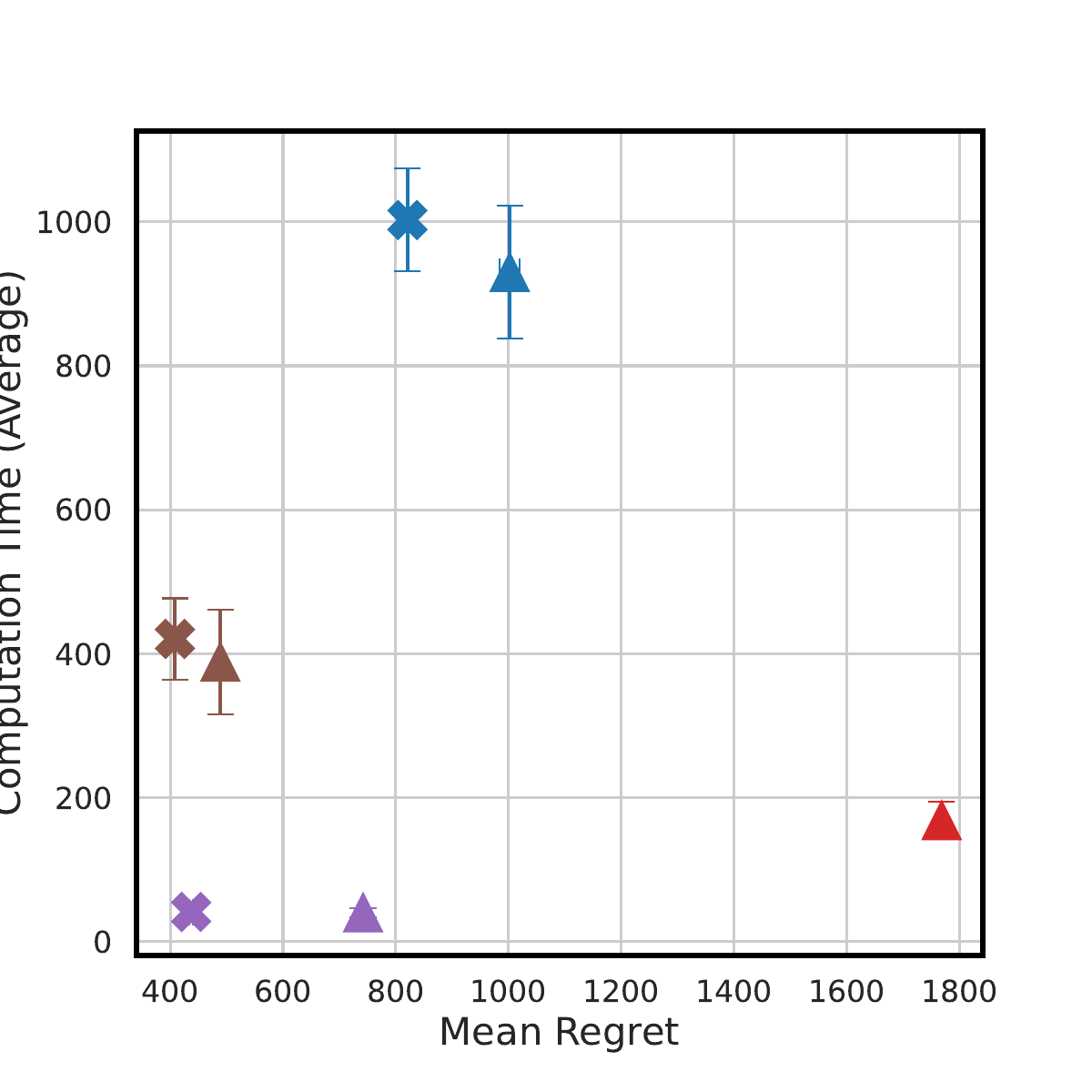}
    \caption{$M=5$, $T=20000$.}
    \label{fig:ComM5}
\end{subfigure}
\hspace{-20pt}
\hfill
\begin{subfigure}{0.25\textwidth}
    \includegraphics[width=\textwidth]{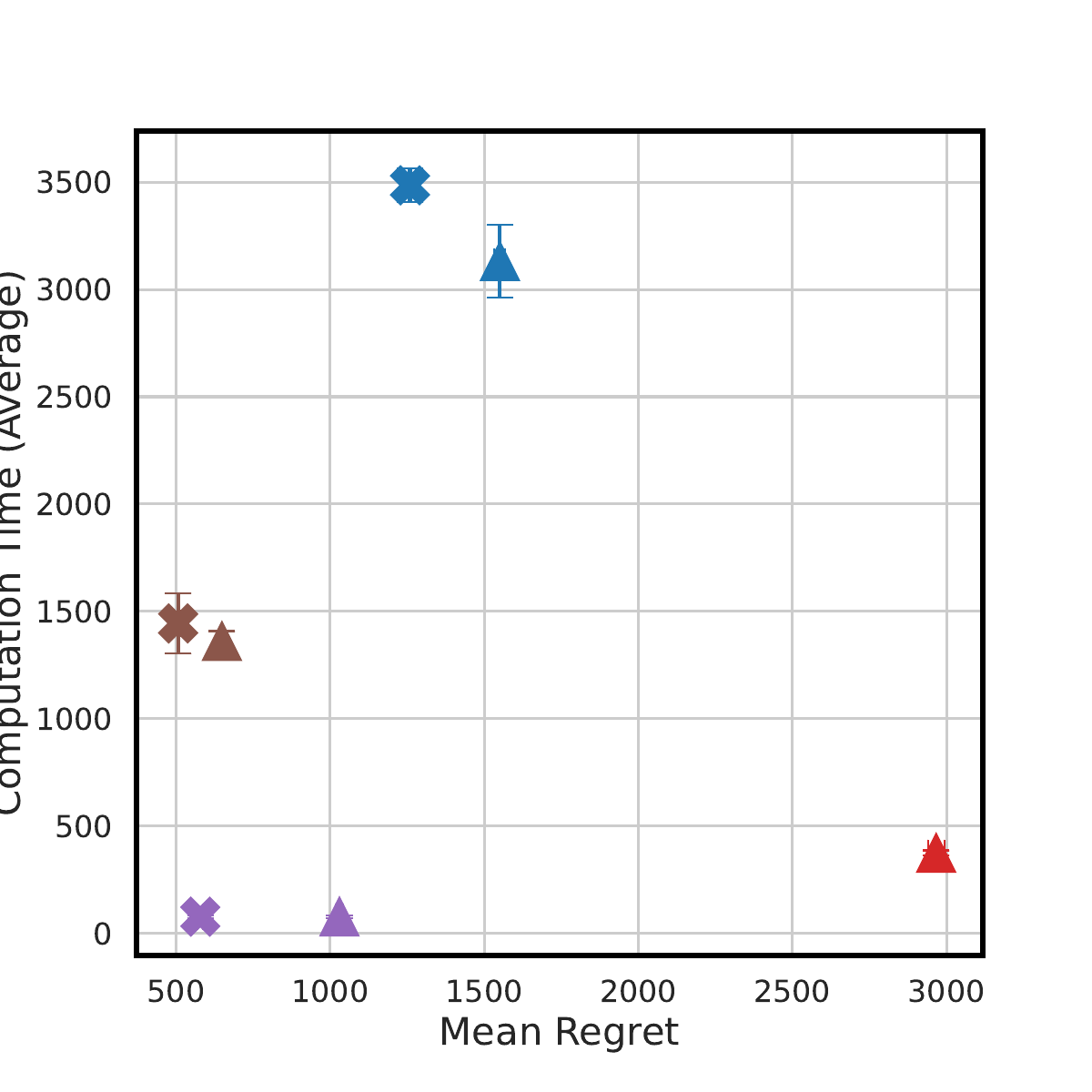}
    \caption{$M=5$, $T=40000$.}
    \label{fig:ComBigT}
\end{subfigure}
\hspace{-20pt}
\hfill
\begin{subfigure}{0.25\textwidth}
    \hspace{-20pt}
    \includegraphics[width=\textwidth]{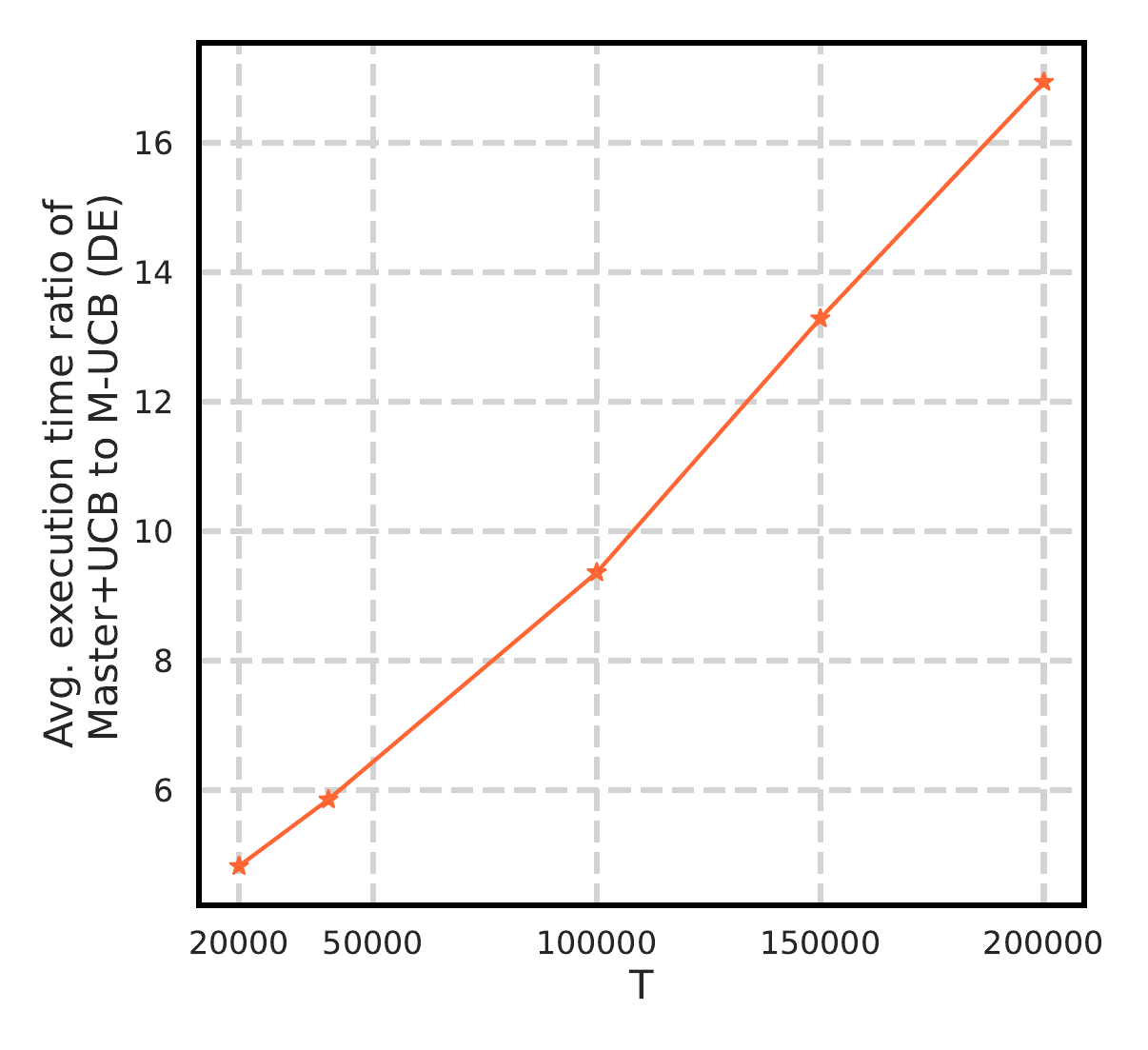}
    \caption{Ratio of avg. execution\\ time of Master+UCB to that\\ of M-UCB (DE).}
    \vspace{-20pt}
    \label{fig:ratio}
\end{subfigure}
\hspace{-20pt}
\vspace{+10pt}
\caption{Regret and computation times.}\label{fig:4}
\end{figure}
\textbf{Regret in an Environment Built from a Real-World Dataset.} We further utilize the benchmark dataset publicly published by Yahoo! for evaluation. To enhance arm distinguishability in our simulation, we scale up the data by a factor of $10$. The number of segments is set to $M=9$ and the number of arms is set to $K=6$. Figure \ref{fig:Yahoo} shows the evolution of dynamic regret. Again, we see that the diminishing exploration scheme could help M-UCB, GLR-UCB and CUSUM-UCB achieve similar or better regret even without knowing $M$.

\textbf{Scenarios when Assumptions are Violated.} In Assumption~\ref{ass:minimum_gap}, we assumes the knowledge of $\delta$ to select an appropriate $w$. We emphasize that our settings in many of the above simulations actually violate this assumption. Take Figure~\ref{fig:t} for example, $w=200$ is chosen, which corresponds to $\delta\approx 0.6$ when back calculating, which is much larger than the actual $\delta=0.3$ of this scenario. Assumptions~\ref{asm:seg_length} provide guarantees that the segment length is sufficiently long. However, in the last data point of our Figure~\ref{fig:M} ($M=100$), these assumptions are clearly violated. In this case, it becomes challenging for the active methods to promptly detect every change point. For algorithms like M-UCB and CUSUM-UCB, which require knowledge of $M$, their awareness of quick changes causes them to invest more effort into change detection, leading to a very high exploration rate. However, this does not always guarantee successful detection, resulting in very high regret. Diminishing exploration, on the other hand, continues to decrease the exploration rate regardless of $M$ when no changes are detected. This allows more resources to be invested for UCB, which might gradually adapt to the environment's changes, leading to a regret that is lower than that of uniform exploration.

%% file: 5-parameter.tex
\newpage
\section{One-State Special Case Algorithms and Parameters Tuning}
\label{app:para_one_state}
In this appendix, we provide an explanation of our parameter selection. For M-UCB, the window size $w$ is set to 200 unless otherwise specified; however, for the last data point ($M=100$) in Figure~\ref{fig:M}, we chose $w=50$ due to the limitations inherent to change detection. We compute the change detection threshold $b_{\textrm{M-UCB}} = \sqrt{w/2 \log\left(2KT^{2}\right)}$ following the original formulation in \cite{cao2019nearly}. Additionally, the uniform exploration rate $\gamma_{\textrm{M-UCB}} = \sqrt{MK \log{T}/T}$ is determined as initially stated in \cite{besson2022efficient}.Concerning CUSUM-UCB, we adhere to \cite{liu2018change} by fixing $\epsilon = 0.1$, setting the change detection threshold $b_{\textrm{CUSUM-UCB}} = \log\left(T/M-1\right)$, and establishing the uniform exploration rate $\gamma_{\textrm{CUSUM-UCB}} = \sqrt{MK \log{T}/T}$ as initially stated in \cite{besson2022efficient}. Additionally, in CUSUM-UCB, the change point detection involves averaging the first $H$ samples, where $H$ is set to 100. For GLR-UCB \citep{besson2022efficient}, we set $\gamma_{m, \textrm{GLR-UCB}} = \sqrt{mK \log{T}/T}$, where $m$ is the number of alarms. We utilize the threshold function $\beta(n, \delta) = \log{\left(n^{3/2}/\delta\right)}$ and set $\delta = 1/\sqrt{T}$. In our setup, for both the diminishing versions of M-UCB and CUSUM-UCB, we follow the parameter selection approach described earlier, except for the choice of the exploration rate. In this context, we opt for $\alpha=1$. 
For passive methods, including DUCB \citep{garivier2011upper} and DTS \citep{qi2023discounted}, we use a discounting factor $\gamma = 0.75$. MASTER, on the other hand, follows the theoretical settings outlined in \cite{pmlr-v134-wei21b} and is categorized as an active method.

To evaluate the scalability of our methods, we conducted three sets of scaling experiments. For scaling in $t$ (Figure~\ref{fig:t}), scaling in $M$ (Figure~\ref{fig:M}), and scaling in $T$ (Figure~\ref{fig:T}), The parameters of DUCB, DTS and MASTER follow the theoretical settings outlined in \cite{garivier2011upper}, \cite{qi2023discounted} and \cite{pmlr-v134-wei21b}, respectively.

\begin{table}[h!]
\centering
\caption{Parameter Selection for Active and Passive Methods}
\resizebox{\linewidth}{!}{
\begin{tabular}{@{}lll@{}}
\toprule
\textbf{Method}        & \textbf{Parameters}                                                                                                                                                           & \textbf{References}                  \\ \midrule
\textbf{M-UCB}         & - Window size $w = 200$ (default), $w = 50$ for $M = 100$ (Figure~\ref{fig:M}). \\
                       & - Threshold: $b_{\textrm{M-UCB}} = \sqrt{w/2 \log\left(2KT^{2}\right)}$. \\
                       & - Exploration rate: $\gamma_{\textrm{M-UCB}} = \sqrt{MK \log{T}/T}$.                                                                                       & \cite{cao2019nearly},  \\ \midrule
\textbf{CUSUM-UCB}     & - Fixed $\epsilon = 0.1$. \\
                       & - Threshold: $b_{\textrm{CUSUM-UCB}} = \log\left(T/M-1\right)$. \\
                       & - Exploration rate: $\gamma_{\textrm{CUSUM-UCB}} = \sqrt{MK \log{T}/T}$. \\
                       & - Change detection based on averaging first $H = 100$ samples.                                                                       & \cite{liu2018change},  \\ \midrule
\textbf{GLR-UCB}       & - Exploration rate: $\gamma_{m, \textrm{GLR-UCB}} = \sqrt{mK \log{T}/T}$. \\
                       & - Threshold function: $\beta(n, \delta) = \log{\left(n^{3/2}/\delta\right)}$, $\delta = 1/\sqrt{T}$.                                  & \cite{besson2022efficient}          \\ \midrule
\textbf{MASTER}        & - All parameters follow theoretical settings.                                                                                                                                                       & \cite{pmlr-v134-wei21b}              \\ \midrule
\textbf{DUCB (Passive)} & - Discounting factor $\gamma = 0.75$.(Figure~\ref{fig:t} and~\ref{fig:Yahoo})\\
                        & - Discounting factor follows theoretical setting. (Figure~\ref{fig:M} and ~\ref{fig:T})& \cite{garivier2011upper}  \\ \midrule
\textbf{DTS (Passive)}  & - Discounting factor $\gamma = 0.75$. (Figure~\ref{fig:t} and~\ref{fig:Yahoo})\\
                        & - Discounting factor follows theoretical setting. (Figure~\ref{fig:M} and ~\ref{fig:T})  & \cite{qi2023discounted}             \\ \bottomrule
\end{tabular}}
\label{tab:parameter_setup}
\end{table}

\color{black}